\documentclass[onecolumn,DIV=8]{scrartcl}
\let\latextextsuperscript\textsuperscript
\AtBeginDocument{\let\textsuperscript\latextextsuperscript}

\usepackage[utf8]{inputenc} 
\usepackage[T1]{fontenc}
\usepackage[dvipsnames]{xcolor}
\definecolor{color1}{RGB}{230,57,70}
\definecolor{color2}{RGB}{29,53,87}
\definecolor{color3}{RGB}{69,123,157}
\usepackage[colorlinks=true,linkcolor=color1,urlcolor=color2,citecolor=color3,anchorcolor=green,pdfusetitle]{hyperref}

\usepackage{authblk}
\author[1]{James D. Watson}
\author[2]{Johannes Bausch}
\affil[1]{\small Department of Computer Science, University College London, UK}
\affil[2]{\small Department of Applied Mathematics and Theoretical Physics, University of Cambridge, UK}

\usepackage{amsfonts,amsthm,amsmath,mathtools}
\usepackage{dsfont,newtxtext,newtxmath,microtype}
\usepackage{booktabs,enumerate} 
\usepackage{graphicx}
\usepackage{subfig}
\graphicspath{{Images/}}

\newcommand{\mysymbol}[1]{\mbox{\raisebox{-0.3em}{\includegraphics[height=14pt]{images/#1}}}}

\newcommand{\midend}{\mysymbol{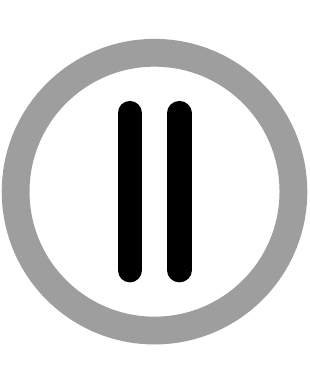}}

\input{braket.tex}
 
\usepackage[capitalise]{cleveref}
\Crefname{lemma}{Lemma}{Lemmas}
\Crefname{proposition}{Proposition}{Propositions}
\Crefname{definition}{Definition}{Definitions}
\Crefname{theorem}{Theorem}{Theorems}
\Crefname{conjecture}{Conjecture}{Conjectures}
\Crefname{corollary}{Corollary}{Corollaries}
\Crefname{example}{Example}{Examples}
\Crefname{section}{Section}{Sections}
\Crefname{appendix}{Appendix}{Appendices}
\Crefname{figure}{Fig.}{Figs.}
\Crefname{equation}{Eq.}{Eqs.}
\Crefname{table}{Table}{Tables}
\Crefname{item}{Property}{Properties}
\Crefname{remark}{Remark}{Remarks}

\newtheorem{theorem}{Theorem}
\newtheorem{definition}[theorem]{Definition}
\newtheorem{corollary}[theorem]{Corollary}

\newtheorem{lemma}[theorem]{Lemma}

\newtheorem{remark}[theorem]{Remark}

\usepackage{ltablex}

\newcommand\prob\textsc
\makeatletter
\newcommand{\probleminput}[1]{\gdef\@probleminput{#1}}
\newcommand{\problemquestion}[1]{\gdef\@problemquestion{#1}}
\newcommand{\problempromise}[1]{\gdef\@problempromise{#1}}
\DeclareDocumentEnvironment{problem}{}{
\probleminput{}\problempromise{}\problemquestion{}
\leavevmode
}{
  \par\addvspace{0\baselineskip}
  \noindent
  \begin{itemize}
      \item[\textbf{Input:}] \@probleminput
\ifx\@problempromise\empty%
\else%
      \item[\textbf{Promise:}] \@problempromise
\fi%
      \item[\textbf{Output:}] \@problemquestion
  \end{itemize}
}
\makeatother

\usepackage{xparse,xspace}
\usepackage{tikz}
\usetikzlibrary{quantikz}

\DeclareMathOperator{\poly}{poly}
\DeclareMathOperator{\C}{\field C}
\DeclareMathOperator{\ox}{\otimes}
\DeclareMathOperator{\HS}{\mathcal{H}}
\DeclareMathOperator{\Tr}{Tr}
\DeclareMathOperator{\BigO}{O}
\DeclareMathOperator{\Span}{span}
\newcommand{\HTM}{H_\mathrm{QTM}}
\newcommand{\HM}{H^{(\boxplus,f)}}
\newcommand\1{\mathds{1}}

\newcommand\field\mathds
\newcommand\YES{{\normalfont{\textsc{Yes}}}\xspace}
\newcommand\NO{{\normalfont{\textsc{No}}}\xspace}
\newcommand\alphal{\alpha'_\ell}

\newcommand\ii{\mathrm i}
\newcommand\ee{\mathrm e}
\DeclareMathOperator{\spec}{spec}
\DeclareMathOperator{\diag}{diag}
\DeclareMathOperator{\argmax}{argmax}

\DeclareMathOperator{\enc}{enc}

\newcommand\PQMAEXP{{\normalfont\textsf{P\textsuperscript{QMA\textsubscript{EXP}}}}\xspace}

\newcommand\QMAEXP{{\normalfont\textsf{QMA\textsubscript{EXP}}}\xspace}
\newcommand\QMA{{\normalfont\textsf{QMA}}\xspace}

\newcommand\PQMALOG{{\normalfont\textsf{P\textsuperscript{QMA[log]}}}\xspace}
\newcommand{\Sbr}{\mathcal S_\mathrm{br}}

\DeclareMathOperator{\lmin}{\lambda_\mathrm{min}}

\newcommand\paramath[1]{\ensuremath{\boldsymbol{\mathbf{#1}}}}

\usepackage[backend=biber,style=alphabetic,doi=false,isbn=false,url=false,maxbibnames=5,maxcitenames=2]{biblatex}
\bibliography{references}
\newbibmacro{string+doi}[1]{\iffieldundef{doi}{#1}{\href{https://dx.doi.org/\thefield{doi}}{#1}}}
\DeclareFieldFormat{title}{\usebibmacro{string+doi}{\mkbibemph{#1}}}
\DeclareFieldFormat[article]{title}{\usebibmacro{string+doi}{\mkbibquote{#1}}}
\DeclareFieldFormat[incollection]{title}{\usebibmacro{string+doi}{\mkbibquote{#1}}}

\newcommand\checked[1]{}
\newcommand\todo[1]{}

\title{The Complexity of Approximating Critical Points of Quantum Phase Transitions}

\date{}

\begin{document}
\maketitle

\begin{abstract}
    Phase diagrams chart material properties with respect to one or more external or internal parameters such as pressure or magnetisation; as such, they play a fundamental role in many theoretical and applied fields of science.
    In this work, we prove that provided the phase of the Hamiltonian at a finite size reflects the phase in the thermodynamic limit, approximating the critical boundary in its phase diagram to constant precision is \PQMAEXP-complete.
    This holds even for translationally-invariant nearest neighbour couplings, and even if the system's phase diagram is promised to have a single critical boundary delineating two phases.
    For the simpler case of a single parameter, the same problem remains \QMAEXP-hard.
    
    Our results extend the study of quantum phases to systems with more realistic phase diagrams than previously studied.
    Furthermore, our findings place complexity-theoretic constraints on the effectiveness of (computational or analytic) methods based on finite size criteria, similar in spirit to the Knabe bound, for the task of extrapolating the properties (e.g.~gapped/gapless) of a system from finite-size observations to the thermodynamic limit.
\end{abstract}

\begin{abstract}
\end{abstract}


\clearpage
\enlargethispage{2cm}
\tableofcontents

\section{Introduction and Motivation}\label{sec:intro}
Phase transitions describe a change in the state of matter, such as the transition of liquid water to ice, or vapour.
While the most commonly known such transitions occur at a variety of temperatures, quantum phase transitions describe changes of matter at zero temperature, by varying a physical parameter such as a magnetic field.
Such quantum phase transitions are characterised by an abrupt change in the properties of the ground state of a many-body system in the limit of infinite system size.

Many important physical phenomena are delineated by quantum phase transitions, such as superconducting and insulating phases, or the various regimes of the quantum Hall effect\cite{Gantmakher_2010, Wen_Wu_1993} .
The latter features an intricate quantum phase diagram, a Hofstadter butterfly, which maps the various integer Hall conductances as a fractal pattern \cite{Hofstadter_1976}.

While quantum phase transitions are among the most important and intriguing phenomena studied in modern physics, they are also one of the most complex ones to address, computationally and theoretically, as the previous example highlights.
Indeed a huge amount of effort has gone into designing algorithms and analogue simulations to examine quantum phase transitions, all in order to determine fixed points and universality classes \cite{Keesling_2019,Sepehr_2020}.

How can one determine when a phase transition is about to occur?
A necessary condition is the closing of a spectral gap, that is the energy difference between the minimum energy state and the first excited state.
The question of determining the spectral gap is made even more important by the fact it characterises many other properties of Hamiltonians, beyond playing the role of an indicator for phase transitions \cite{Ambainis2013, Deshpande_Gorshkov_Fefferman_2020}.
A continuous spectrum above the ground state is associated with critical phenomena and power-law decaying correlation functions.
In contrast, a constant spectral gap implies exponentially decaying correlation functions \cite{Hastings_Koma_2006}.
And, in the case of one-dimensional systems, that the ground state can be well approximated by Matrix Product States, and that the ground states obey an entanglement entropy area law \cite{Hastings_2007}.

A key tool in addressing the question of whether a system is gapped in the infinite-size limit that is the so-called ``Knabe bound'' \cite{Knabe_1988}.
Loosely speaking, the Knabe bound is a ``finite-size criterion'' saying that, given a frustration-free Hamiltonian, if the spectral gap of the Hamiltonian decays slowly enough with the system size, then it is necessarily gapped in the thermodynamic limit.
Over time, multiple improvements and variations of the Knabe bound have been derived \cite{Bravyi_Gosset_2015, Gosset_Mozgunov_2016,Lemm_Mozgunov_2019,  Lemm_2019, Anshu_2020}.
Indeed, the Knabe bound has been used extensively to determine phases and spectral gaps of frustration-free systems, including variants of the AKLT model \cite{Abdul-Rahman_et_al_2020,  Lemm_Sandvik_Wang_2020}, to characterise the phases of translationally invariant Hamiltonians on 1D chains of qubits \cite{Bravyi_Gosset_2015, Gosset_Mozgunov_2016}, and characterising gaps for Product Vacua with Boundary States (PVBS) models \cite{Lemm_Nachtergaele_2019}.

Finite-size criteria similar to this are often (implicitly) assumed in condensed matter physics when performing numerical studies: the idea that the phase (or gap) at finite lattice sizes allows us to extrapolate to the phase (or gap) in the thermodynamic limit \cite{Fath_Solyom_1993,Totsuka_Nishiyama_Hatano_Suzuki_1995, Yamamoto_1997, Garcia-Saez_Murg_Wei_2013}. 

Beyond the Knabe method, there exist other techniques for determining whether spectral gaps close or remain open, including the commonly-used Martingale Method \cite{Nachtergaele_1996} and variants thereof \cite{Spitzer_Starr_2003, Kastoryano_Lucia_2018}.
The Martingale Method says that given an absorbing sequence of increasingly large sections of a lattice which tend towards the full lattice, there are three extra conditions placed on the local terms of the Hamiltonian which must be uniformly satisfied along the sequence. 
If so, there is a lower bound on the spectral gap.
Furthermore, several numerical algorithms have been developed to  compute spectral gaps, including variational algorithms \cite{Higgott_Wang_Brierley_2018, Tyson_2019} and using density matrix renormalisation group techniques \cite{Chepiga_Mila_2017}

On the other hand, \citeauthor{Cubitt_Perez-Garcia_Wolf2015} have
shown that the general problem of determining whether a local Hamiltonian is gapped or gapless in the thermodynamic limit is undecidable \cite{Cubitt_Perez-Garcia_Wolf2015},
a result later extended to translationally invariant, nearest neighbour interactions on a 1D spin chain \cite{Bausch_2020_Undecidability}.
Further results show that for a one-parameter Hamiltonian in 2D, determining the phase diagram is uncomputable \cite{Bausch_Cubitt_Watson2019}, and that renormalisation group methods provably fail to resolve these systems \cite{Watson_Onorati_Cubitt_2021}.
Nonetheless, these undecidability and uncomputablity results only imply that it is impossible to study the phase diagrams or spectral gap of models in full generality; no claims are made regarding more restricted sub-types, which may well be solvable.
By considering restricted families of Hamiltonians (e.g.\ frustration-free Hamiltonians as above), we can still hope to determine useful properties about these families.

\paragraph{Our Contribution.}
In this work, we focus our attention on either one-parameter or two-parameter continuous family of Hamiltonian which satisfy a computable ``finite-size criterion'' similar in spirit to the Knabe bound, which gives an avenue to extrapolate to the infinite-size limit, given conditions on how the local gap scales.
Furthermore, the family of many-body systems we include are promised to have a \emph{single} critical boundary; either a single critical point for the one-parameter case or a critical line in the two-parameter setting.
The two phases we delineate are a gapped and gapless phase, but can, in principle, be any other type of phase of interest (e.g.~topological).

The finite size criterion assumes that if the gap decays sufficiently slowly (resp.~quickly) for increasing lattice sizes, then the overall Hamiltonian must be gapped (resp.~gapless) in the thermodynamic limit.
This finite-size criterion immediately implies that the system cannot have an uncomputable phase in the thermodynamic limit, as solving for its phase on some finite-sized region is always a computationally bounded task (if intractable, and where we exclude pathological cases of uncomputable matrix entries or threshold sizes).

As a first contribution, we show that these finite size conditions place an upper bound on the computational complexity of approximating critical lines, by placing them within the class \PQMAEXP---i.e.~solvable by a polynomial-time Turing Machine with a \QMAEXP oracle---and which is deemed ``slightly harder'' than just \QMAEXP \cite{Ambainis2013}. 
The complexity class \QMAEXP is defined as \QMA, but with an exponential-time verifier.\footnote{The EXP subscript is a technicality, as we are interested in translationally-invariant problems, whose input is given by specifying a local term only. For an in-depth discussion see \cite[Sec.~3.4]{Bausch2016}.}

\begin{figure}[!tb]
    \centering\vspace{-1.5cm}
    \hspace*{-3.5cm}
    \begin{minipage}{18cm}
    \includegraphics[height=7cm,trim=1cm 0cm 0cm 0cm, clip]{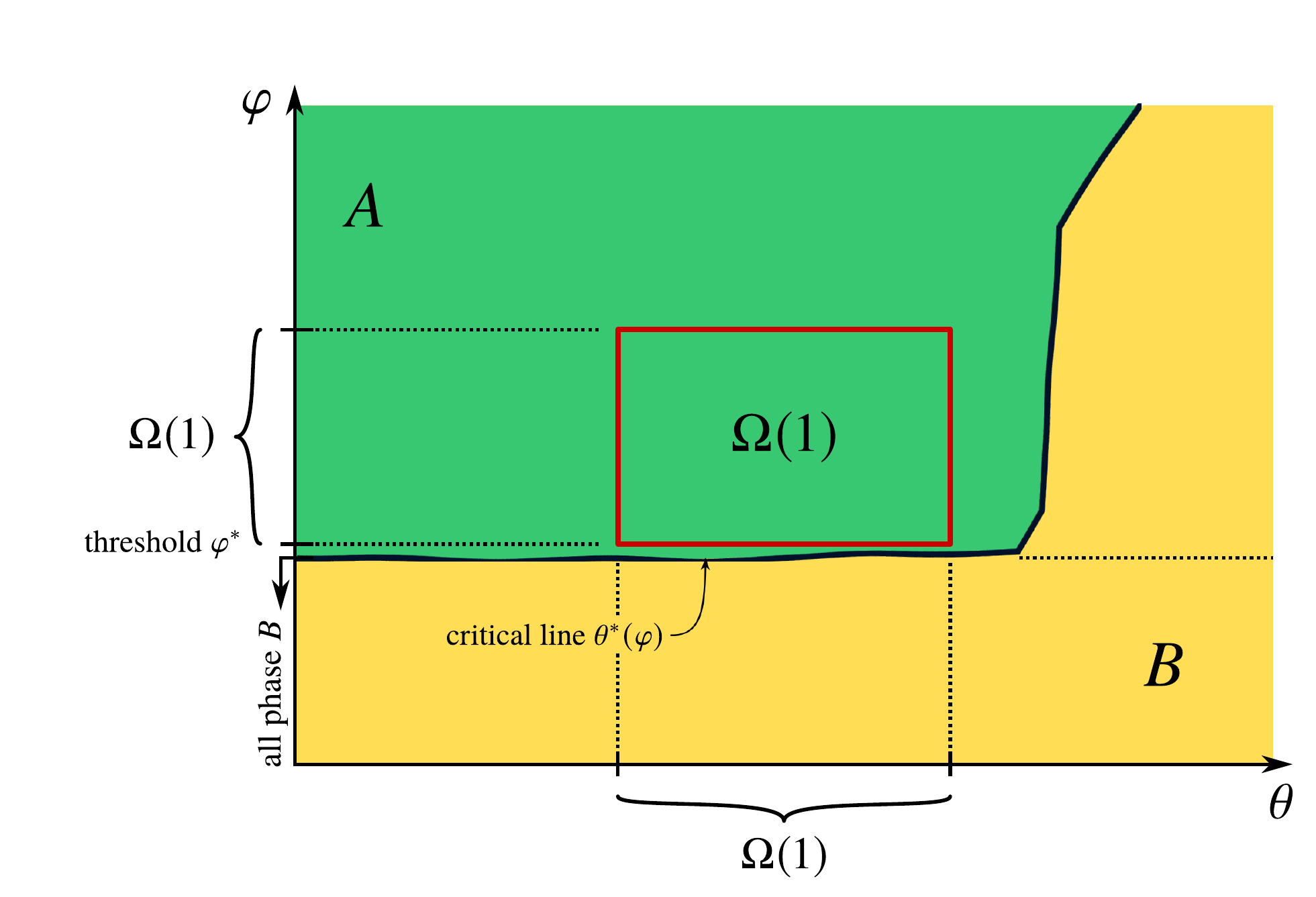}
    \hspace*{-0.5cm}
    \includegraphics[height=7cm,trim=1cm 0cm 0cm 0cm, clip]{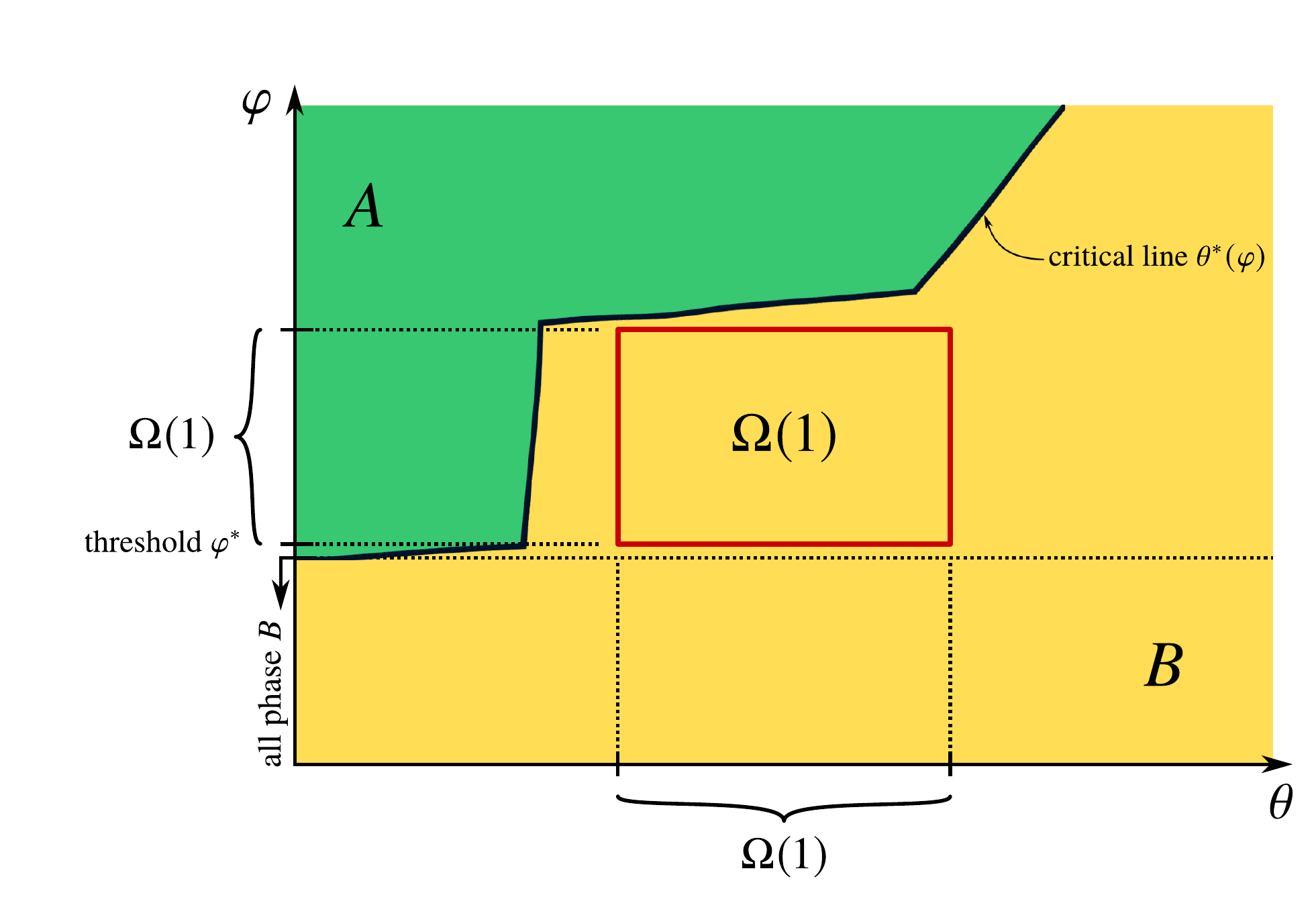}
    \end{minipage}
    \caption{The two possible phase diagrams for the family of 2-parameter Hamiltonians we construct.
    The critical line $\varphi^*(\theta)$ is continuous, and there is a $\Omega(1)$-sized area which in the first case is guaranteed to be completely in phase $A$ (e.g.\ a gapped phase), and in the second case completely in phase $B$ (e.g.\ a gapless phase), for parameters $\theta \times \varphi \in [0, 1] \times [0, \poly N]$.
    We prove that determining which of the two cases holds is a \PQMAEXP-complete problem.}
    \label{fig:phasediags-intro}
\end{figure}

For these restricted sets of one or two-parameter Hamiltonians, we show that the problem of determining the critical boundary in parameter space to even constant precision is \QMAEXP (for the one-parameter case) resp.~\PQMAEXP-hard (for the two-parameter case).

Loosely speaking, we prove the two following theorems (the rigorous versions are given in \cref{sec:main-results}).
\begin{theorem}[Informal]\label{th:main-1-intro}
\checked{James}
Let $\{ H_N \}$ be a family of local translationally invariant Hamiltonians on an infinite lattice, indexed by $N\in \field N$, and such that $H_N=H_N(\varphi)$ depends on an external parameter $\varphi \in [0, 1]$.
Suppose the following promises hold for Hamiltonians in this family: i.~For all $\varphi$, the phase in the thermodynamic limit can be extrapolated from the order parameter obtained from a finite lattice size; and ii.~There exists precisely one critical parameter $\varphi^*$ such that for $\varphi<\varphi^*$, the system is in phase A; otherwise in phase $B$.
Then the problem of determining $\varphi^*$ to even constant precision is \QMAEXP-hard;
and determining $\varphi^*$ up to polynomial precision (in $N$) is contained in \PQMAEXP.
\end{theorem}
\noindent
And the two-parameter case reads as follows.
\begin{theorem}[Informal]\label{th:main-2-intro}
\checked{James}
Let the setup be as in \cref{th:main-1-intro}, but such that now $H_N = H_N(\varphi, \theta)$ depends on two parameters $\varphi \in [0, \poly N]$ and $\theta \in [0, 1]$.
Suppose the following promises hold for all Hamiltonians in this family: i.~For all $\varphi$ and $\theta$, the phase in the thermodynamic limit can be extrapolated from the order parameter obtained from a finite lattice size; and ii.~There exists precisely one critical boundary $\varphi^*(\theta)$ such that on one side the system is in phase A; on the other side in phase B.
Then determining the critical boundary $\varphi^*(\theta)$ to constant precision is \PQMAEXP-complete.
\end{theorem}
We emphasise that in \cref{th:main-2-intro}, the precision to which $\varphi*(\theta)$ is to be resolved in the $\varphi$-direction is constant, but over a $\poly N$ parameter range, in contrast to the $\theta$-direction, and the \cref{th:main-1-intro} case.
The two phase diagram cases that are hard to differentiate between are shown in \cref{Fig:1Param_Phase_Diagram} for the 1-parameter case and \cref{fig:phasediags-intro} for the 2-parameter case.

This means that even for systems satisfying the finite size criterion and a promise of a single phase transition within the parameter range, there is unlikely to be an efficient analytic method or efficient computational algorithm to determine gapped/gaplessness---or, more generally, phase diagrams (unless \QMAEXP is efficiently solvable).
In other words, even if we restrict to the set of Hamiltonians which are known to be gapped/gapless in the thermodynamic limit based on some finite-size scaling criteria, and even if we know that the regions within which the system is gapped or gapless are delineated in a smooth manner, determining the critical boundary between these regions is generally computationally intractable.
Alternatively, any efficiently-computable condition cannot fully characterise the set of all many-body phase diagrams, even if the phase diagrams are simple, and even if the many-body systems satisfy finite-size scaling criteria.

Furthermore, our results imply that for nearest-neighbour, translationally invariant Hamiltonians which have exactly two phases and a single critical boundary, determining the phase diagram to $\BigO(1)$ precision is computationally difficult.
In particular, we show that there exists a Hamiltonian with two parameters $(\varphi,\theta)$ with a phase diagram which looks like one of the cases illustrated in \cref{fig:phasediags-intro}; however determining which one of the two is \PQMAEXP-complete.
We contrast this with the undecidability/uncomputability results which require an infinite number of phases, and which only satisfy a local-global promise for an uncomputably large system.

\paragraph{Proof Idea.}
At a high level, the idea of the paper is to translate known \QMAEXP or \PQMAEXP-complete problems to the spectrum and phase of a Hamiltonian, in such a way that the spectrum becomes dense---or stays gapped---if the reduced problem was a \YES or \NO instance, respectively.

Two natural problems to consider are the translationally-invariant version of the local Hamiltonian problem, proven to be \QMAEXP complete in \cite{Gottesman_Irani_2009}; and the translationally-invariant version of the approximate simulation (APX-SIM) problem, proven to be \PQMAEXP-complete in \cite{Watson_Bausch_Gharibian_2020}.
Conveniently, the local Hamiltonian problem has a single parameter, namely the ground state energy; whereas APX-SIM features two parameters: the ground state energy, \emph{and} a threshold for a local observable.


In the one-parameter case, we construct a quantum Turing machine (QTM) which takes as input a problem instance $N$, a parameter $\varphi$, and an unconstrained state $\ket{\nu}$.
Let $G_N$ be a Hamiltonian for which approximating the ground state is \QMAEXP-complete: the QTM then runs a phase comparison between $\exp(\ii t G_N)$ with eigenvalues $\lambda_j(G_N)$, and $U_\varphi$ which encodes a phase $\varphi$, enabling an approximate comparison $\varphi \lessgtr \lambda_j(G_N)$ to be performed.
If this QTM is translated into a history state Hamiltonian, and the case $\varphi < \lambda_j(G_N)$ is penalised, we have a system that has low energy if and only if $\varphi < \lambda_j(G_N)$.
As $\varphi$ is lowered, this the unconstrained input state $\ket{\nu}$ ultimately self-adjusts so that the comparison becomes $\varphi < \lmin(G_N)$; by penalising states for which $\varphi < \lambda_j(G_N)$, the unconstrained input state will be chosen such that $\varphi > \lambda_j(G_N)$ provided there is such a state.

This dichotomy in the ground state energy is then combined with: i.~a ``bonus'' energy Hamiltonian which gives an unconditional energy bonus, which modifies the dichotomy to a ground state energy either $>0$ (for \NO), or $<0$ (for \YES); ii.~a checkerboard tiling Hamiltonian, such that the history plus bonus Hamiltonian are repeated periodically across a 2D spin lattice, which raises the dichotomy to a ground state energy diverging to either $+\infty$ or $-\infty$. 
Finally, by combining this with a Hamiltonian with a dense spectrum, we can engineer it such that an order one spectral gap that is maintained if the ground state energy spectrum diverges to $+\infty$, while if the energy goes to $-\infty$ the dense Hamiltonian is ``pulled down'' and the overall spectrum is dense.

This mechanism and setup are akin to the one used in \cite{Bausch_2020_Undecidability}; however there is a number of crucial differences and obstacles we needed to address.
First off, to combine a history state Hamiltonian and a bonus Hamiltonian in a way as to place the energy bonus precisely in between the promise gap of a history state Hamiltonian means we needed to gear together tight spectral bounds on either system. For the history state Hamiltonian, we rely on a result derived in \cite{Watson_2019}; for the bonus Hamiltonian, we obtain exponentially tight bounds on the bonus energy inflicted, significantly strengthening the bounds derived in both \cite{Bausch_2020_Undecidability,Bausch_Cubitt_Watson2019}.

The two-parameter case is similar to the single parameter case, yet differs in a few crucial aspects.
In essence, we encode the same QTM in a history state, yet one that performs phase estimation on a \PQMAEXP-complete Hamiltonian $K_N$ as in \cite{Watson_Bausch_Gharibian_2020}.
However, we now choose the history state penalty to be proportional to $\bra{\psi}B\ket{\psi}-\theta$, where $B$ is a local observable, $\ket{\psi}$ is the unconstrained state, and $\theta$ is a parameter in the Hamiltonian we are able to vary.
For the APX-SIM local Hamiltonian construction in \cite{Watson_Bausch_Gharibian_2020}, the low energy eigenstates are promised to have a particular value of $\bra{\psi}B\ket{\psi}$.
Thus the additional output penalty that we add gives an energy offset depending on the value of $\bra{\psi}B\ket{\psi}-\theta$.
From \cite{Watson_Bausch_Gharibian_2020}, determining $\bra{\psi}B\ket{\psi}$ is \PQMAEXP-complete.

Overall, by choosing the value of $\varphi$ to lie within some interval $[\lmin(K_N), \lmin(K_N) + \delta]$ above the ground state of $K_N$, the ground state energy of the Hamiltonian encoding the computation will then be either $<a$ if $\bra{\psi}B\ket{\psi}>\theta$, or $>b$ in case of $\bra{\psi}B\ket{\psi}<\theta$, where there is a constant offset between $a$ and $b$.\footnote{This point is subtle but crucial: our construction does \emph{not} prove a reduction from \PQMAEXP to the translationally-invariant local Hamiltonian problem, and hence does not show the containment $\PQMAEXP \subset \QMAEXP$, since we do not know the ground state $\lmin(K_N)$ in the first place. Obtaining it to within sufficient precision requires $\poly$-many queries to a \QMAEXP oracle.}
Otherwise, if $\varphi<\lmin(K_N)$, the ground state is again $>b$; and if $\varphi > \lmin(K_N) + \delta$, we make no prediction (and it is conceivable that eventually, for larger and larger $\varphi$, the ground state energy again drops $<a$, but it could also lie anywhere in between).

The resulting Hamiltonian is again combined with a bonus Hamiltonian (with negative ground state energy), checkerboard tiling, and dense and trivial Hamiltonian as in the one-parameter case.
As such, for the range of $\varphi \in [\lmin(K_N), \lmin(K_N) + \delta]$, we know that for an $\Omega(1)$ range of $\theta$ the system is either completely in the gapped or gapless phase, and that for any $\varphi$, there exists at most one critical point $\theta^*(\varphi)$ delineating the two phases (here $\theta^*(\varphi)$ denotes the critical point in $\theta$ when $\varphi$ is held constant).
Indeed, we show more, namely that for any $\theta$ there exists precisely one critical point $\varphi^*$, and that the resulting critical line $\varphi^*(\theta)$ is continuous.

A schematic of the phase diagrams for the two-parameter case is shown in \cref{fig:phasediags-intro} (more details in \cref{Fig:2Param_Phase_Diagram-NO,Fig:2Param_Phase_Diagram-YES}), and the one-parameter case in \cref{Fig:1Param_Phase_Diagram}.
For a more detailed overview of the one- and two-parameter constructions and proof ideas see the introductory paragraphs in \cref{sec:1-hard,sec:2-hard}.

\enlargethispage{2cm}
\thispagestyle{empty}
\paragraph{Paper Structure.}
The paper is organised as follows. In \cref{sec:prelims} we set up necessary notation, give basic definitions (gapped/gapless, APX-SIM, local Hamiltonian, etc.) and summarise some results derived elsewhere that our construction relies on.
Furthermore, we formally define the local-global conditions (\cref{Def:Local-Global_Gap,Def:Local-Global_Phase}), as well as the critical parameter problems 1-CRT-PRM and 2-CRT-PRM in \cref{def:1-crt-prm,def:2-crt-prm}.
In \cref{sec:main-results}, we summarise the two main results, leading up to the containment proof in \cref{sec:containment}, and the two explicit hardness constructions in \cref{sec:1-hard,sec:2-hard}.
In \cref{sec:discussion}, we finally give a fairly extensive discussion of further implications and open questions.

\section{Preliminaries}\label{sec:prelims}
\subsection{Notation}\label{sec:notation}

Let $\mathcal{B}(\mathcal{H})$ be the space of bounded linear operators on a complex Hilbert space $\mathcal{H}$.
Following convention,
$\lmin(H)$ will denote the ground state energy (i.e the minimum eigenvalue) of a Hamiltonian $H$  (a Hermitian, positive semi-definite matrix), and larger eigenvalues will be denoted $\lambda_i(H)$, such that $\lmin(H)\le \lambda_1(H)\le \lambda_2(H)\le \ldots$.
The spectral gap $\Delta(H)$ of a Hamiltonian $H$ is the difference between the first excited state and the ground state energy, $\Delta(H)\coloneqq\lambda_1(H) - \lmin(H)$, which we assume to be zero in the case of a degenerate ground space.
$\Lambda(L\times H)\subseteq \mathds{Z}^2$ denotes a rectangular $L\times H$ sublattice of $\mathds{Z}^2$; and with $\Lambda(L)\coloneqq\Lambda(L \times L)$ a square grid of side length $L\times L$.
We set $\Lambda = \field Z^2$ as the infinite square lattice.

We denote the evaluation of logical formula as with square brackets. 
For example, for $x,y\in \mathbb{R}$, $[x>y]$ is equal to 1 if $x>y$ and $0$ if $x\leq y$.
For $z\in \{0,1\}$, then $[[x>y]\land z]$ is then the logically AND of $[x>y]$ and $z$.
We further abbreviate the integer set $[n] \coloneqq \{0,\ldots,n-1\}$ for $n\in\field N$.

For a positive integer $N\in\field N$, we will write $|N| \coloneqq \lceil \log_2 N \rceil$ as the number of bits necessary to specify $N$; in case a base other than $2$ is specified, the definition changes accordingly. Which base is meant is always clear from context, as is the case when $| \cdot |$ is meant as an absolute value.
For sums in large expressions we will often write $\sum_{y<x}$ to represent a sum over all $y$ and $x$ such that $y<x$, instead of $\sum_x \sum_{y<x}$.
The domain and running variables in these cases will be clear from the context.

Finally, following \cite{Cubitt_Perez-Garcia_Wolf2015}, we adopt the following definitions of gapped and gapless:
\begin{definition}[Gapped, from \cite{Cubitt_Perez-Garcia_Wolf2015}]\label{Def:gapped}
  We say that $H^{\Lambda(L)}$ of Hamiltonians is gapped if there is a constant $\gamma>0$ and a system size $L_0\in\mathbb{ N}$ such that for all $L>L_0$, $\lambda_0(H^{\Lambda(L)})$ is non-degenerate and $\Delta(H^{\Lambda(L)})\geq\gamma$. In this case, we say that \emph{the spectral gap is at least $\gamma$}.
\end{definition}
\begin{definition}[Gapless, from \cite{Cubitt_Perez-Garcia_Wolf2015}]\label{Def:gapless}
  We say that $H^{\Lambda(L)}$ is gapless if there is a constant $c>0$ such that for all $\epsilon>0$ there is an $L_0\in\mathbb{N}$ so that for all $L>L_0$ any point in $[\lambda_0(H^{\Lambda(L)}),\lambda_0(H^{\Lambda(L)})+c]$ is within distance $\epsilon$ from $\spec H^{\Lambda(L)}$.
\end{definition} 
\noindent We note that these definitions of gapped and gapless do not characterize all Hamiltonians; there are Hamiltonian which fit into neither definition, such as systems with a closing gap or degenerate ground states.

\subsection{Quantum Phase Transitions}
We formally define one- and two- critical parameter problems rigorously in this section; to this end, we first need to rigorously introduce what we mean by the terms \emph{phase transition} and \emph{critical parameter} for a local Hamiltonian in the thermodynamic limit.
We closely follow the notion by \citeauthor{Sachdev_2011}.

\begin{definition}[Quantum Phase Transition, from \cite{Sachdev_2011}]
\checked{James}
Consider a local Hamiltonian $H(\varphi)=\sum h_i(\varphi)$, where the matrix entries of $h_i(\varphi)$ are analytic in $\varphi$.
In the thermodynamic limit, a \emph{quantum phase transition} occurs where there is a non-analytic change in the ground state energy $\lmin(H(\varphi))$ as a function of $\varphi$.
If the matrix elements of a Hamiltonian are functions of multiple parameters $h_i(\varphi_1, \theta_2,\dots)$, then there is a phase transition at $\varphi=\varphi^*$ if there is a non-analytic change in the ground state energy at $\varphi^*$ when all other parameters are held constant.
\end{definition}

\begin{definition}[Critical Parameter/Critical Point, $\varphi^*$]\label{def:critical-param}
\checked{James}
 Let $H(\varphi, \{\theta_i\}_i)$ be a Hamiltonian defined in the thermodynamic limit which undergoes a phase transition as a function of $\varphi$. 
Then a \emph{critical parameter} is a point at which the phase transition happens when all other parameters are held constant.
We denote this point with $\varphi^*$.
\end{definition}

\subsection{Finite Systems and Relation to the Thermodynamic Limit}

For experimentalists with access to only finite-sized systems and finite precision measurements, it may not be possible to determine where non-analyticities occur in the ground state energy, and thus to determine where the phase transitions take place.
Instead of pinpointing the critical point itself, measurements that indicate which phase one is currently in are often used instead; multiple such observations then allow bounding the exact location where a phase transition is likely to occur.

\paragraph{Order Parameters.}
Commonly an \textit{order parameter} is an observable which characterizes a certain property of a phase and is used to distinguish phases for a system.
One possible definition is given as follows.
\begin{definition}[Efficiently Computable Local Order Parameter]
    Assume a translationally invariant Hamiltonian $H^{\Lambda(L)}(\varphi)=\sum_{\langle i,j \rangle}h_{i,j}(\varphi)$ has two phases $A$ and $B$, and let $h_{i,j}(\varphi)$ be describable in $n$ bits.
    A \emph{local order parameter} $O_{A/B}$ is a projector $O_{A/B}$ that is $O(1)$-local and computable in $\poly(n)$ time, if $\langle O_{A/B} \rangle$  undergoes a non-analytic change between phases $A$ and $B$ at a critical parameter $\varphi^*$.
\end{definition}
\noindent All order parameters we will use will trivially fall into this category, as can be easily verified.

For a finite system, examining how the order parameter of the systems changes could yield insight over the phase that the system is in when approaching the thermodynamic limit, assuming that for a particular set of parameter values $\varphi, \theta, \dots$, and a sufficiently large systems size the system is in the same phase as in the limit.
While this is not generically the case \cite{Cubitt_Perez-Garcia_Wolf2015,Bausch_2020_Undecidability,Bausch_Cubitt_Watson2019},
many Hamiltonians studied in nature \emph{do have} this property: for sufficiently large but finite sizes, the order parameter indicates the same phase as it would in the thermodynamic limit \cite{Palmer_Tracy_1981, Brezin_Zinn-Justin_1985, Hamer_Barber_1988}.

We formalize this notion of being able to extrapolate from a finite-size system to the thermodynamic limit with the following definitions.
\begin{definition}[Local-Global Phase Condition]\label{Def:Local-Global_Phase}
\checked{James}
     A translationally invariant Hamiltonian $H^{\Lambda(L)}(\varphi)=\sum_{\langle i,j \rangle}h_{i,j}(\varphi)$ defined on a square lattice $\Lambda(L)$ is defined to have the \emph{locally-globally phase property} if in the thermodynamic limit it has two phases $A$ and $B$ distinguished by an order parameter $O_{A/B}$. 
     If $h_{i,j}(\varphi)$ is describable in $n$ bits,
     then there exists an integer $L_0= \BigO(\exp(\poly(n)))$, an
     $\omega = \Omega(1/\poly(L))$
     and polynomials $p,q$ with $1/p(L)-1/q(L)=\Omega(1/\poly(L))$ such that for the  states $\ket{\psi}$ satisfying $\bra{\psi}H^{\Lambda(L)}(\varphi)\ket{\psi}\leq \lmin(H^{\Lambda(L)}(\varphi)) + \omega $ the following holds: 
 \begin{itemize}
        \item[Phase A:] if for $L= L_0$ and $|\varphi - \varphi^*|\geq 1/\poly(L)$, such that $\bra{\psi} O_{A/B} \ket{\psi} \leq  1/p(L)$, then $H^{\Lambda(L)}(\varphi)$ is in phase $A$ for all $L\geq L_0$ and in the thermodynamic limit.
        \item[Phase B:] if for all $L= L_0$ and $|\varphi - \varphi^*|\geq 1/\poly(L)$, such that $\bra{\psi} O_{A/B} \ket{\psi} \geq  1/q(L)$, then $H^{\Lambda(L)}(\varphi)$ is in phase $B$ for all $L\geq L_0$ and in the thermodynamic limit.
\end{itemize}
    Furthermore, $L_0$ is independent of $\varphi$ and is computable in time $\poly(n)$.
\end{definition}

We note that this condition is fulfilled by many Hamiltonians.
For example, for the transverse quantum Ising model in 1D, a phase transition occurs in terms of the ratio of magnetic field strength to coupling strength.
In this case, we see that magnetization along the $Z$ direction acts as an order parameter and satisfies the above condition \cite{Sachdev_2011}.

\paragraph{The Spectral Gap.}
Treating the Local-Global phase condition in a more generic sense, i.e.\ as a property that holds locally and lets you deduce thermodynamic properties, unveils other familiar conditions that allow a similar picture.
One such set of conditions are Knabe bounds \cite{Knabe_1988}.
In brief, Knabe bounds treat the case where one phase has is gapped while the other is gapless, and the gapped case can be discriminated by the gap behaviour on finite-sized systems.

More concretely, the condition is such that the spectral gap remains open---i.e.\ lower-bounded by a constant---in the thermodynamic limit if it closes sufficiently slowly at finite lattice sizes.
In other words, if we take a Hamiltonian restricted to a finite but growing region of the interaction graph, and if the spectral gap on this restricted graph closes slowly enough then the system is guaranteed to be gapped in the thermodynamic limit. 
Such a Local-Global gap bound has been shown to exist for the unfrustrated case \cite{Knabe_1988}.

In particular, Knabe proved that if the local gap on $N$ spins is larger than the threshold value
$1/(N - 1)$ for some $N > 2$ the system is gapped in the thermodynamic limit \cite{Knabe_1988}. 
Recently improvements to this bound have been made such that the threshold value is $6/N(N+1)$ which is known to be asymptotically optimal \cite{Gosset_Mozgunov_2016}.
Another well-known method of bounding spectral gaps is given by Nachtergaele \cite{Nachtergaele_1996}, successively improved by \cite{Spitzer_Starr_2003, Kastoryano_Lucia_2018}. Based on relations between ground space projectors, it says that if a model is gapless then the spectral gap cannot decay too slowly with system size.

We define a ``finite-size criterion'' which states that provided the spectral gap decays sufficiently quickly/slowly as the lattice size increases above some computable lattice size $L_0$, then the Hamiltonian is gapless/gapped in the thermodynamic limit.
Given this criterion, we examine the complexity of Hamiltonians which satisfy it:
\begin{definition}[Local-Global Gap Condition]\label{Def:Local-Global_Gap}
\checked{James}
    A translationally invariant Hamiltonian $H^{\Lambda(L)}(\varphi)=\sum_{\langle i,j \rangle}h_{i,j}(\varphi)$ defined on a square lattice $\Lambda(L)$ is defined to be \emph{locally-globally gapped} if there exist polynomials $p,q$ with $1/p(L)-1/q(L)=\Omega(1/\poly(L))$ and $L_0\in\field N$ such that
    \begin{itemize}
        \item[Gapped:] if $\Delta(H^{\Lambda(L)}(\varphi)) \ge 1/p(L)$ for some $L= L_0$, then there exists a constant $c>0$ such that in the thermodynamic limit the spectral gap $\lim_{L\rightarrow\infty}\Delta(H^{\Lambda(L)}(\varphi))\geq c$.
        \item[Gapless:] if $\Delta(H^{\Lambda(L)}(\varphi)) \le 1/q(L)$ for some $L= L_0$, then the Hamiltonian is gapless in the thermodynamic limit.
    \end{itemize}
    Furthermore, $L_0$ is independent of $\varphi$. 
    Let $h_{i,j}(\varphi)$ be describable in $n$ bits, then $L_0$ is computable in time $\poly(n)$ and $L_0=\BigO(\exp(\poly n))$.
\end{definition}

This is in part motivated by Knabe's bound: if we consider the gapped case as providing a lower bound (and ignoring the gapless case with its $1/q(L)$ bound) then we get a version of Knabe's bound where it is explicitly promised that the system size for which the ``gappedness'' can be verified at is poly-time computable.

Knabe's bound states that provided for some system size $N>2$ if the spectral gap is greater than $1/(N-1)$ then the Hamiltonian is gapped in the thermodynamic limit.
However, the point at which this happens can be at \emph{any} $N$.
As such this does not necessarily conflict with the possibility of undecidability results for frustration-free Hamiltonians (similar to those proved in \cite{Cubitt_Perez-Garcia_Wolf2015}).
Instead, the size of the systems for which the Knabe bound holds would be uncomputable (and thus actually implementing the bound is uncomputable).
The Local-Global gap we have assumed to hold in \cref{Def:Local-Global_Gap} is explicitly computable and $L_0$ is assumed to be computable in polynomial time, thus forcing the problem of determining the spectral gap to be decidable for systems satisfying \cref{Def:Local-Global_Gap}.

As such our promise is strictly stronger than the regular Knabe bound.
It is further worth noting that the polynomials that bound the gaps will vary between different classes of Hamiltonians and that by standard padding arguments the exact scaling of these polynomials is not overly relevant: if, for instance, the gapped/gapless property of the many-body system is determined by patches distributed across the system, where the density of the patches goes as $\sim 1 / \sqrt L$ in the system size $L$, the polynomials $p(L)$ and $q(L)$ can effectively be replaced by $p(\sqrt L)$ and $q(\sqrt L)$.
This argument is the same as for the local Hamiltonian problem, where the $1/\poly$ promise gap can be magnified with the same technique---naturally without any interesting implications regarding the problems' complexity, or the resulting phenomenology exhibited by the system.

In this work, we shall consider different classes of Hamiltonians that have distinct polynomial scalings; the exact form of which is irrelevant, but crucially all of them obey the finite size criteria as per \cref{Def:Local-Global_Gap}.

\subsection{Critical Parameter Problem Definitions}
With the notion of phase transitions and Local-Global properties clarified within the last two sections, we can now define the Critical Parameter Problem (1-CRT-PRM) as follows.
\begin{definition}[1-CRT-PRM{$_f$}]\label{def:1-crt-prm}
\begin{problem}
\probleminput{
    $N\in\field N$.
    A constant-size set of $k$-local interaction terms $h^{(l)}(\varphi)\in \mathcal{B}((\field C^d)^{\otimes S_l})$, for $l\in I$, and such that $S_l\subset \Lambda$, and $|S_l|\le k\ \forall l$.
    Two positive numbers $\alpha < \beta$ which satisfy $\beta-\alpha=\Omega(1/f(N))$.
    $\alpha,\beta$ and the matrix entries of each of the $\{ h^{(l)}(\varphi) \}$ are specified to $|N| \coloneqq \lceil \log_2 N \rceil$ bits of precision.
}
\problempromise{
    Let $S+(i,j)\coloneqq\{ (a+i, b+j) : (a,b) \in S \}$.
    Define the translationally-invariant Hamiltonian on $\Lambda$ via
    \begin{equation}\label{eq:H-on-lattice}
    H \coloneqq \sum_{l\in I}\sum_{(i,j)\in\Lambda} h_{S_l+(i,j)}(\varphi).
    \end{equation}
    H has two phases A and B, and satisfies a Local-Global property as per \cref{Def:Local-Global_Phase}  or \cref{Def:Local-Global_Gap}, for some $L_0=\poly N$, independent of $\varphi$.
    There is precisely one critical parameter $\varphi^*\in [0,1] \setminus [\alpha,\beta]$ as per \cref{def:critical-param}.
    If $\varphi < \varphi^*$ the system is in phase A, and for $\varphi > \varphi^*$ it is in phase B.
}
\problemquestion{
    \YES if $\varphi^*\leq \alpha $.\\
	\NO if $\varphi^*\geq \beta $.
}
\end{problem}
\end{definition}

\noindent
This problem characterises the hardness of estimating the point at which there is a phase transition (e.g.\ from gapped to gapless), for a one-parameter translationally-invariant Hamiltonian which is promised to have precisely one such transition; the input is the specification of the Hamiltonian to precision $N$, and the input size (given in, say, binary) is thus linear in $|N|$.

We kept the \emph{precision} to which the critical point has to be estimated as a parameter, denoted by the function $f$ subscript to the problem definition; in brief, 1-CRT-PRM$_{\poly}$ denotes the case where we wish to approximate it to precision $\Omega(1/\poly)$, and we define the ``precise'' version of the problem to be Precise-1-CRT-PRM = 1-CRT-PRM$_{\exp\poly}$.
For ease of notation, we let 1-CRT-PRM = 1-CRT-PRM$_{\Theta(1)}$ be the case where the phase transition has to be estimated to constant precision.

We emphasise that while the subscript determines the precision to which we wish to compute the critical point, the Local-Global property as per \cref{Def:Local-Global_Gap,Def:Local-Global_Phase} is still required with polynomial precision throughout.
Natural extensions of \cref{def:1-crt-prm} where either scalings are given as parameters are of course possible.

A two-parameter variant of 1-CRT-PRM can be defined analogously as follows.
\begin{definition}[2-CRT-PRM$_f$]\label{def:2-crt-prm}
\begin{problem}
\probleminput{
    $N\in\field N$. 
    A finite set of $k$-local interactions $h^{(l)}(\theta,\varphi)\in \mathcal{B}((\field C^d)^{\otimes S_l})$, for $l\in I$, and such that $S_l\subset \Lambda$, and $|S_l|\le k\ \forall l$.
    Four positive numbers $\alpha_1 < \beta_1$ and $\alpha_2 < \beta_2$, such that the rectangle $[\alpha_1, \beta_1] \times [\alpha_2, \beta_2]$ covers an $\Omega(1/f(N))$ area.
    $\alpha_1,\beta_1, \alpha_2,\beta_2$ and the matrix entries of each of the $\{ h^{(l)}(\theta,\varphi) \}$ are specified to $\poly(|N|)$ bits of precision.
}
\problempromise{
    $H$ is defined as in \cref{eq:H-on-lattice}, and satisfies a Local-Global property for two phases A and B, as per \cref{Def:Local-Global_Phase} or \cref{Def:Local-Global_Gap} for $L_0=\poly N$, independent of $\theta$ and $\varphi$.
    The  critical line $\theta^*(\varphi)$ is a function of $\varphi$---i.e.\, for each $\varphi$, there exists exactly one critical parameter $\theta^*$ as per \cref{def:critical-param} such that $H(\theta,\varphi)$ is in phase $A$ if $\theta < \theta^*$, and in phase $B$ if $\theta > \theta^*$.
    The critical line $\theta^*(\varphi)$ is either such that the rectangle $R \coloneqq [\alpha_1, \beta_1 + y] \times [\alpha_2, \beta_2 + y]$ lies completely in phase $A$, or completely in phase $B$, where
    \[
        y \coloneqq \max \{ \varphi^* : H(\theta, \varphi)\ \text{is in phase $B$ } \forall \varphi < \varphi^*, \forall \theta \}.
    \]
    is promised to be well-defined, and that $H(0, \varphi)$ is in phase $A$ for all $\varphi>y$, and in phase $B$ for all $\varphi<y$.
}
\problemquestion{
    \YES if the critical line is such that rectangle $R$ lies in phase $A$.\\
	\NO otherwise.\footnote{This order is switched compared to the 1-CRT-PRM case, matching the bounds for APX-SIM, whereas the bounds for 1-CRT-PRM match those of the Local Hamiltonian problem.}
}
\end{problem}
\end{definition}

We emphasise that this definition is a direct analogue of the 1-CRT-PRM case, where the rectangle's role was taken by the one-dimensional interval $[\alpha, \beta]$; the offset $y$ is necessary to obtain a well-defined problem definition, and is analogous to how APX-SIM (\cref{Def:APX-SIM}) is defined with respect to a ``natural reference point'', i.e.\ the Hamiltonian's ground state energy. The ``natural reference point'' for phase diagrams we choose is simply a point along one of the parameter axes below which the system is completely in one of the two phases, as shown in \cref{fig:phasediags-intro}.

Other equivalently well-motivated definitions can of course be given. 
We give the following variant of 2-CRT-PRM, which reads more akin to the way 1-CRT-PRM is formulated, but which is otherwise identical in meaning to \cref{def:2-crt-prm}.
\begin{definition}[2-CRT-PRM$_f$, alternative formulation]\label{def:2-crt-prm_2}
\begin{problem}
\probleminput{
    $N\in\field N$. 
    A finite set of $k$-local interactions $h^{(l)}(\theta,\varphi)\in \mathcal{B}((\field C^d)^{\otimes S_l})$, for $l\in I$, and such that $S_l\subset \Lambda$, and $|S_l|\le k\ \forall l$.
    Positive numbers $\varphi, \alpha_1,\beta_1$,  such that $\beta_1 -\alpha_1 =\Omega(1)$.
    $\alpha_1,\beta_1, \varphi$ and the matrix entries of each of the $\{ h^{(l)}(\theta,\varphi) \}$ are specified to $\poly(|N|)$ bits of precision.
}
\problempromise{
    $H$ is defined as in \cref{eq:H-on-lattice}, and satisfies a Local-Global property for two phases A and B, as per \cref{Def:Local-Global_Phase} or \cref{Def:Local-Global_Gap} for $L_0=\poly N$, independent of $\theta$ and $\varphi$.
    The  critical line $\theta^*(\varphi)$ is a function of $\varphi$---i.e.\, for each $\varphi$, there exists exactly one critical parameter $\theta^*$ as per \cref{def:critical-param} such that $H(\theta,\varphi)$ is in phase $A$ if $\theta < \theta^*$, and in phase $B$ if $\theta > \theta^*$.
    It is promised that
    \[
        y \coloneqq \max \{ \varphi^* : H(\theta, \varphi)\ \text{is in phase $B$ } \forall \varphi < \varphi^*, \forall \theta \}
    \]
    is well-defined, and that $H(0, \varphi)$ is in phase $A$ for all $\varphi>y$, and in phase $B$ otherwise.
    Furthermore, there is an interval $S_\kappa \coloneqq [y+ \kappa, y+2\kappa]$ for $\kappa = \Omega(1/f( N))$, such that if $\varphi \in S_\kappa$, then for all $\varphi \in S_\kappa$ either $\theta^*(\varphi)>\beta_1$ or $\theta^*(\varphi)<\alpha_1$ .
  
}
\problemquestion{
    \YES $\theta^*(\varphi)>\beta_1$ for all $\varphi\in S_\kappa$.\\
	\NO $\theta^*(\varphi)<\alpha_1$ for all $\varphi\in S_\kappa$.\footnote{This order is switched compared to the 1-CRT-PRM case, matching the bounds for APX-SIM, whereas the bounds for 1-CRT-PRM match those of the Local Hamiltonian problem.}
}
\end{problem}
\end{definition}

While \cref{def:2-crt-prm,def:2-crt-prm_2} sound somewhat contrived, we emphasise that the core idea behind it is analogous to how APX-SIM is defined.
The latter asks: if I take a low-energy state, what's the expectation value with respect to a given observable? It is natural, in this context, to imply the ground state within the problem definition, and not to demand it to be given as input in first place.
In a similar fashion, to approximate a critical line in a phase diagram to some desired precision it is conceivable that one knows a region below which the phase diagram is entirely in one phase; and to draw the critical line from this reference point onward to some desired precision.

Just as the precision to which we wish to approximate the ground state energy dictates how \emph{hard} a problem it will be, the Local-Global properties are in place to ensure we can prove \emph{containment} of the problems, i.e.\ to place them within a complexity class that solely depends on how hard it is to solve the Local-Global property.
Example variants---and the ones we will focus on in this paper---are when phase A and B are gapped vs.\ gapless states; the Local-Global property is then given by \cref{Def:Local-Global_Gap}.

When proving complexity results, we will be interested in the class \QMAEXP, which is to \QMA what NEXP is to NP; its use over \QMA is a technicality based on how the input for a translationally-invariant system is specified (i.e., as a single interaction term repeated over the lattice).
Formally the class is defined as follows:
\begin{definition}[\QMAEXP~\cite{Gottesman_Irani_2009}] \label{Def:QMA_EXP}
\checked{James}
	A promise problem $\Pi=(A_{YES},A_{NO})$ is in $\QMAEXP$ if there exists a $k\in \BigO(1)$ and a Quantum Turing Machine $M$ such that for each input $x\in\{0,1\}^*$ with $|x|=n$, and any proof $\ket{\psi}\in (\C^d)^{\otimes 2^{n^k}}$, on input $(x,\ket{\psi})$, $M$ halts in $\BigO(2^{n^k})$ steps. Furthermore, the following conditions hold.
	\begin{enumerate}
		\item (Completeness) If $x \in A_{YES}$, there exists a state $\ket{\psi}\in (\C^d)^{\otimes 2^{n^k}}$ such that $M$ accepts $(x,\ket{\psi})$ with probability $\ge2/3$.	
		\item (Soundness) If $x \in A_{NO}$, then for all $\ket{\psi}\in (\C^d)^{\otimes 2^{n^k}}$, $M$ accepts $(x,\ket{\psi})$ with probability $\le1/3$.
	\end{enumerate}
\end{definition}

\subsection{APX-SIM}
To prove the hardness result for the two-parameter case, we will prove a reduction to the hardness of simulating measurements on local Hamiltonians.
Generically, the problem has the following setup:

\begin{definition} [{APX-SIM($H,A,a,b,\delta$), \cite{Ambainis2013}}]\label{Def:APX-SIM}
\begin{problem}
\probleminput{
Local term $h$ of a translationally invariant Hamiltonian $H=\sum h $, and a local observable $A$, and real numbers $a$, $b$ such that $b-a\geq N^{-c'}$, for $N$ the number of qubits $H$ acts on and $c,c'>0$ some constants.
  
}
\problempromise{
    Either $H$ has a ground state $\ket{\psi}$ with $\bra{\psi}A\ket{\psi} \leq a$ or for any $\ket{\psi}$ with $\bra{\psi}H\ket{\psi}\leq \lmin(H) +\delta$, we have $\bra{\psi}A\ket{\psi} \geq  b$.
}
\problemquestion{
    \YES$\ket{\psi}$ with $\bra{\psi}A\ket{\psi} \leq a$.\\
	\NO for any $\ket{\psi}$ with $\bra{\psi}H\ket{\psi}\leq \lmin(H) +\delta$, we have $\bra{\psi}A\ket{\psi} \geq  b$.
}
\end{problem}
\end{definition}

As we aim to establish the hardness of  2-CRT-PRM for translationally-invariant Hamiltonians, we will consider the following variant of APX-SIM.
We note that the $\forall$ quantifier is a technically
(the $\forall$ case is trivially easier than the non-$\forall$ case since it has a stronger promise; however, they can be shown to be equally hard).
%

\begin{definition}[$\forall$-TI-APX-SIM($H,A,a,b,\delta$) \cite{Gharibian_Piddock_Yirka2019}]\label{Def:forall_APX-SIM}
\begin{problem}
\probleminput{
Local term $h$ of a translationally invariant Hamiltonian $H=\sum h $, and a local observable $A$, and real numbers $a$, $b$ such that $b-a\geq N^{-c'}$, for $N$ the number of qubits $H$ acts on and $c,c'>0$ some constants.
  
}
\problempromise{
 Let $S_\delta$ be the set of all states $\ket{\psi}$ satisfying $\bra{\psi}H\ket{\psi}\le \lambda(H)+\delta$ for $\delta\geq N^{-c}$.
 For any $\ket{\psi}\in S_\delta$, we have either $\bra{\psi}A\ket{\psi}\geq b$ or $\bra{\psi}A\ket{\psi} \le a$.
}
\problemquestion{
    \YES$\bra{\psi}A\ket{\psi}\geq b$ for all $\ket{\psi}\in S_\delta$.\\
	\NO $\bra{\psi}A\ket{\psi} \le a$ for all $\ket{\psi}\in S_\delta$.
}
\end{problem}
\end{definition}
\noindent Thus the problem promises that all states within energy $\delta$ of the ground state (i.e. low-energy states) have similar expectation values when measured relative to $A$.
(Note we have switched the yes and no instances from those defined in \cite{Gharibian_Piddock_Yirka2019}, however, this choice is arbitrary and does not change any results).

We will make use of the following result 
\begin{theorem}[Implicitly proved by \citeauthor{Watson_Bausch_Gharibian_2020} \cite{Watson_Bausch_Gharibian_2020}] \label{Lemma:APX-SIM_Complexity}
$\forall$-TI-APX-SIM($H,A,a,b,\delta$) is $\PQMAEXP$-complete for 1D, nearest neighbour, translationally invariant Hamiltonians on spin chains of length $N$.
This is true even if $A$ is a $1$-local observable with two eigenvalues $x, y\in[0, 1]$, and the gap in expectation $b-a=|x-y|-1/\poly N$.
\end{theorem}
We give some exposition on how this can be adapted from \cite{Watson_Bausch_Gharibian_2020} in \cref{Sec:Prelims_2-CRT-PRM}.

\subsection{The Gottesman-Irani Hamiltonian and the Local Hamiltonian Problem}
The Local Hamiltonian problem asks for an approximation to the ground state energy of a local many-body system, up to $1/\poly$ precision in the system size.
In its non-translationally invariant version the 
local Hamiltonian problem was first proven to be \QMA-complete by \citeauthor{Kitaev2002} \cite{Kitaev2002}, and \QMAEXP-complete by \citeauthor{Gottesman_Irani_2009} \cite{Gottesman_Irani_2009}; over time, many variants of the problem have been shown to be hard, under various constraints on the systems of interest.
Formally, the ground state approximation problem can be stated as a promise problem in the following fashion.
\begin{definition}[Translationally-Invariant Local Hamiltonian Problem]
\begin{problem}
\probleminput{$N \in \field N$ and const-local term $h$ of a 
translationally-invariant Hamiltonian $H=\sum_{i=1}^m h_i$ on an $N$-partite Hilbert space of constant local dimension, and $m\le \poly N$ for some $N \in \field N$.
$\| h \| \le 1$.
Two numbers $\alpha,\beta>0$ with spacing $\beta-\alpha > 1/\poly N$.
$h$, $\alpha$ and $\beta$ are all given to $|\poly N|$ many bits of precision.
}
\problempromise{
Either $\lmin(H) \ge \beta$, or $\lmin(H) \le\alpha$.
}
\problemquestion{
\YES if  $\lmin(H) \le \alpha$.\\
\NO otherwise.
}
\end{problem}
\end{definition}

We will make direct use of \citeauthor{Gottesman_Irani_2009}'s constructive hardness proof showing that the translationally-invariant Local Hamiltonian problem is \QMAEXP-complete, in the sense of utilising the ability to encode quantum computation into the ground state of a Hamiltonian, and summarize their result in the following statement.
\begin{theorem} [\citeauthor{Gottesman_Irani_2009} \cite{Gottesman_Irani_2009}] \label{Theorem:Gottesman-Irani}
A Gottesman-Irani Hamiltonian $G_N\coloneqq\sum_{i=1}^N h_i$ is a translationally-invariant nearest neighbour Hamiltonian on a one-dimensional spin chain with finite local dimension $d\in\field N$, and with open boundary conditions; $G_N$ has the following properties:
\begin{enumerate}
    \item For all $N\ge 10$, the ground state of $G_N$ is a history state which encodes a binary counter with output $N-2$ in binary; and then takes this as input to a universal quantum Turing machine $\mathcal M$.
    Part of the input to $\mathcal M$ remains unconstrained.
    \item If $N$ describes a \QMA verifier and a valid problem instance for it---as e.g.\ shown in \cite[Fig.~11]{Bausch2016}---then there exist two polynomials $p,q$ such that $1/p(N)-1/q(N)=\Omega(1/\poly N)$, and
    \begin{align*}
        \lmin(G_N) \begin{cases}
        \geq 1/p(N) & \text{if $\mathcal M$ outputs \NO} \\
        \leq 1/q(N) & \text{if $\mathcal M$ outputs \YES}.
        \end{cases}
    \end{align*}
    Determining which case occurs is \QMAEXP-complete.
\end{enumerate}
\end{theorem}

The existence and construction of $G_N$ is a by-now standard technique; aside from \citeauthor{Gottesman_Irani_2009}'s original construction there also exists a lower local dimension variant with local dimension $d=42$ \cite{Bausch2016}.
Other nondeterministic computations can be encoded in a similar fashion, and we will make use of this fact in \cref{sec:2-hard} (cf.~\cite{Watson_Bausch_Gharibian_2020}).

\section{Main Results}\label{sec:main-results}
Now that we have introduced the necessary technical background, we can give rigorous statements of our two main results, \cref{th:main-1-intro,th:main-2-intro}.
\begin{theorem}
\checked{James}
$1$-CRT-PRM is \QMAEXP-hard and contained in $\PQMAEXP$ for Hamiltonians satisfying either \cref{Def:Local-Global_Phase} or \cref{Def:Local-Global_Gap}.
\end{theorem}
In fact, we prove a slight bit more than this; we show that $1$-CRT-PRM is \QMAEXP-hard, even if the parameter is only to be inferred to constant precision (as stated in the informal \cref{th:main-1-intro}).

\begin{theorem}\label{th:main-2}
\checked{James}
$2$-CRT-PRM is $\PQMAEXP$-complete for Hamiltonians satisfying either \cref{Def:Local-Global_Phase} or \cref{Def:Local-Global_Gap}.
\end{theorem}

Containment within some complexity class for the two cases hinges, as aforementioned, on how hard it is to answer the respective local-global properties; the corresponding reductions are proven in \cref{sec:containment}; for us, we will constrain the threshold system size for the local-global promises to be polynomial in the input size, which in both cases will result in a containment within \PQMAEXP.
Whether containment for the 1-CRT-PRM case can be made tighter (e.g.~prove containment within \QMAEXP) is an open question (see the extended discussion in \cref{sec:discussion}).
We prove the two hardness results in \cref{sec:1-hard,sec:2-hard}, respectively, by explicitly constructing translationally-invariant nearest neighbour families of Hamiltonians which encode the answer of a $\QMAEXP$ resp.~$\PQMAEXP$-hard problem within the decision problem of whether the system is in phase A or B.

Each Hamiltonian will be defined on a lattice $\Lambda$, $\{ H_N(\varphi) \}_{N\in\field N}$, where $H_N(\varphi) = \sum_{i\sim j} h^{(i,j)}_N(\varphi) + \sum_{i\in\Lambda} h^{(i)}_N(\varphi)$ for neighbouring sites $i\sim j$, $i,j\in\Lambda$.
All local terms $h^{(i)} \coloneqq h^1_{\{i\}} \otimes \1_{\Lambda\setminus\{i\}}$, and analogously for $h^{(i,j)}$ constructed from a constant two-site matrix $h^2$.
Furthermore, all the matrix entries of $h^1$ and $h^2$ will be specified to $O(|N|)=O( \log_2 N )$ bits of precision; and naturally we allow the local terms to depend on the family parameter $N$ in a trivially-computable fashion (i.e.\ we demand the matrix entries have to be computable classically from $N$ in time $\poly|N|$), where we do however require a constant-bounded interaction strength $\| h^1 \|,\| h^2 \| \le 1$.
As aforementioned, containment in \PQMAEXP is a corollary of the local-global phase/gappedness promise imposed on the systems.

\section{Containment of 1- and 2-CRT-PRM in P\textsuperscript{QMA\textsubscript{EXP}}}\label{sec:containment}

\begin{lemma}\label{lem:tech-containment-1}\label{Lemma:Approximate_Spectral_Gap}
\checked{Johannes, James}
Consider a Hamiltonian $H^{\Lambda(L)}(\varphi)$  such that the local terms are describable in $n$ bits, and that satisfies the global-local \textbf{gap} condition (\cref{Def:Local-Global_Gap}) .
Then for $n=O(\log(L))$ determining whether $\Delta(H^{\Lambda(L)}(\varphi))\geq 1/q(L)$ or $\leq 1/p(L)$ for any point in parameter space  is in $\PQMAEXP$.
\end{lemma}
\begin{proof}
The algorithm showing containment of SPECTRAL GAP (as defined in \cite{Ambainis2013}) in $ \PQMALOG $ can be used.
However, now the Hamiltonian is promised to be translationally invariant, we need exponentially less information to input the Hamiltonian; the only input is now $L$ which only requires $n=\BigO(\log(L))$ bits to express.
As the required precision, on the other hand, is still $1/\poly L$, the relevant complexity class is now  $\PQMAEXP$.
\end{proof}

\begin{lemma}\label{lem:tech-containment-2}\label{Lemma:Order_Parameter_Complexity}
\checked{James,Johannes}
Consider a Hamiltonian $H^{\Lambda(L)}(\varphi)$  such that the local terms are describable in $n$ bits, and that satisfies the global-local \textbf{phase} condition (\cref{Def:Local-Global_Phase}).
Then for $n=O(\log(L))$ and all states $\bra{\psi}H^{\Lambda(L)}(\varphi)\ket{\psi}\leq \lmin(H^{\Lambda(L)}(\varphi)) + \omega$  determining whether $\langle \bra{\psi}O_{A/B}\ket{\psi}\rangle\geq 1/q(L)$ or $\leq 1/p(L)$ for any point in parameter space such that $|\varphi - \varphi^*|\geq 1/\poly(L)$ is in $\PQMAEXP$.
\end{lemma}
\begin{proof}
Follows directly from \cref{Lemma:APX-SIM_Complexity} in the following way.
Set $\delta$ from the definition of $\forall$-TI-APX-SIM in \cref{Def:forall_APX-SIM} to be equal to the energy parameter $\omega$ from \cref{Def:Local-Global_Phase}.
Let the order parameter $O_{A/B}$ be the local observable to be measured (i.e. $A$ in the definition of $\forall$-TI-APX-SIM), and let the polynomials $p,q$ correspond to the bounds $a,b$.
Then finding $\bra{\psi}O_{A/B}\ket{\psi}$ for a state with energy $\bra{\psi}H^{\Lambda(L)}(\varphi)\ket{\psi}\leq \lmin(H^{\Lambda(L)}) +\omega $ is just an instance of $\forall$-TI-APX-SIM.
\end{proof}

\begin{theorem}\label{lem:tech-containment-3}
\checked{James}
Let $H$ be a translationally invariant Hamiltonian with local terms describable in $n$ bits, and on a lattice of size $L=\exp(\poly n))$. 
Further let $H$ satisfy either the global-local gap property in \cref{Def:Local-Global_Gap} or the global-local phase property in \cref{Def:Local-Global_Phase}, and have a critical point at $\varphi^*$.
Then the 1-CRT-PRM and 2-CRT-PRM are contained in $\PQMAEXP$. 
\end{theorem}
\begin{proof}
We start with the more general case.
\paragraph{Local-Global Phase Assumption.}
For 1-CRT-PRM, we must show it is possible to find an approximation $\tilde{\varphi}^*$ such that for $n=\BigO(\log(L))$
\[
        |\tilde{\varphi}^*-\varphi^*|< \BigO(1/\poly L) = \BigO(1/ \exp n ).
\]
with an algorithm in $\PQMAEXP$.

The algorithm is as follows:
\begin{itemize}
    \item Calculate $L_0$.
    By \cref{Def:Local-Global_Phase}, this can be calculated in $\poly(n)$ time.
    \item Take a $L_0\times L_0$ region of the lattice.
    Using the algorithm in \cref{lem:tech-containment-1}, determine $\langle O_{A/B} \rangle = \bra{\psi }H\ket{\psi}$ for states satisfying $\bra{\psi }H\ket{\psi}\leq \lmin(H)+\omega$ to precision $1/r(L_))$, where $r(L)\gg q(L) > p(L)$. 
    As per \cref{Lemma:Order_Parameter_Complexity}, this can be done using $\poly(n)$ many calls to a $\QMAEXP$ oracle.
    \item If $\langle O_{A/B} \rangle < 1/p(L) + 1/r(L) $ then by the global-local phase condition \cref{Def:Local-Global_Phase}, the Hamiltonian must be in phase $A$. 
    If $\langle O_{A/B} \rangle > 1/q(L) - 1/r(L) $, then we know it must globally be in phase $B$. 
    Due to the earlier promise, we are guaranteed that it satisfies one of these conditions.
    \item Perform a binary search through the parameter space of $\varphi$.
    Using $\BigO(\poly(n))=\BigO(\log(L))$ runs of the algorithm we can identify the point at which the order parameter from  $\langle O_{A/B} \rangle < 1/p(L) + 1/r(L) $ to $\langle O_{A/B} \rangle > 1/q(L) - 1/r(L) $ or visa-versa to within $O(1/\poly(L))$ precision.
    The interval in which the expectation $\langle O_{A/B} \rangle$ changes must contain the critical parameter $\varphi^*$.
    \item 
    The result is we get an estimate $\tilde{\varphi}^*$ such that 
    \[
        |\tilde{\varphi}^*-\varphi^*|<  \BigO(1/\poly L)  = \BigO(1/\exp n).
    \]
\end{itemize}
Finally we note that running the algorithm to compute the $\langle O_{A/B}\rangle$ takes $\poly(n)=\BigO(\log(L))$  \QMAEXP queries at each point in parameter space.
To do the binary search procedure, we choose $\poly(n)$ points, each of which runs this algorithm.
Thus, overall we make $\poly(n)=\BigO(\log(L))$ queries to the $\QMAEXP$ oracle and hence 1-CRT-PRM is in $\PQMAEXP$. 

For the 2-CRT-PRM case, the extra ingredient is the offset $y=\varphi^*$ along the $\theta=0$ axis in \cref{def:2-crt-prm}.
As we are promised that there is a unique such critical point, we can use binary search for the special case $H(0, \varphi)$ to approximate $y$ to precision $\BigO(1/\exp n)$ which takes at most $\poly( n )$ oracle queries.
Using \cref{def:2-crt-prm_2} as our definition of 2-CRT-PRM, we choose some $\varphi\in [y + \kappa, y+2\kappa]$ and query the order parameter for that point as above.
This can all be done in $\PQMAEXP$.

\paragraph{Spectral Gap Assumption.}
The proof for containment in the case that the Hamiltonian satisfies the global-local spectral gap condition is almost identical.
This is because the algorithm to determine the spectral gap to precision $\BigO(1/\exp n) = 1/\poly L$ precision is also contained in $\PQMAEXP$, as shown in \cref{Lemma:Approximate_Spectral_Gap}.
\end{proof}
\section{QMA\textsubscript{EXP} Hardness of 1-CRT-PRM}\label{sec:1-hard}

In order to prove \QMAEXP-hardness of 1-CRT-PRM, we explicitly construct a 1-parameter family of Hamiltonians $H_N(\varphi)$ on a qudit lattice $\Lambda$, where $N\in\field N$ and $\varphi\in[0,1]$ are encoded into phases of a local term, and with all matrix entries specified to bit precision of at most $|N|$.

The local terms $h^N(\varphi)$ will be a function of two parameters: $N$ and $\varphi$ (both encoded into the phase of a local term).
Here $N$ encodes the problem instance and should be thought of as changing the form of the Hamiltonian. This point is subtle, but important: in the thermodynamic limit, where the lattice size is infinite, the only term that encodes the instance is the local coupling terms $h^N(\varphi)$.
For every such $h^N$, we now ask: is the critical point $\varphi^*$ above or below some threshold? This is entirely analogous to the local Hamiltonian problem, where say the size of the spin chain $N$ determines the instance, and we ask whether the ground state is above or below some threshold. So what does $N$ encode? It is an integer that encodes the hard problem that we reduce to the phase decision; as such, not all $N$ might be valid (as the hard classes we use in the reduction are promise problems; as such, there might be invalid $N$ too for which we cannot say anything).
When making reductions between promise problems we need only show that if the promise of the initial problem is satisfied, then the promise of the problem we are mapping it to much also be satisfied. We will not be concerned with the case where $N$ corresponds to an invalid instance. 

In contrast to the LH Problem, or the question of whether a system is gapped or not in the thermodynamic limit, we \emph{know} by construction that the system will be in phase A or B for various choices of $\varphi$. This means $\varphi$ is the parameter of interest: we ask whether the critical point $\varphi^*$ is above or below some threshold---just as we could have asked for the spectral gap to be larger or smaller than some threshold, or the ground state energy for that manner---or the colour of the resulting material.
This, in turn, means that \emph{every instance, indexed by $N$, is itself a family of Hamiltonians, parameterised by $\varphi$, and 1-CRT-PRM asks questions about families of Hamiltonians.}
The complexity scaling will then be in terms of $|N|$, i.e.\ the number of bits required to encode $N$, as the input size, or precision to which we describe the matrix elements of the local terms.
This is entirely natural: we would expect the problem to be ever harder the more precise we specify the local terms; and the precision to which we want to resolve a critical point to be related to the precision to which the local Hamiltonian is specified as well.

\paragraph{Proof Outline.}
\checked{James}
$H_N(\varphi)$ is constructed so that its ground state  partitions the lattice into checker board grids of varying side length (motivated by the idea in \cite{Bausch_Cubitt_Watson2019}). 
Within each square of the checker board there is a QTM-to-Hamiltonian mapping, which means that its ground state is a so-called ``history state''; for the particular Turing machine we choose to encode, the ground state represents the following procedure:
\begin{enumerate}
	\item Perform QPE to extract $N$ from local terms.
	\item Perform a phase comparator QPE on the unitary encoding $\varphi$ and $\exp(\ii t G_N)$, where $G_N$ is a translationally-invariant local spin Hamiltonian with \QMAEXP-hard local Hamiltonian problem (on a spin chain of length $N$). This computation is performed with an unconstrained input state.
	Assuming said input state is an eigenstate of $G_N$ with eigenvalue $\lambda$, the phase comparator QPE extracts the \emph{difference} $\lambda-\varphi$ to bit precision $\sim |N|$.
	\item 
	If $\varphi<\lambda$ we set an output flag to $\ket{0}$.
	If $\varphi>\lambda$ we set it to $\ket{1}$.
\end{enumerate}

We can then give an energy penalty to the $\ket{0}$ state of the flag qubit at the output of this computation, which ensures that the hitherto unconstrained input state to the phase comparator QPE assumes $G_N$'s ground state; we can thus assume that $\lambda=\lmin(G_N)$.
Adding an unconditional bonus on the square size that the history computation runs on (using a 2D Marker Hamiltonian), this combination of history state and Marker Hamiltonian results in a ground space energy
\[
\lmin(H_N(\varphi)) = \begin{cases}
< 0 & \text{if $\lmin(G_N) < \varphi$, and} \\
\ge 0 & \text{otherwise.}
\end{cases}
\]
As $G_N$ satisfies a promise gap, there will also be a promise gap on $\lmin(H_N(\varphi))$ with respect to $\varphi$.

We will describe this construction outlined above in rigorous detail in the following sections; we start with the construction of the quantum Turing machine performing the listed operations in \cref{sec:phase-comparator-qtm}, and under the assumption of having access to all necessary gates directly.
In \cref{sec:phase-comparator-qtm-approx}, we do away with this assumption and replace the execution of $\exp(\ii t G_N)$ with a variant based on performing Hamiltonian simulation; as well as using Solovay-Kitaev to replace all remaining gates by a fixed gate set.

We then embed the resulting quantum Turing machine into a history state Hamiltonian in \cref{sec:phase-comparator-ham}, combine this construction with a 2D Marker Tiling in \cref{sec:with-marker}, lift the phase comparison ground state energy to a phase transition in \cref{sec:two-phases}, prove that our construction has a unique phase transition in \cref{sec:unique-critical-point}, and show the reduction $\QMAEXP \longrightarrow$ 1-CRT-QTM in \cref{sec:1-hardness-reduction}.
Finally, in \cref{sec:verify-promise-1}, we prove that this constructed family of Hamiltonians indeed satisfies the local-global gap property from \cref{Def:Local-Global_Gap}.

\subsection{A Phase Comparator Quantum Turing Machine}\label{sec:phase-comparator-qtm}
In this section we consider a multi-tape QTM which will take as input two different numbers $N\in\field N$ and $\varphi\in[0,1]$\footnote{Some of the previous theorems are stated for $\varphi\in[0,\poly(N)]$. We will actually prove a result for $\varphi\in[0,1]$ first and then extend it to a larger interval.} as well as an input state $\ket{\nu}\in (\C^d)^{\ox N}$.
Rather than straightforwardly inputting $N, \varphi$ on the tape, we will give the QTM access to particular gates such that when it performs quantum phase estimation on these gates, the resulting string will be an encoding of $N$ and $\varphi$.

As we require both numbers to be extracted from the complex phase of a unitary, we need to encode them into the fractional part of the phase in some fashion.
To this end, we devise the following encoding map for $N$:
\begin{definition}[Encoding]\label{def:encoding}
	\checked{James}
	Let $(\cdot\ ,\cdot)$ represent the string concatenation operation, let $k$ be a fixed integer, and let $n=|N|$.
	Let a string $w\in [4]^{2n} \equiv ([4]^n, [4]^{n})$ be \emph{valid} if $w\in ([2]^n,2^{\times n})$, and denote the set of all valid strings of length $2n$ as $V_{2n}$, where further $V \coloneqq \bigcup_{n=1}^\infty V_{2n}$.
	For $N\in \field N$, we define
	\vspace{-5mm}
	\[
	\enc:\field N \longrightarrow V
	\quad\text{where}\quad
	\enc(N) = (N_1N_2\cdots N_{n},\overbrace{22\cdots2}^{n\ \text{times}})
	\]
	where $N$ has binary expansion $N=N_1N_2 \cdots N_n$.
	For $z\in\field N$, we set $\enc^{-1}(z)$ to be the number $z$ truncated to the first half (rounded down) of the base-4 digits of $z$.
\end{definition}
We remark that the encoding is in base $4$ and hence the number of \textbf{bits} required for $\enc(N)$ is $|\enc(N)| = 4|N|$.
We note that for the inverse map in \cref{def:encoding}, we have that $\enc^{-1}\enc(N)=N$ for all $N\in\field N$, and such that if $z$ has $m$ base-4 digits (i.e.\ $m\le 4^m$), it always holds that $\enc^{-1}(z) \le 2^m$.

In this section we loosely speak of performing QPE to base four; this is of course simply a shorthand for performing base $2$ QPE with twice the number of bits; it is straightforward to verify that the following derivation is well-defined in this context.
In order to assess how well-suited said encoding is to get a precise handle on the QPE error, we formulate the following technical lemma: 
\begin{lemma}[Encoded QPE Extraction] \label{Lemma:String_Encoding}
	\checked{James}
	Denote with $V$ the set of all valid strings from \cref{def:encoding}.
	Consider the map
	\[
	V \longrightarrow \field R
	\quad\text{where}\quad
	V \ni y \longmapsto 0.y
	\quad\text{(in base 4)}
	\]
	and set $U_y  \coloneqq  \diag(\exp(\ii \pi 0.y), 1)$.
	Denote with $|y|$ the length of $y$ in base 4.
	Let the output of performing quantum phase estimation on the unitary $U_y$ (wrt.\ to the eigenstate $\ket0$) with a perfect gate set (i.e.\ all gates are performed without error) on $t\in 2\field N$, $t\ge2$ qudits be $\sum_{m\in[4]^t} \beta_m\ket{m}$.
	Then if $t \ge |y|$, $\beta_y=1$ and all other $\beta_i=0$;
	if $t<|y|$ the total probability amplitude on valid strings is bounded as
	\[
	\sum_{\substack{m\in V_t}}\left|\beta_m\right|^2<\frac{1}{2^{t/2}}.
	\]
\end{lemma}
\begin{proof}
	
	We consider QPE on $t$ qudits for two cases: $t<|y|$, $t\geq|y|$.
	The fact that we work with base 4 numbers bears no significance since we indirectly treat the setup as a base-2 QPE of twice the length, and it is useful to view the following proof through this lens.
	
	\paragraph{Case \paramath{t\geq|y|}.}
	In this case the quantum phase estimation can be done exactly; this means the probability amplitude on states which are not $y$ is precisely zero \cite[Sec.~5.2]{Nielsen_and_Chuang}.
	
	\paragraph{Case \paramath{t<|y|}.}
	In this case the QPE is not performed exactly and the output is some superposition clustered around the best $t$ digit estimates of $y$; the following analysis closely follows the QPE error analysis in \cite[Sec.~5.2]{Nielsen_and_Chuang}.
	As QPE is done in little Endian order, if we denote with $b\in[4]^t$ the string such that $b/4^t$ is the best $t$ digit approximation to $y$ \emph{less} than $y$, we know that $b$ is simply $y$ truncated on the right hand side to $t$ digits.
	We note that $b\not\in V_t$ as the truncation means that $b_{t/2} \neq 2$ (remember that we assumed $t$ even, so $t/2\in\field N$).
	
	Denote with $b'\in V_t$ the closest \emph{valid} string to $y$, in the same sense as $b$ (i.e.\ such that $b'/4^t$ is closest to $y$ amongst all $b'\in V_t$).
	Since $b'_{t/2} = 2$, we clearly have $| b - b' | \ge 4^{t/2}$, as the two strings have to differ by at least $1$ at the $(t/2)$\textsuperscript{th} position.
	
	By \cite[Eq.~5.27]{Nielsen_and_Chuang} we see that the probability of measuring an outcome $m$ further than $e\in\field N$ away from $b$ is bounded by
	\begin{align}\label{eq:state-away}
	p(|b-m|>e)\leq \frac{1}{2(e-1)}.
	\end{align}
	As $b'$ was the closest valid string to $y$, we know that all other valid strings $m\in V_t$ also satisfy $|b-m| \ge 4^{t/2}$.
	This means
	\[
	\sum_{m\in V_t} | b_m|^2 \le \sum_{m\in V_t} \frac{1}{2(4^{t/2}-1)} \le \frac{2^{t/2}}{2(4^{t/2}-1)} \le \frac{1}{2^{t/2}}
	\]
	as $|V_t| = 2^{t/2}$ by \cref{def:encoding}.
	The claim follows.
\end{proof}

QPE extracts the phase $\lambda$ of some unitary $U$ with respect to an eigenstate $\ket u$---i.e.\ such that $U\ket u = \ee^{\ii \lambda}\ket u$.
This means the algorithm assumes that there exists a ``black box'' capable of preparing the register $\ket u$ to be the correct eigenstate.
More often than not we do not have such a state: obtaining an eigenstate for an operator is generally at least as hard as estimating the associated eigenvalue.\footnote{Generally, it is as ``cheap'' to calculate $\bra u U \ket u$ as it is to write out $\ket u$, modulo polynomial overhead.}
On the other hand, this allows us to form a nondeterministic variant of QPE, in the sense that if we leave $\ket u$ unconstrained, we can ask questions like ``does there exist a state $\ket u$ for which $\lambda_u$ is less than a certain quantity?''
This notion of a nondeterministic QPE has been employed in different contexts before, e.g.\ in the context of Hamiltonian simulation \cite{Kohler2020}.

In the following lemma we will thus assume that the eigenstate $\ket u$ is an external quantity to be supplied to the procedure, where we keep in mind that we translate the algorithm to a history state Hamiltonian in due course.
As history state Hamiltonians allow for an unconstrained section (which can later be filtered by a suitable penalty addition), the above decision problem of existence of a state $\ket u$ for which $\lambda_u$ is below or above some threshold then maps naturally to the eigenspectrum of the Hamiltonian.
More concretely, as the particular unitary we are interested in performing nondeterministic QPE on is itself a local Hamiltonian with a hard ground state energy problem, the notion of existence/non-existence of an eigenstate $\ket u$ with a phase $\lambda_u$ below a certain threshold can be analysed precisely as in the case of encoding a (nondeterministic) QMA verifier in hardness proofs of the local Hamiltonian problem \cite{Kitaev2002}.

In the remainder of this section, we will leave the number of bits to which QPE is performed (generally called $t$ in the following) vs.~the length of the chain on which the Hamiltonian-to-be-analysed sits (generally called $N$, or an integer $z\le N$ depending on the context) independent; yet in order to prove hardness, we will later on require the number of bits large enough to resolve the promise gap of the encoded Hamiltonian to high enough precision.

In order to compare two phases extracted via QPE, there is two options: extract each phase individually and perform a binary comparison, or perform QPE on the first unitary and the second unitary's inverse and compare against $0$.
We opt for the latter, as it will be easier to prove that a \emph{single} critical point exists. 
Details in due course.

We now describe the QTM we will utilise for the rest of the 1-CRT-PRM hardness proof.

\begin{figure}
	\centering
	\begin{quantikz}
		\lstick{$\ket a$} & \qw & \gate[wires=2,bundle={2},disable auto height][3cm]{\shortstack{phase gradient\\of $U_a$}} & \qw  & \qw & \qw & \qw & \qw  \rstick{$\ket a$} \\
		\lstick{$\ket+^{\otimes t}$} & \qwbundle[alternate]{} & & \qwbundle[alternate]{} & \gate[wires=2,bundle={1},disable auto height][3cm]{\shortstack{phase gradient\\of $U_b^\dagger$}} & \qwbundle[alternate]{} & \gate[bundle={1}]{\text{QFT}_t^{-1}} & \qwbundle[alternate]{} \rstick{$\ket{\mathrm{out}}$} \\
		\lstick{$\ket b$} & \qw & \qw & \qw & & \qw & \qw & \qw & \rstick{$\ket b$}
	\end{quantikz}
	\caption{Phase comparator circuit. For two unitaries $U_a$ and $U_b$ with eigenstates $\ket a$, $\ket b$ and associated eigenvalues $\lambda_a$, $\lambda_b$, the output register $\ket{\mathrm{out}}$ contains a $t$-bit approximation to the phase $\lambda_a-\lambda_b$, as can be seen by writing out the phase gradient operations as given in \cite[Fig.~5.2]{Nielsen_and_Chuang}.}
	\label{fig:phase-comparator-circuit}
\end{figure}
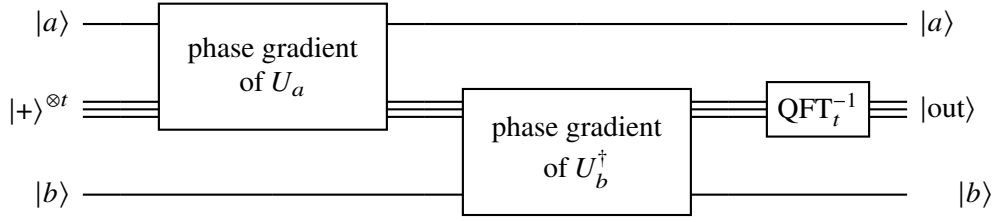

\newcommand\zp{{z'}}
\begin{lemma}[Multi-QPE QTM] \label{Lemma:QTM_Output}
	\checked{James}
	Let $G_z$ be a Gottesman-Irani Hamiltonian on a chain of length $z$, defined in \cref{Theorem:Gottesman-Irani}.
	Let $N\in \field N$.
	Denote by $U_\varphi, U_N\in SU(2)$ and $U_{G_z}$ the unitaries
	\[
	U_{G_z} = \ee^{\ii \pi G_z}
	\quad
	U_\varphi=
	\begin{pmatrix}
	\ee^{\ii\pi\varphi} & 0 \\
	0 & 1
	\end{pmatrix}
	\quad
	U_N=
	\begin{pmatrix}
	\ee^{\ii\pi 0.\enc(N)} & 0 \\
	0 & 1
	\end{pmatrix}
	\quad
	\]
	where $\enc(N)$ is given in \cref{def:encoding}.
	Let $\ket\nu \in (\field C^d)^{\otimes N}$.
	
	Then there exists a quantum Turing Machine, denoted $\mathcal{M}(N,\varphi,t,\ket\nu)$, with access to the unitary gates $U_\varphi$, $U_N$, and $U_{G_z}$ for all $z\in[N]$, all powers $2,4,\ldots,2^t$ of the $U_{G_z}$ gates, as well as $R_k \coloneqq \diag(1, \ee^{2\pi\ii / 2^k})$ for all $k=1,\ldots,t$ in addition to the standard gate set; such that $\mathcal M$ acts on a Hilbert space of $2 + 4t + t$ qubits and $N$ qudits, plus a slack space of size at most $\poly t$ (left implicit in the following); and such that $\mathcal M$ performs the following operations:
	\begin{enumerate}
		\item Initialise the first $2 + 4t+t$ registers to zero, and assume the last $N$ qudit registers are in state $\ket\nu$.
		\item Execute QPE on $U_N$ on $4t$ qubits to get a state
		\[
		\ket{\chi'} = \ket{00}_f \otimes \left( \sum_{z\in[4]^t} \gamma_z\ket{z} \right) \otimes
		\ket{0}^{\ox t}\otimes\ket\nu
		\]
		where the $\gamma_z$ are the amplitudes from quantum phase estimation of $U_N$. 
		\item For any basis state $\ket z$ in $\ket{\chi'}$ that is invalid, $\mathcal M$ places a marker  on the first qubit of the tape (the flag space, labelled with subscript $f$), such that the first qubit is flipped to 1 if $z\in V_t$.
		This gives
		\[
		\ket{\chi''} = \left( \sum_{z\in V_t} \ket{10}_f\gamma_z\ket{z} + \sum_{z\not \in V_t} \ket{00}_f \gamma_z\ket{z} \right) \otimes
		\ket{0}^{\ox t}\otimes\ket\nu
		\]
		
		\item Let $\zp \coloneqq \min\{ N, \enc^{-1}(z) \}$ as per \cref{def:encoding}, such that always $\zp \le 2^t$.
		On $\ket{\nu}$, the QTM performs a phase comparator QPE as shown in \cref{fig:phase-comparator-circuit} with the two unitaries $U_{G_\zp}$ (on the $\ket\nu$ register) and $U_{\varphi}^\dagger$ (on a $\ket0$ ancilla register).
		More concretely, for the Hamiltonian $G_\zp$, let its eigenstates be $\{\ket{g_\zp}\}$ and let $\ket\nu=\sum_g \kappa_{g}(\zp) \ket{g_\zp}\ket{\xi_{\zp,g}}$ be a decomposition with respect to a bipartition into $\zp$ and $N-\zp$ qubits; as we chose the basis of the first subsystem, the $\ket{\xi_{\zp,g}}$ are not necessarily orthogonal, but can be assumed normalised.
		Then the input to this step can be written as
		\[
		\ket{\chi''} = \sum_{z\in [4]^t} \gamma_z \ket{ [z\in V_t]0}_f\ket{z} \otimes\ket{0}^{\ox t}\otimes\sum_g \kappa_{g}(\zp)   \ket{g_\zp} \ket{\xi_{\zp,g}},
		\]
		where $[z\in V_t]$ is equal to 1 iff $z\in V_t$ and is otherwise 0 (cf.~\cref{sec:notation} for notation).
		The output of the total QTM after this stage is then
		\[
		\ket{\chi'''} = \sum_{z\in [4]^t} \gamma_z \ket{ [z\in V_t]0}_f\ket{z} \sum_{x\in[2]^t}\sum_g \alpha_{x}(\zp,g)\kappa_{g}(\zp) \ket x  \ket{g_\zp} \ket*{\xi_{\zp,g}}
		\]
		
		The $\alpha_{x}(\zp,g)$ are the coefficients obtained from the comparator QPE routine for operator $G_z$ and $U_\varphi$ on eigenstate $\ket{g_\zp}$.
		\item The flag qubit is updated via $\ket{b0}_f \longmapsto \ket{b[x \le 0]]}_f$,
		which corresponds to the comparison $\lmin(G_z) \le \varphi$ to $t$ digits of precision.
		The resulting state is
		\begin{align}
		\ket\chi =
		&\ \sum_{ z\in V_t} \sum_{x\le 0}\sum_{g} \alpha_{x}(\zp,g) \gamma_z \kappa_{g}(\zp) \ket{11}_f\ket{z}\ket{x}\ket{g_\zp}\ket{\xi_{\zp,g}} +  \label{eq:chi-with-flag}\\
		&\ \sum_{z\in V_t} \sum_{x > 0}\sum_{g} \alpha_{x}(\zp,g) \gamma_z \kappa_{g}(\zp) \ket{10}_f\ket{z}\ket{x}\ket{g_\zp}\ket{\xi_{\zp,g}} + \nonumber  \\
		&\ \sum_{ z\not\in V_t} \sum_{x \le 0} \sum_{g}  \alpha_{x}(\zp,g)  \gamma_z \kappa_{g}(\zp) \ket{01}_f\ket{z}\ket{x}\ket{g_\zp}\ket{\xi_{\zp,g}} + \nonumber \\
		&\ \sum_{ z\not\in V_t} \sum_{x > 0} \sum_{g}  \alpha_{x}(\zp,g)  \gamma_z \kappa_{g}(\zp) \ket{00}_f\ket{z}\ket{x}\ket{g_\zp}\ket{\xi_{\zp,g}}. \nonumber
		\end{align}
	\end{enumerate}
	The QTM $\mathcal M$ runs for time $T=\BigO(2^{4t})$.
\end{lemma}
\begin{proof}
	This QTM can be implemented by dovetailing a set of QTMs which perform QPE on $4t$ qubits for $U_N$, as well as $t$ qubits for  $U_{G_\zp}$ and $U_\varphi^\dagger$, where the last one implementing the conditional QPE for $G_\zp$ can be trivially implemented by adding additional control lanes.
	QPE without gate approximation to $t$ bits of precision normally takes $\BigO(t^2)$ calls to $U_{G_{\zp}}$ (which we assumed to have access to as a single gate for now) \cite[Sec.~5.2]{Nielsen_and_Chuang}.
	Implementing the up to $2^{4t}$-powers of the three phase gates $U_N$, $U_{G_\zp}$ and $ U_\varphi^\dagger$ takes time $\propto 2^{4t}$.
	All other operations take time $\poly t$, and it is clear that the QTM does not need in excess of $\poly t$ of slack work space. Hence we have an overall runtime of $\BigO(2^{4t})$.
\end{proof}


For the next set of lemmas we define the following quantity: for the output state $\ket\chi$  of $\mathcal{M}(N,\varphi,t)$ from \cref{Lemma:QTM_Output}, we set
\begin{align}
\eta(N,\varphi,t,\ket{\nu})  \coloneqq&\ \Tr\left(  (\ketbra{11}_f\otimes \1) \ketbra{\chi}   \right) \nonumber \\
=& \sum_{z\in V_t} \sum_{x \le 0}\sum_{g}|\alpha_{x}(\zp,g)|^2|\gamma_z|^2 |\kappa_{g}(\zp)|^2
\label{eq:eta}
\end{align}
where as in \cref{Lemma:QTM_Output} we have $\zp=\min\{ N, \enc^{-1}(z) \}$.
This is the total probability that $\ket\chi$ will have an accepted flag---i.e.\ the first qubit will be in the $\ket{11}_f$ state as given in \cref{eq:chi-with-flag}.
It will later be shown that when the above QTM is encoded in a circuit-to-Hamiltonian mapping, then the ground state energy depends on $\eta$.

\newcommand\bbar[1]{\bar{\bar{#1}}}
\begin{figure}[t]
	\centering
	\includegraphics[width=12cm]{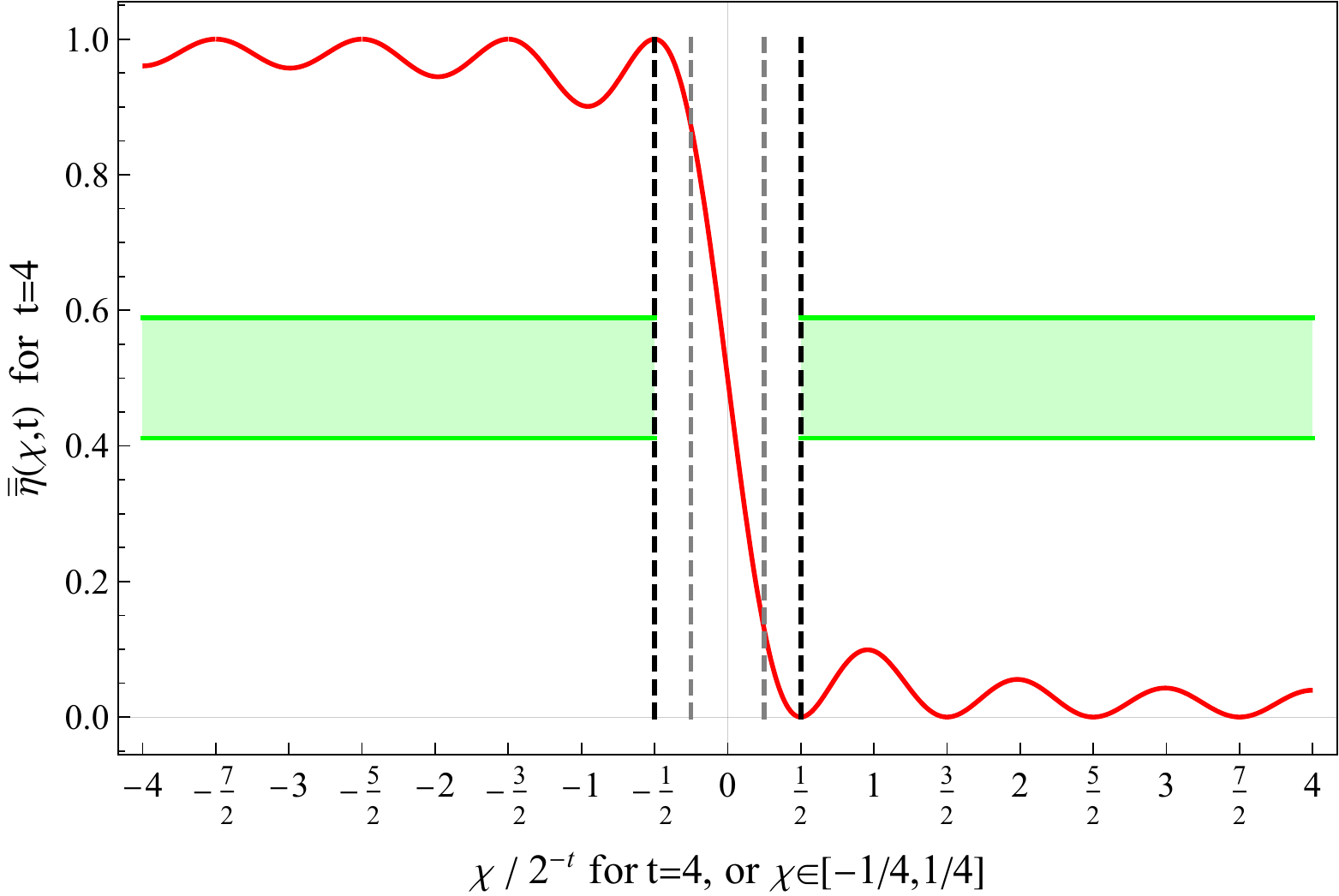}
	\caption{$\bbar\eta(\chi,t)$ (red line) for $t=4$ from \cref{lem:eta-monotonous}, vs.~$\chi\in(-1/4,1/4)$.
		The interval between the black dashed vertical lines denote the region within which we prove $\bbar\eta$ to be strictly monotonously falling, and hence $\eta$ from \cref{eq:eta} to be strictly monotonously increasing; within the grey dashed areas the slope of the red line is $\ge1$, as shown in \cref{lem:eta-monotonous}.
		The shaded green region marks the interval $[1/3, 2/3]$, which \cref{lem:eta-monotonous} proves $\bbar\eta(\chi,t)$ to be bounded away from.}
	\label{fig:eta-monotonous}
\end{figure}
As a first step, we present a monotonicity argument for the value of $\eta$ around the point where $\varphi$ equals the eigenvalue associated to an eigenstate $\ket\nu$ of $G_N$.

\begin{lemma}\label{lem:eta-monotonous}
	\checked{James}
	Let $\ket\nu$ be an eigenvector of $G_N$ with eigenvalue $\lambda$, $t\geq |N|$, and let $\eta(N,\varphi,t, \ket{\nu})$ be defined as in \cref{eq:eta}.
	Then for $t\geq |N|$ and $\varphi\in (\lambda-2^{-t}+2^{-3t/2},\lambda-2^{-3t/2})$
	\[
	\frac{\partial \eta(N,\varphi,t, \ket{\nu})}{\partial \varphi}  \ge 1.
	\]
	Furthermore, for all $\varphi < \lambda-2^{-t}+2^{-3t/2}$, $\eta \leq \pi^2/24$ and $\varphi > \lambda - 2^{-3t/2} $, $\eta \geq 1- \pi^2/24$.
	
\end{lemma}
\begin{proof}
	We guide the reader to \cref{fig:eta-monotonous} to aid in an intuitive understanding of the proof.
	As a first step, we abbreviate $\xi\coloneqq\lambda-\varphi$.
	By \cite[Eq.~5.26]{Nielsen_and_Chuang} and \cref{eq:eta}, and relabelling $L=x$ to follow the notation in \cite{Nielsen_and_Chuang} closely, we can write
	\begin{align}
	\eta(N,\varphi,t,\ket \nu) =&\ \sum_{z\in V_t} \sum_{x \le 0}\sum_{g}|\alpha_{x}(\zp,g)|^2|\gamma_z|^2 |\kappa_{g}(\zp)|^2 
	\overset*= \!\!\!\!\sum_{L=-2^{t-1}+1}^{0}\!\!\!\! |\alpha_L|^2  \nonumber\\
	\text{for}\quad
	\alpha_L \coloneqq&\ 2^{-t} \frac{1-\exp(2\pi\ii (2^t\xi - L))}{1-\exp(2\pi\ii(\xi-2^{-t}L))}
	= 2^{-t} \frac{\sin(2^t\pi\xi)}{\sin(\pi(\xi-2^{-t}L))}
	\label{eq:alpha_L}
	\end{align}
	where in the step marked with $\overset*=$ we have used the fact that for $t\ge|N|$, precisely one of the $\gamma_z$ and $\kappa_g(z')$ equal $1$, and all others are zero; so $\alpha_L \equiv \alpha_x(z',g)$ for those $z$ and $g$ for which $\gamma_z=\kappa_g(z')=1$.
	We set $\bar\eta(\xi,t) \coloneqq \eta(N, \varphi, t, \ket{\nu})$.
	As
	\begin{align}\label{eq:alpha_L-norm}
	\bar\eta(\xi, t) = \sum_{L \le 0} |\alpha_L|^2
	= 1 - \sum_{L > 0} |\alpha_L|^2,
	\end{align}
	we can calculate the midpoint where $\bar\eta(\xi,t)= 1/2$; this happens at $\xi = 2^{-t-1}$, as can be confirmed by explicit calculation.
	
	Set $\chi \coloneqq \xi -2^{-t-1}$ and define $\bbar \eta(\chi,t) \coloneqq \bar\eta(\chi + 2^{-t-1},t)$, such that $\bbar\eta(0,t)=1/2$, and
	\begin{align}
	\bbar\eta(-\chi,t) 
	&= \bar\eta(-\chi + 2^{-t-1},t) 
	= \smashoperator{\sum_{L=-2^{t-1}+1}^0} 2^{-2t} \frac{\sin^2(2^t\pi(-\chi + 2^{-t-1}))}{\sin^2(\pi((-\chi + 2^{-t-1})-2^{-t}L))}  \nonumber\\
	&\overset*= \smashoperator{\sum_{L=-2^{t-1}+1}^0} 2^{-2t} \frac{\sin^2(2^t\pi(-\chi - 2^{-t-1} + 2^{-t}))}{\sin^2(\pi(-\chi - 2^{-t-1} - 2^{-t}(L-1)))}  \nonumber\\
	&= \smashoperator{\sum_{L=-2^{t-1}+1}^0} 2^{-2t} \frac{\sin^2(2^t\pi(\chi + 2^{-t-1}) - \pi)}{\sin^2(\pi(\chi + 2^{-t-1} - 2^{-t}(1 - L)))}  \nonumber\\
	&= \sum_{L=1}^{2^{t-1}} 2^{-2t} \frac{\sin^2(2^t\pi(\chi + 2^{-t-1}))}{\sin^2(\pi(\chi + 2^{-t-1} - 2^{-t}L))}  \nonumber\\
	&\overset{**}{=} 1 - \bbar\eta(\chi,t),
	\label{eq:bbar-eta-symm}
	\end{align}
	where in the line with $\overset*=$ we added $2^{-t-1}-2^{-t-1}=0$ in the enumerator and denominator, and in the last step $\overset{**}{=}$ we made use of \cref{eq:alpha_L-norm}.
	
	\paragraph{Gradient Bound:}
	\checked{James}
	Note that we can now write:
	\[
	|\alpha_L|^2 = 2^{-2t} \frac{\cos^2(2^t\pi \chi)}{\sin^2(\pi(\chi-2^{-t-1}(2L-1)))}.
	\]
	From the above we see that:
	\begin{align*} 
	\bbar\eta(\chi,t) &= \smashoperator{\sum_{L= -2^{t-1}+1}^{0}} 2^{-2t} \frac{\cos^2(2^t\pi \chi)}{\sin^2(\pi(\chi-2^{-t-1}(2L-1)))}  \\
	&= \smashoperator{\sum_{L= 0}^{2^{t-1}-1}} 2^{-2t} \frac{\cos^2(2^t\pi \chi)}{\sin^2(\pi(\chi+2^{-t-1}(2L+1)))} = \smashoperator{\sum_{L= 0}^{2^{t-1}-1}} |\alpha_{-L}|^2,
	\end{align*}
	where we have just relabelled $L\rightarrow -L$.
	Differentiating this expression wrt.~$\chi$ gives
	\begin{align*}
	2^{2t}\frac{\partial \bbar\eta(\chi,t)}{\partial \chi} = \smashoperator{\sum_{L= 0}^{2^{t-1}-1}}\ \  \bigg[ &-\frac{2^{t+1}\pi \sin(2^t\pi \chi)\cos(2^t\pi \chi)}{\sin^2(\pi(\chi + 2^{-t-1}(2L+1)))} \\
	&-2\pi \cos(\pi(\chi + 2^{-t-1}(2L+1)))\frac{\cos^2(2^t\pi \chi)}{\sin^3(\pi(\chi + 2^{-t-1}(2L+1)))}\bigg].
	\end{align*}
	Now note that for $t\ge1$, $\chi\in (0,2^{-t-1})$ and $0\leq L\leq 2^{t-1}-1$ we have that $2^t\pi \chi \leq \pi/2$ and $0 \leq \pi(\chi + 2^{-t-1}(2L+1))\leq \pi/2$.
	Thus all of the sine and cosine terms in the above expression are individually positive, and hence both the terms in the summand are individually negative, for all $0 \leq L\leq 2^{t-1}-1$.
	
	We now focus on the $0$\textsuperscript{th} coefficient, i.e.\ $\alpha_0$.
	For $0\leq x\leq \pi/2$ the following holds:  $1/\sin(x) \geq 1/x$,  thus giving:
	\begin{equation}
	\begin{split}
	2^{2t} \left|\frac{\partial|\alpha_{0}|^2}{\partial \chi} \right| \geq&   \frac{2^{t+1}\pi \sin(2^t\pi \chi)\cos(2^t\pi \chi)}{(\pi(\chi + 2^{-t-1}))^2} \\
	+& \frac{2\pi\cos^2(2^t\pi \chi)}{|\pi(\chi + 2^{-t-1})|^3}\cos(\pi(\chi + 2^{-t-1}))  \label{Eq:partialEta_Derivative} 
	\end{split}
	\end{equation}
	Now we consider the interval $\chi \in (0, 2^{-t-1}-2^{-3t/2})$.
	Within this interval, and for $t\ge 1$, $\cos(2^t\pi\chi)$ has its minimum value at the rightmost limit, $\cos^2(\pi/2 - \pi2^{-t/2})=\sin^2(\pi2^{-t/2}) \ge 2^{-t}$.
	Similarly, we have that
	\[
	2^{-3t/2}\pi \leq \pi\left|\chi + 2^{-t-1}\right|\leq 2^{-t-1}\pi,
	\]
	and hence
	$\cos(\pi(\chi + 2^{-t-1}) \ge \cos(\pi / 4) \ge 1/\sqrt 2$.
	Together with \cref{Eq:partialEta_Derivative} and dropping the term with $\sin(2^t\pi\chi)$ in the denominator (which vanishes for $\chi\rightarrow0$),
	the $0$\textsuperscript{th} coefficient $\alpha_0$ thus satisfies
	\begin{align}
	2^{2t}\left|\frac{\partial|\alpha_0|^2}{\partial \chi} \right| 
	\geq& \frac{2\pi 2^{-t}}{(\pi 2^{-t-1})^3} \times \frac{1}{\sqrt 2} = \frac{8\sqrt2}{\pi^2} \times 2^{2t} \ge 2^{2t}, \label{Eq:partialEta_Scaling_bounds}
	\end{align}
	Hence $\partial|\alpha_0|^2/\partial \chi \le - 1$.
	Since $\partial|\alpha_L|^2/\partial \chi <0$ for all $0 \leq L\leq 2^{t-1}-1$, then $\partial|\alpha_0|^2/\partial \chi > \partial\bbar \eta/\partial \chi$ for $\chi \in (0,2^{-t-1} - 2^{-3t/2})$.
	Thus
	\[
	\frac{\partial \bbar\eta(\chi,t)}{\partial \chi} \le - 1.
	\]
	Using the antisymmetry of the function $\bbar\eta(\chi,t)-1/2$ around the point $\chi=0$, we see the same bounds on the derivative hold for the whole interval $\chi\in (-2^{-t-1}+2^{-3t/2},2^{t-1}-2^{-3t/2})$.
	Finally, noting that $\partial \chi /\partial \varphi=-1$
	\[
	\frac{\partial\eta(N,t,\varphi,\ket{\nu})}{\partial \varphi} \ge 1,
	\]
	for the interval $\varphi \in (\lambda-2^{-t}+2^{-3t/2},\lambda-2^{-3t/2})$.
	
	\paragraph{Bounds Outside Interval.}
	\checked{James}
	To address the bounds when $\chi$ is outside of the monotonicity interval $\chi \in [-1/4, 1/4] \setminus (-2^{-t}/2+2^{-3t/2},2^{-t}/2-2^{-3t/2})$, we again by \cref{eq:bbar-eta-symm} we only have to consider the right half of the interval. There we have
	\begin{align}
	\bbar\eta(\chi,t) &= \sum_{L\leq 0} |\alpha_L|^2
	\overset*= \smashoperator{\sum_{L= -2^{t-1}+1}^{0}} 2^{-2t} \frac{\cos^2(2^t\pi \chi)}{\sin^2(\pi(\chi+2^{-t-1} - 2^{-t}(L+1)))} \nonumber\\
	&\le \sum_{L=0}^{2^{t-1}-1} \frac{2^{-2t}}{\sin^2(\pi(\chi + 2^{-t-1} + 2^{-t}(L-1)))} \nonumber\\
	&= \sum_{L=1}^{2^{t-1}} \frac{2^{-2t}}{\sin^2(\pi(\chi + 2^{-t-1} + 2^{-t}L))}.
	\label{eq:sinineq1}
	\end{align}
	where in the step $\overset*=$ we again added and subtracted $2^{-t-1}$ in the denominator.
	For $\chi \in [2^{-t-1}-2^{-3t/2}, 1/4]$ and $L=1, \ldots, 2^{t-1}$, we can bound
	\[
	\pi (\chi + 2^{-t-1}+2^{-t}L) \ge \pi2^{-t}(L+1)-\pi2^{-3t/2}
	\]
	and thus
	\begin{equation}\label{eq:sinineq2}
	\sin(\pi(\chi + 2^{-t-1} + 2^{-t}L)) \ge \sin(\pi(2^{-t}(L+1)-2^{-3t/2})) \ge \frac12 \times \pi (2^{-t}(L+1) -2^{-3t/2}).
	\end{equation}
	Combining \cref{eq:sinineq1,eq:sinineq2}, we get
	\begin{align}
	\bbar\eta(\chi,t)
	\le& \frac{4}{\pi^2} \sum_{L=1}^{2^{t-1}} \frac{1}{ (L+1 - 2^{-t/2})^2} \nonumber \\
	\leq&  \frac{4}{\pi^2} \sum_{L=2}^\infty \frac{1}{L^2}\left(1 - \frac{2^{-t/2}}{L}\right)^{-2} \nonumber \\ 
	=&  \frac{4}{\pi^2} \sum_{L=2}^\infty \frac{1}{L^2}\left(1 + \frac{2\times 2^{-t/2}}{L}+ \BigO\left( \frac{2^{-t}}{L^2}\right)\right) \label{Eq:bbar_eta_1}\\
	=& \frac{4}{\pi^2} \sum_{L=2}^\infty \frac{1}{L^2} + \frac{8}{\pi^2}\sum_{L=2}^\infty\frac{ 2^{-t/2}}{L^3}+ \BigO\left( 2^{-t}\right) \label{Eq:bbar_eta_2}\\
	=&\frac{4}{\pi^2} \times \left(\frac{\pi^2}{6} -1\right) + \BigO(2^{-t/2}) \label{Eq:bbar_eta_3} \\ 
	\le& \frac{\pi^2}{24}. \nonumber
	\end{align}
	Here for \cref{Eq:bbar_eta_1} we have used a binomial expansion and for \cref{Eq:bbar_eta_3} we have used the well known identity $\sum_{n=1}^\infty n^{-2}=\pi^2/6$.
	The bound for negative $\chi$ follows by \cref{eq:bbar-eta-symm}.
\end{proof}

\Cref{lem:eta-monotonous} puts bounds on $\eta$ for a specific input state $\ket{\nu}$.
Here we consider the maximum value $\eta$ can take: this corresponds to the maximum acceptance probability that $\mathcal{M}$ can have when the input state $\ket{\nu}$ is unconstrained. 
\begin{corollary}\label{cor:good-QPE-probs}
	\checked{James}
	Let $\eta$ be as defined in \cref{eq:eta}, and $|N|$ be the number of base-2 digits\footnote{$N$ is expressed in binary, but $\enc(N)$ in quaternary, hence $|\enc(N)|=4|N|$; note that we expand to $4t$ bits in \cref{Lemma:QTM_Output}, hence the statement ``$t\ge |N|$'' implies we expanded enough digits to see all of $\enc(N)$ exactly.} of $N$.
	\begin{itemize}
		\item If $t \ge |N|$ and if $\varphi\leq \lmin(G_N) - 2^{-t}+2^{-3t/2}$, we have
		\[
		\max_{\ket{\nu}}\eta(N,\varphi,t,\ket{\nu})\leq \frac{\pi^2}{24}.
		\]
		\item If $t \ge |N|$ and $\varphi\geq \lmin(G_N) - 2^{-3t/2} $, we get
		\[
		\max_{\ket{\nu}}\eta(N,\varphi,t,\ket{\nu}) \ge 1-\frac{\pi^2}{24}.
		\]
		\item If $t<|N|$, then irrespective of the value of $\varphi$,
		\[
		\max_{\ket{\nu}}\eta(N,\varphi,t,\ket{\nu})=\BigO\left(\frac{1}{2^{t/2}}\right).
		\]
	\end{itemize}
\end{corollary}
\begin{proof}
	We address each case individually.
	
	\paragraph{Case \paramath{t<|N|}.}
	We do not expand enough bits to expand $\enc(N)$ in full: by \cref{Lemma:String_Encoding},
	the probability mass on valid strings (none of which are $\enc(N)$) is $\leq 1/2^{t/2}$.
	
	\paragraph{Case \paramath{t\ge |N|}.} 
	Again by \cref{Lemma:String_Encoding} we know that $\gamma_{\enc(N)}=1$ and all other $\gamma_z=0$.
	If $\varphi\geq \lmin(G_N)-2^{-3t/2}$, then choose input state $\ket{\psi_0}=\ket{g_\mathrm{min}}\ket{\xi_0}$ where $\ket{g_\mathrm{min}}$ is the ground state of $G_N$ (and $\ket{\xi_0}$ is just the state resulting from the bipartition of $\ket{\psi_0}$'s state space in the proof of \cref{Lemma:QTM_Output}).
	By applying \cref{lem:eta-monotonous}, we see that 
	\[
	\eta(N,\varphi,t,\ket{\psi_0}) \geq 1- \frac{\pi^2}{24}.
	\]
	The result for $\varphi\geq \lmin(G_N)-2^{-3t/2}$ follows.
	
	If $\varphi\leq \lmin(G_N) - 2^{-t}+2^{-3t/2}$, then consider any eigenstate $\ket{\psi_i}$ of $G_N$. 
	We see that if $\varphi\leq \lmin(G_N) - 2^{-t}+2^{-3t/2}$, then for any $\ket{\psi_g}=\ket{g}\ket{\xi_g}$, where $\ket{g}$ is an eigenstate of $G_N$ with corresponding eigenvalue $\lambda_g$, $\varphi\leq \lambda_g - 2^{-t}+2^{-3t/2}$. 
	As a result, for any energy eigenstate $\ket{\psi_g}$, by \cref{lem:eta-monotonous},
	\[
	\eta(N,\varphi,t,\ket{\psi_i}) \leq \frac{\pi^2}{24}.
	\]
	Any state $\ket{\nu}\in (\C^d)^{\ox N}$ can be written as  
	$\ket\nu=\sum_g \kappa_{g}(N) \ket{g_N}\ket{\xi_{N,g}}$; hence by convexity of \cref{eq:eta} in the coefficients of $\ket\nu$ we have
	\[
	\max_{\ket{\nu}}\eta(N,\varphi,t,\ket{\nu}) \leq  \frac{\pi^2}{24}.
	\qedhere
	\]
\end{proof}


\subsection{An Approximate Phase Comparator QTM}\label{sec:phase-comparator-qtm-approx}

So far we have assumed that the QTMs can implement the algorithms without error, by providing them with all the necessary gates required.
In this section we relax these assumptions and show that the error in the output is bounded sufficiently small for our purposes, even if the QTM only has access to a fixed universal gate set.\footnote{The motivation for this is that when the QTM is encoded in a Hamiltonian, in order to have a fixed local Hilbert space dimension for all $N, \varphi$, and $t$, the QTM must have a gate set which does not depend on $N, \varphi$ or $t$. }

If we wish our QTM to have a fixed predetermined number of gates available for an arbitrary track length $t$ and arbitrary length inputs $N, \varphi$, the gate powers of $U_{G_\zp}$, as well as the controlled rotations $R_k$ necessary for the Fourier transform subroutine for QPE in \cref{Lemma:QTM_Output} cannot be given explicitly; we need to approximate them.
The $R_k = \diag(1, \ee^{-\ii\pi2^k}) $ gates can be approximated via the Solovay-Kitaev algorithm to the necessary precision.
For $U_{G_\zp}$, however, such a simple compilation argument does not work: we need to perform Hamiltonian simulation in order to implement $U_{G_\zp}$ itself, for any given spin chain length $\zp\in[N]$; to this end, we include the following result.

\begin{lemma}[Hamiltonian Simulation QTM]\label{Lemma:HamSim-QTM}
	\checked{James}
	Let $S  \coloneqq  \{ h^{(l)} \}_{l\in I}$ be a constant and finite set of local interactions of a translationally-invariant Hamiltonian $H=\sum_{i=1}^z h_i$, $h_i\in S$, defined on a spin chain of length $z\in\field N$, and let $\epsilon>0$.
	Then there exists a QTM which, on input $S$ and $z$, simulates the time evolution $U(T) \coloneqq \exp(\ii H T)$ as a circuit $\tilde U(T)$ to precision $\| \tilde U(T) - U(T) \| \le \epsilon$ in spectral norm, in time $\Tilde\BigO(T^2 z^2/\epsilon)$.\footnote{As per convention, $\tilde\BigO$ hides polylogarithmic factors in the argument.}
\end{lemma}
\begin{proof}
	This is a straightforward application of a second order Trotter formula (see e.g.\ \cite{Childs2019}).
	We first assume that we can implement the local time evolution operators $U_i(T) = \exp(\ii T h_i)$ for $H=\sum_i h_i$ exactly.
	For time $\delta>0$, a second order Trotter formula (e.g. \cite{Berry_2006}) breaks up :
	\[
	\tilde U(\delta) = \prod_{i=1}^z U_i(\delta/2) \prod_{i=z}^1 U_i(\delta/2),
	\]
	which requires $\BigO(z)$ gates, and has an error bound
	\[
	\| \tilde U(\delta) - U(\delta) \| = \BigO\left( z \delta^2 \right).
	\]
	As we require a simulation to time $T$, the overall error will be
	\[
	\| \tilde U(T) - U(T) \| \le \frac T\delta \| \tilde U(\delta) - U(\delta) \| = \BigO\left( T \delta z \right),
	\]
	requiring $\BigO(z T/\delta)$ gates.
	Demanding $T \delta z \le \epsilon$ means $\delta = \Theta(\epsilon/Tz)$; the Trotter simulation thus requires $\Theta(z^2 T^2 / \epsilon)$ gates overall.
	
	In case we cannot implement the local Trotter operators $U_i(\delta)$ exactly, we have to approximate them with a sequence of elementary gates $\tilde U_i(\delta)$, e.g.\ using Solovay-Kitaev \cite{Dawson_Nielsen_2005}: as errors in a circuit accumulate at most linearly, the approximation has to be precise to
	\[
	\| \tilde U_i(\delta) - U_i(\delta) \| = \epsilon / \Theta(z^2 T^2 / \epsilon) = \Theta(\epsilon^2 / z^2 T^2) =: \epsilon'
	\]
	where $\epsilon'$ is now the precision we must approximate each $\tilde U_i(\delta)$ .
	Now we know that Solovay-Kitaev can approximate any unitary operation to within precision $\epsilon'$ within $\BigO(\log^4 1/\epsilon')$ many steps.
	The claim follows.
\end{proof}

Since we will be encoding our QTMs in a Hamiltonian where part of the ground state is chosen in a nondeterministic manner, approximating gates leads not only to slight errors in the gates, but also since the gates no longer have the same eigenvalues, it may lead to differences between which eigenvalue is nondeterministically chosen when the unconstrained input state $\ket{\nu}$ is nondeterministically chosen. 
In the following two lemmas we characterise this error.

\begin{lemma}\label{Lemma:Close_Evolution_Hamiltonian}
	\checked{James}
	Let $V(s)=\ee^{\ii H s}$ for a Hamiltonian $H$ such $\| H \|_{\infty}\leq \pi/4s$, and let $\tilde{V}(s)$ satisfy
	$\| V(s)-\tilde{V}(s) \|_{\infty}\leq \epsilon$.
	Then there exist an effective Hamiltonian $H'$ such that $\tilde{V}(s)=\ee^{\ii H's}$, and a constant $\kappa=\BigO(1)$ such that
	\[
	\| H'-H \|_{\infty}\leq \frac{\kappa \epsilon}{s}.
	\]
\end{lemma}
\begin{proof}
	See supplementary information of \cite{Poulin_Wocjan_2009}.
\end{proof}

The following lemma shows that even when the Hamiltonian simulation regime is used (to sufficient accuracy) the new value of $\eta$ maximised over all input states, has similar bounds to the value obtained without Hamiltonian simulation.
\begin{lemma}[Hamiltonian Simulation Error]\label{Lemma:hamSim_Error}
	\checked{James}
	Let $\mathcal M(N,\varphi,t,\ket{\nu})$ be the QTM described in \cref{Lemma:QTM_Output} with all gates done without error.
	Then there exists a QTM ${\mathcal M}'(N,\varphi,t,\ket{\nu})$ performing the same algorithm, except where the phase estimation for $U_{G_\zp}$ is instead performed by a Hamiltonian simulation algorithm in \cref{Lemma:HamSim-QTM}, and such that $\mathcal M'$ satisfies the following:
	\begin{enumerate}
		\item Let $\eta'(N,\varphi,t,\ket{\nu})$ be defined in the same way as $\eta(N,\varphi,t, \ket{\nu})$ from \cref{eq:eta}, but corresponding to the output of ${\mathcal M}'$.
		The following bounds are satisfied:
		\begin{itemize}
			\item If $t \ge |N|$ and if $\varphi\leq \lmin(G_N) - 2^{-t}+\BigO(2^{-3t/2})$, we have
			\[
			\max_{\ket{\nu}}\eta'(N,\varphi,t,\ket{\nu})\leq \frac{\pi^2}{24}.
			\]
			\item If $t \ge |N|$ and $\varphi\geq \lmin(G_N)-\BigO(2^{-3t/2}) $, we get
			\[
			\max_{\ket{\nu}}\eta'(N,\varphi,t,\ket{\nu})\geq 1-\frac{\pi^2}{24}.
			\]
			\item If $t<|N|$, then irrespective of the value of $\varphi$,
			\[
			\max_{\ket{\nu}}\eta'(N,\varphi,t,\ket{\nu})=\BigO\left(\frac{1}{2^{t/2}}\right).
			\]
		\end{itemize}
		\item The runtime overhead relative to $\mathcal M$ is at most $\poly(N, 2^{t})$.
	\end{enumerate}
\end{lemma}
\begin{proof}
	We first note that the largest power of $\ee^{\ii \pi G_N}$ to be performed in the QPE in \cref{Lemma:QTM_Output} is $T=2^{t}$. 
	Writing $G'_N \coloneqq 4G_N / \pi N$, we have
	\[
	\ee^{\ii \pi G_N T} = \ee^{\ii G'_N \pi^2NT/4} = \left( \ee^{\ii G'_N} \right)^{\pi^2 NT/4}.
	\]
	We set $V(s) \coloneqq \ee^{\ii G'_N s}$, and note that $V(s)$ satisfies the conditions of \cref{Lemma:Close_Evolution_Hamiltonian} for $s\le 1$.
	Let $H'$ be the effective Hamiltonian generated in the Hamiltonian simulation scheme given in \cref{Lemma:HamSim-QTM} for the short pulse $V(1)$, such that $\tilde V(s) = \ee^{\ii H' s}$ and $\| V(s) - \tilde V(s) \|_\infty \le \epsilon'$ is the precision to which we want to perform Hamiltonian simulation of the small time step $V(1)$; we leave $\epsilon'$ implicit for now, and determine its scaling in due course.
	By \cref{Lemma:Close_Evolution_Hamiltonian} it holds that $\| H' - G'_N \|_\infty \le \kappa\epsilon'$ for some constant $\kappa$.
	Set $H'' \coloneqq \frac{\pi N}{4} H'$, then
	\[
	\left\|  H'' - G_N \right\|_\infty
	=\left\| \frac{\pi N}{4} H' - \frac{\pi N}{4} G'_N \right\|
	\le \frac{\pi N}{4}  \kappa\epsilon'
	\le \frac{\kappa\epsilon}{\pi T} 
	\]
	where we have chosen
	\[
	\epsilon' \le \frac{4\epsilon}{\pi^2 NT},
	\]
	and consequently
	\begin{align*}
	|\lmin(\pi G_N)-\lmin(\pi H'')| &\le \kappa\epsilon/ T \\
	|\lmin(2\pi G_N)-\lmin(2\pi H'')| &\le 2\kappa\epsilon/ T \\
	&\ \ \vdots\\
	|\lmin(\pi T G_N)-\lmin(\pi T H'')| &\le \kappa\epsilon.
	\end{align*}
	This immediately implies that the deviation for QPE even in the highest Endian bit is upper-bounded by $\kappa\epsilon$.
	
	Let $\tilde U(T) \coloneqq \tilde V(1)^{\pi^2 NT/4}$.
	Then by an iterative expansion we have
	\begin{align*}
	\| U(T) - \tilde U(T) \|_\infty &= \| V(1)^{\pi^2 NT/4} - \tilde V(1)^{\pi^2NT/N} \|_\infty \\
	&\le \frac{\pi^2NT}{4} \| V(1) - \tilde V(1) \|_\infty \\
	&\le \frac{\pi^2NT}{4} \epsilon' = \epsilon.
	\end{align*}
	Now we would like this deviation of QPE to be \emph{less} than the smallest digit of precision of the exact QPE, which is satisfied for $\epsilon = \operatorname{o}(2^{-2t})$.
	We thus know that if $\varphi>\lmin(G_\zp)-2^{-3t/2}$ or $\varphi<\lmin(G_\zp)-2^{-t}+2^{-3t/2}$, then $\varphi>\lmin(H'')-\kappa\epsilon/\pi$ or $\varphi<\lmin(G_\zp)-2^{-t} + 2^{-3t/2} + \kappa\epsilon/\pi$, respectively.
	Thus, by \cref{lem:eta-monotonous} and \cref{cor:good-QPE-probs}, the same bounds as in  \cref{cor:good-QPE-probs} hold if we choose $\epsilon=\BigO(2^{-2t})$.
	
	The runtime overhead is then determined by an outer loop of applying $V(s)$ $\pi^2 NT / 4 = \poly 2^t$ times, and by the cost of approximating $V(s)$ to precision $\epsilon'$ in spectral norm, which by \cref{Lemma:HamSim-QTM} takes $\tilde\BigO(N^2/\epsilon') = \BigO(N^3 T / \epsilon) = \BigO(N^3 2^{3t} t)$, which is $\BigO(\poly(N, 2^t)$.
	The claim follows.
\end{proof}

We now fully characterise the output of the QTM when non-determinism and approximate gate sets are taken into account.
That is, the QTM only has access to a standard universal gate set and the gates $U_\varphi$, $U_{N}$.
\begin{lemma}[Gate Approximation Error]\label{Lemma:SK_HamSim_Approximation}
	\checked{James}
	Let $\mathcal M'(N,\varphi,t, \ket{\nu})$ be the QTM described in \cref{Lemma:hamSim_Error} with the Hamiltonian simulation subroutine performed, but all other gates still done exactly.
	Then there exists a QTM $\mathcal M''(N,\varphi,t,\ket{\nu})$ that satisfies the following.
	\begin{enumerate}
		\item $\mathcal M''$ only has access to a fixed universal gate set and the gates $U_\varphi$, $U_{N}$.
		\item Let $\eta''(N,\varphi,t,\ket{\nu})$ be defined in the same way as $\eta(N,\varphi,t, \ket{\nu})$ from \cref{eq:eta}, but corresponding to the output of $\mathcal M''$.
		Then $\max_{\ket{\nu}}\eta''(N,\varphi,t,\ket{\nu})$ satisfies the same bounds as $\max_{\ket{\nu}}\eta'(N,\varphi,t, \ket{\nu})$ in \cref{Lemma:hamSim_Error}.
		\item The additional runtime overhead relative to $\mathcal M'$ is at most a factor $\poly \log(N, 2^{t})$.
	\end{enumerate}
\end{lemma}
\begin{proof}
	By \cref{Lemma:HamSim-QTM}, we already know that the Hamiltonian simulation subroutine utilises a fixed gate set.
	We approximate all other gates (apart from $U_N$ and $U_\varphi$, which are given explicitly)---of which there are at most $\#g \coloneqq \poly(N,2^t)$ many by combining the runtime of $\mathcal M$ from \cref{Lemma:QTM_Output} and runtime overhead of $\mathcal M'$ from \cref{Lemma:hamSim_Error}.
	We know that in order to approximate a const-local gate $U$ to precision $\epsilon$ using Soloay-Kitaev, an overhead $\BigO(\log^4 1/\epsilon)$ is introduced.
	Choosing $\epsilon = 1/\poly(N, 2^{2t})$ small enough such that $\#g\epsilon = 2^{-2t}/t$ suffices to satisfy both claims.
\end{proof}

We need one final property of $\eta''$: we will need to show that it is monotonically increasing within a certain region.
This doesn't follow straightforwardly from the monotonicity of $\eta$ proved in \cref{lem:eta-monotonous} as the approximation using Solovay-Kitaev and the Hamiltonian simulation routine may change the gradient $\partial \eta''/\partial \varphi$ to negative at some point. We show this is not the case for a sufficiently large precision.

\begin{lemma}[Approximate $\eta''$ Monotonicity] \label{lem:eta-solovay}
	\checked{James}
	Let $\ket\nu$ be an eigenvector of $G_N$ with eigenvalue $\lambda$, and let $\eta''(N,\varphi,t, \ket{\nu})$ be defined as in \cref{Lemma:SK_HamSim_Approximation}.
	Then, provided the gates are approximated to precision $\epsilon\le 2^{-2t}$,
	$\eta''(N,\varphi,t, \ket{\nu})$ is monotonically increasing for $\varphi \in \left( \lambda - 2^{-t}+\BigO(2^{-3t/2}), \lambda - \BigO(2^{-3t/2}) \right)$.
\end{lemma}
\begin{proof}
	\newcommand\Upg{U_\mathrm{PG}}
	We can express $\eta = \bra\psi U(\varphi) \ket\phi$ for some initial states $\ket\phi$, and some appropriate state $\ket\psi$, respectively, and $U(\varphi) = U_1 \Upg(\varphi) U_2$, where $\Upg(\varphi)$ denotes the phase gradient part dependent on $\varphi$, and $U_1$ and $U_2$ collect all the unitary operations before and after (and are independent of $\varphi$).
	We then have
	\[
	\frac{\partial\eta}{\partial\varphi} = \bra\psi U_1 \left(\frac{\partial\Upg(\varphi)}{\partial\varphi}\right) U_2 \ket\phi.
	\]
	This means the deviation is entirely dependent on the gradient of the phase gradient matrix (as expected).
	
	Denote with $\tilde U_1, \tilde U_2$ the circuits $U_1$ and $U_2$ circuits with the  Solovay-Kitaev approximation used for any gates which cannot be performed exactly, which we assume to be implemented to accuracy $\epsilon$, i.e.~such that $\| U_i - \tilde U_i \|_\infty \le \epsilon$ for $i=1,2$. We emphasise our assumption that no gate in $\Upg$ has to be approximated.
	Analogously to before, we then have
	\[
	\frac{\partial\tilde\eta}{\partial\varphi} = \bra\psi \tilde U_1 \left(\frac{\partial\Upg(\varphi)}{\partial\varphi}\right) \tilde U_2 \ket\phi.
	\]
	Consequently,
	\begin{align}
	\left| \frac{\partial\tilde\eta}{\partial\varphi} - \frac{\partial\eta}{\partial\varphi} \right|
	&\le \left| \Tr\left(  \frac{\partial\Upg(\varphi)}{\partial\varphi} \left(\tilde U_1 \ketbra{\phi}{\psi}\tilde U_2 - U_1 \ketbra{\phi}{\psi} U_2 \right) \right) \right|  \nonumber\\
	&\le \left\| \frac{\partial\Upg(\varphi)}{\partial\varphi} \right\|_\infty \left\| \tilde U_1 \tilde U_2 - U_1 U_2 \right\|_\infty.
	\label{eq:eta-derivatives-bound}
	\end{align}
	The dependence of $\Upg(\varphi)$ only comes from the controlled $U_\varphi$ operations, and $\Upg(\varphi)=cU_\varphi U_\varphi^2\dots cU_\varphi^{2^t-1}$ ($\Upg(\varphi)$ actually has a set of $t$ Hadamards, however, we can absorb these into $U_1$ for convenience).
	All controlled gate powers within the phase gradient circuit are of the form
	\[
	\diag\left(1, 1, 1, \exp\left( 2\pi \ii \varphi 2^k \right) \right)
	\quad\text{with derivative}\quad
	2\pi\ii \times 2^k \diag\left(0, 0, 0, \exp\left( 2\pi\ii \varphi 2^k \right) \right).
	\]
	Thus using the product formula we can write
	\begin{align}
	\left\| \frac{\partial\Upg(\varphi)}{\partial\varphi} \right\|_\infty \le\ &\left\| \frac{\partial cU_\varphi}{\partial\varphi} cU_\varphi^2\dots cU_\varphi^{2^t-1} \right\|_\infty \nonumber \\
	+&\left\|   cU_\varphi \frac{\partial cU_\varphi^2}{\partial\varphi} cU_\varphi^{2^2} \dots cU_\varphi^{2^t-1} \right\|_\infty  \nonumber\\
	&\quad \quad \quad \quad  \vdots \nonumber \\
	+&\left\|   cU_\varphi cU_\varphi^2 \dots \frac{\partial cU_\varphi^{2^t-1}}{\partial\varphi}   \right\|_\infty \nonumber\\
	\leq t &\left\| \frac{\partial cU_\varphi^{2^t-1}}{\partial\varphi} 
	\right\|_\infty \nonumber
	\end{align}

	
	From this standard product formula for the derivative of a sequence of gates and using \cref{eq:eta-derivatives-bound} this means
	\[
	\left| \frac{\partial\tilde\eta}{\partial\varphi} - \frac{\partial\eta}{\partial\varphi} \right|
	\le 2\pi \times t 2^t \times 2\epsilon.
	\]
	By \cref{lem:eta-monotonous}, $\partial\eta / \partial\varphi$ within the interval $\varphi\in (\lambda-2^t+ \BigO(2^{3t/2}), \lambda- \BigO(2^{3t/2}) )$ is $\ge 1$; it thus suffices to demand $2\pi \times t2^t \times 2\epsilon \le 2^{-t}$; a choice of $\epsilon = 2^{-2t}$ proves the claim.
\end{proof}

The results within this section then culminate in a result proving that important properties which hold for $\eta$ hold for the approximated version $\eta''$; in particular montonicity in a particular interval and bounds on the energy outside of this interval.
\begin{theorem}[Phase Comparator QTM]\label{th:phase-comparator-qtm}
	\checked{James}
	Let $N\in\field N$, $\varphi\in[0,1]$.
	For any Gottesman-Irani Hamiltonian $G_N$ there exists a quantum Turing machine $\tilde{\mathcal M}(N,\varphi,t, \ket{\nu})$ with access to special gates $U_N$ and $U_\varphi$ as in \cref{Lemma:QTM_Output} which, on input $t\in\field N$, $\ket{\nu}\in (\C^d)^{\ox N}$, and in time $\poly(2^t, N)$ produces an output state as $\mathcal M''$ in \cref{Lemma:SK_HamSim_Approximation}.
	Abbreviating
	\begin{align}\label{eq:eta-tilde}
	\tilde\eta(N,\varphi,t) \coloneqq  \max_{\ket{\nu}}\eta''(N,\varphi,t, \ket{\nu}),
	\end{align}
	where $\eta''(N,\varphi,t,\ket{\nu})$ is defined in in \cref{Lemma:SK_HamSim_Approximation}, we have that
	\[
	\tilde\eta(N,\varphi,t)  \begin{cases}
	\geq 1-\frac{\pi^2}{24} & t\ge |N| \ \text{and} \ \varphi \ge \lmin(G_N) -\BigO(2^{-3t/2}) \\
	\leq \frac{\pi^2}{24} & t\ge |N| \ \text{and} \ \varphi \le \lmin(G_N) - 2^{-t}+\BigO(2^{-3t/2}) \\
	=\BigO(2^{-t/2}) & t < |N|.
	\end{cases}
	\]
	Furthermore, $\tilde\eta(N,\varphi,t)$ is monotonically increasing in the interval
	\[
	\varphi\in \big[ \lmin(G_N) -\BigO(2^{-3t/2})\,, \lmin(G_N) - 2^{-t}+\BigO(2^{-3t/2}) \big].
	\]
\end{theorem}
\begin{proof}
	We take $\mathcal M''$ from \cref{Lemma:SK_HamSim_Approximation}, and leave the input for the $\ket\nu$ section unconstrained.
	The rest follows by \cref{cor:good-QPE-probs,Lemma:SK_HamSim_Approximation}.
	The fact $\tilde\eta$ is monotonically increasing in the given region is proven in   \cref{lem:eta-solovay}.
\end{proof}

One fact that we have glossed over is that we can assume that the QTM in \cref{th:phase-comparator-qtm} is well-formed as defined in \cite[Def.~3.3]{Bernstein1997}---as its evolution is trivially unitary---and well-behaved as in \cite[Def.~3.12]{Bernstein1997}; the latter condition simply means that the QTM halts in a final state such that the halting head state is in the same cell.

Moreover, we can assume further ``good'' properties we wish: unidirectionality (meaning each state is only ever entered from one direction) which \citeauthor{Bernstein1997} show can be simulated if not originally present (\cite[Lem.~5.5]{Bernstein1997}).

\subsection{A Phase Comparator History State Hamiltonian}\label{sec:phase-comparator-ham}
In this section we translate the QTM designed in \cref{sec:phase-comparator-qtm} into a history state Hamiltonian.
This technique is by now standard \cite{Kitaev2002,Gottesman_Irani_2009} and used ubiquitously throughout literature (see e.g.\ \cite[Sec.~4.1]{Bausch2018c}, \cite[Sec.~1]{Bausch_Crosson2016}, or \cite[Sec.~3.1]{Bausch_Cubitt_Watson2019} for an overview).
We start with the following refinement regarding standard form Hamiltonians (those with an initial and final penalty at the start and end of the computation).

\begin{theorem}[Adaptation of Theorem 3.4 from \cite{Watson_2019}]\label{Theorem:Precise_Energies}
	\checked{James}
	Let $H(\mu)$ be standard-from Hamiltonian with minimum output penalty $\mu$ on the final time step, such that the encoded QTM has runtime $T$. 
	Then if $\mu = \frac{k}{256T}$ for $0 \leq k\leq 1$, the following bound holds:
	\[
	\frac{0.99k}{256T^2} \leq \lmin\left(H\left(\frac{k}{256T}\right)\right)\leq \frac{1.05k}{256T^2}
	\]
\end{theorem}

Following from this, we take the arguably shortest rigorous route, and directly formulate the following theorem.

\begin{theorem}[Phase Comparator Hamiltonian]\label{Theorem:QTM_in_local_Hamiltonian}
	\checked{James}
	Let $N\in\field N$, and $\varphi\in[0,1]$.
	For any Gottesman-Irani Hamiltonian $G_N$
	there exists a constant $d>0$, and Hermitian operators $h^{(1)}\in\mathcal B(\field C^d)$, $h^{(2)}\in\mathcal B(\field C^d\times\field C^d)$, such that
	\begin{enumerate}
		\item $h^{(1)},h^{(2)}\ge 0$, with matrix entries in $\field{Z}$.
		\item $h^{(2)}=A+\ee^{\ii\pi\varphi} B + \ee^{-\ii\pi\varphi} B^\dagger +  \ee^{\ii\pi 0.\enc(N) }C + \ee^{-\ii\pi 0.\enc(N)}C^\dagger$, where
		\begin{itemize}
			\item $B,C \in \mathcal{B}(\C^d)$ with coefficients in $\field{Z}$, and
			\item $ A\in \mathcal{B}(\C^d)$ is Hermitian and with coefficients in $\field{Z}+\field{Z}/\sqrt 2+\ee^{\ii\pi/4}\field Z$.
		\end{itemize}
	\end{enumerate}
	Define a translationally-invariant nearest-neighbour Hamiltonian on a spin chain of length $L$ via
	\[
	\HTM(L)  \coloneqq  \sum_{i=1}^L h^{(1)}_i + \sum_{i=1}^{L-1} h^{(2)}_{i,i+1}.
	\]
	Denote with $\ket*{\blacksquare}$ and $\ket*{\midend}$ two special basis states of $\field C^d$, and for $m\in\field N$, denote the \emph{bracketed} subspace
	\begin{equation}\label{eq:bracketed}
	\Sbr(m)  \coloneqq  \ket*{\midend} \ox (\field C^d)^{\ox m} \ox \ket*{\blacksquare} \ox (\field C^d)^{\ox (L-m)} \ket*{\midend}.
	\end{equation}
	Then $\HTM(L)$ has the following properties.
	\begin{enumerate}\addtocounter{enumi}{2}
		\item $\HTM(L) = \bigoplus_{m=1}^{L-1} H(L,m) \oplus R$, where $H(L,m)  \coloneqq  \HTM(L)|_{\Sbr(m)}$; i.e.\ $\HTM(L)$ is block-diagonal with respect to the subspaces spanned by $\Sbr(m)$, and $R$ captures the remaining block.
		\item $R \ge 1$.
		\item $\lmin(H(L,m)) \ge 1$ if $m=0,1$.
		\item
		There exist $L_N = \poly N$ and $m_N = \poly\log_2 N$ and an integer constant $b$, such that the ground state energy $\lmin(H(L,m))$ of the other blocks satisfies
		{\small
			\[
			\hspace{-1cm}\lmin(H(L,m)) \begin{cases}
			\leq \frac{1.05}{256L^b}\frac{\pi^2}{24} & (m,L)=(m_N,L_N) \land \varphi \ge \lmin(G_N)-\BigO(N^{-6C}) \\
			\geq \frac{0.99}{256L^b}\left(1 - \frac{\pi^2}{24}\right) & (m,L)=(m_N,L_N) \land \varphi \le \lmin(G_N) - N^{-4C}+ \BigO(N^{-6C}) \\
			\geq \frac{0.99}{256L^b}\left(1 - \frac{\pi^2}{24}\right) & m < |N| \lor (m,L) \neq (m_N, L_N),
			\end{cases}
			\]}
		where $T(L)=L^{b/2}$ is the runtime of the encoded computation.
		\item If $(m,L)=(m_N,L_N)$, then $\lmin(H(L,m))$ is monotonically decreasing with $\varphi$ for \label{Point:Monotonicity} 
		\[
		\varphi\in \big[\lmin(G_N) - N^{-4C}+ \BigO(N^{-6C})\,, \lmin(G_N)-\BigO(N^{-6C})\big].
		\]
		
	\end{enumerate}
\end{theorem}
\begin{proof}
	A QTM  can be translated into a $2$-local quantum Thue system \cite{Bausch2016}; its associated Hamiltonian is then a $2$-local nearest-neighbour translationally-invariant Hamiltonian.
	Classical QTS transition rules yield integer matrix entries; the only nontrivial matrix entries stem from transition rules involving quantum letters (i.e.\ those that label sites where the quantum state is encoded, and transitions between those are unitary)---which in turn are simply the quantum gates available to the computation we will encode; as the phase gradient gates $U_N$, $U_\varphi$, and a universal gate set comprising CNOT, Hadamard and T will suffice for our purposes, the first two claims follow.
	
	It is furthermore clear that one can statically penalise all but the bracketed configurations, such that $R\ge 1$, and all bracketed states in $\Sbr(0)$ and $\Sbr(1)$ (see \cite[Sec.~5]{Gottesman_Irani_2009} for how this can be done).
	Moreover, the QTM we will construct will treat $\ket\blacksquare$ as a passive state, i.e.\ it will never move the block; as such, on the bracketed states themselves, the Hamiltonian is block-diagonal with respect to the position of the $\ket{\blacksquare}$ marker state, which is assumed to be at distance $m$ along the spin chain.
	The next three claims follow.
	
	Let $N^{-C}$ be the promise gap of the local Hamiltonian problem $\lmin(G_N)$ for some constant $C\in\field N$.\footnote{We can always shrink the promise gap to obtain this scaling for some integer $C$.}
	In order to resolve this promise gap with the given precision of $t$ bits, we would require
	\begin{equation}\label{eq:m_N}
	2^{-t} \le N^{-C}
	\quad\Longleftarrow\quad t \ge 2\times 4C \lceil \log_2 N \rceil \eqqcolon m_N
	\end{equation}
	where the extra factor of $4$ was added such that at least four times as many bits than necessary are resolved.
	The computation we encode then performs the following steps.
	\begin{enumerate}
		\item Translate the segment length $L$ into binary onto a track, and do the same with $m$.
		\item Perform $\tilde{\mathcal M}$ from \cref{th:phase-comparator-qtm}, using $t=m$ as precision input.
		\item 
		Verify that
		\begin{enumerate}
			\item $L=L_N\coloneqq 2 + 4t + t  + N$, and
			\item $t=m=m_N$ as defined in \cref{eq:m_N}, and
			\item $t$ is large enough such that $\tilde\eta(N,\varphi,t) \le \pi^2/24$ even in case $t<|N|$ (i.e.~$\BigO(2^{-t/2}) \le \pi^2/24$ in \cref{th:phase-comparator-qtm}).
		\end{enumerate}
		If any of the above conditions do not hold, set a penalty flag.
	\end{enumerate}
	It is a standard exercise to ensure that in all computational branches the size of the history state (i.e.\ the length of the computation) has the same length; this is usually done by introducing a global clock and idling steps (we point the reader to \cite[Sec.~4]{Cubitt_Perez-Garcia_Wolf2015}). 
	We can assume that the runtime of all of the above computation is precisely $T(L)=L^{b/2}$ for some even integer constant $b>0$.
	
	We assume our history state Hamiltonian features a single in- and output penalty, as in \cite[Fig.~2\&Sec.~5]{Bausch_Cubitt_Watson2019}; it is straightforward to show that $H(L,m)$ is also standard form as in \cite[Sec.~5]{Watson_2019}.
	We choose the output penalty to penalise the complement of the accepting subspace defined by the projector $\Pi = (\1-\ketbra{11}_f)$ on the final time step of the computation, as well as the case where $L\neq L_N$ (this can both be done locally for standard form Hamiltonians); the penalty will have strength $1/256T^2=1/256L^b$.
	Since we do not want our Hamiltonian to have another explicit dependence on $L$, we remark that this can be done by rotating an ancilla
	\begin{align}\label{Eq:Qubit_Rotation}
	\ket 0_a \longmapsto \delta \ket0_a + \sqrt{1-\delta^2} \ket 1_a
	\quad\text{for}\quad
	\delta = \frac{1}{256L^b}
	\end{align}
	as $L$ and $b$ are both known; then the penalty can be conditioned onto $\ketbra 0_a\otimes\Pi \eqqcolon \Pi'$.
	Then for a valid history state $\ket\chi$ and by using \cref{Theorem:Precise_Energies}, we have that
	\[
	\Tr[\ketbra{\chi} \Pi' ]=\frac{1}{256T^2}(1-E(L,N,\varphi)),
	\]
	where $E(L,N,\varphi)$ is the weight on the accepting subspace of the computation, which is the product of $\tilde\eta$ (i.e.\ $\tilde{\mathcal M}$'s output) and the test that $L=L_N$ and $m=|N|$.\footnote{We note that the last test can fail to produce the right result if $t=m$ was too small to begin with to expand enough bits of $N$; but in this case, the output of the QTM $\tilde{\mathcal M}$ already asserts a small acceptance probability, by \cref{th:phase-comparator-qtm}.}
	
	The case where $m=t$ is too short to expand $N$ in full is captured by the output of the QTM $\tilde{\mathcal M}$, i.e.\ by \cref{th:phase-comparator-qtm} we have that $\tilde\eta(N,\varphi,t) = \BigO(2^{-t/2})$ in this case, and hence also $E(L,N,\varphi)=\BigO(2^{-t/2})$.
	
	Let us thus focus on the case when $t$ is large enough (i.e.\ $t\ge |N|$).
	If $L\neq L_N$ or $m \neq m_N$, $E(L,N,\varphi)=0$ by construction, so we only need to analyse the remaining case of $L=L_N$ and $m=m_N$.
	By \cref{th:phase-comparator-qtm} the accepting state overlap is then
	\[
	\tilde\eta(N,\varphi,t)  \begin{cases}
	\ge 1-\frac{\pi^2}{24}  &  \varphi \ge \lmin(G_N) - \BigO(N^{-6C})  \\
	\le \frac{\pi^2}{24}  &  \varphi \le \lmin(G_N) - N^{-4C}+\BigO(N^{-6C}).
	\end{cases}
	\]
	where we made use of the fact that $2^{-t} = 2^{-m_N} = N^{-4C}$ and $2^{-3t/2}=2^{-3m_N/2}=N^{-6C}$,
	and thus overall
	{\small
		\begin{equation}\label{eq:bds-22}
		E(L,N,\varphi) \begin{cases}
		\ge 1-\frac{\pi^2}{24}  & (L,m)=(L_N,m_N) \land \varphi \ge \lmin(G_N) - \BigO(N^{-6C})  \\
		\le \frac{\pi^2}{24}  & (L,m)=(L_N,m_N) \land  \varphi \le \lmin(G_N) - N^{-4C}+\BigO(N^{-6C}) \\
		\le \frac{\pi^2}{24} & t=m<|N| \lor (L,m) \neq (L_N,m_N).
		\end{cases}
		\end{equation}}
	The bounds in the statement then follow from combining \cref{eq:bds-22} with \cref{Theorem:Precise_Energies}.
	
	Finally, to prove $\lmin(H(L,m))$ is monotonically decreasing with $\varphi$, we note that $1-\tilde\eta(N,\varphi,t)$ is monotonically decreasing ( since \cref{th:phase-comparator-qtm} shows $\tilde\eta(N,\varphi,t)$ is monotonically increasing).  
	Since the Hamiltonian is standard form, this acts as a penalty of the form $\frac{1}{256T}(1-\tilde\eta(N,\varphi,t))$, which is monotonically decreasing, and which can be shown by standard techniques to be equivalent to adding on a positive semi-definite projector \cite{Watson_2019}.
	As adding a positive semi-definite matrix to another matrix can never lead to a \emph{decrease} in the combined eigenvalues, the claim follows.
\end{proof}

\subsection{Combining the Comparator Hamiltonian with a 2D Marker Tiling}\label{sec:with-marker}
We import the 2D Marker Hamiltonian from \cite[Sec.~7]{Bausch_Cubitt_Watson2019}, which describes a checkerboard pattern for which each checkerboard square of size $L\times L$, with a special marker offset at position $m<L$ on one of the edges, has a net negative energy contribution $\propto 1/4^{f(L,m)}$.

To do this, we introduce a special marker state $\ket{\star}$\footnote{Not to be confused with the bracketing state $\ket{\blacksquare}$ in \cref{eq:bracketed}; and note we also re-use the letter $m$ here to indicate the offset of $\ket\star$. This is not necessarily the same offset as the offset for $\ket\blacksquare$.} which interacts with the Marker Hamiltonian.
The function $f$ is then defined by the placement of $\ket{\star}$, and so we can use the placement of  $\ket{\star}$ (controlled by a classical tiling pattern within each of the squares) to define $f$ to have the appropriate properties for our proof.

We paraphrase the following result, tightening the bounds on the Marker Hamiltonian's ground state energy as we go.
\begin{theorem}[{\cite[Th.~7.6]{Bausch_Cubitt_Watson2019}}]\label{Theorem:Marker_Energy}
	\checked{Johannes}
	Let $\HS=(\field C^d)^{\ox \Lambda}$ be a square spin lattice $\Lambda$ with spins of dimension $d$, and let $c>0$.
	Further let $f(L)$ be a function such that $f(L) \le L$ is an integer and computable in time and space $\leq kL$, for some constant $k\in \mathbb{N}$.
	Then there exists a translationally-invariant nearest-neighbour Hamiltonian $\HM=\sum_{\langle i,j\rangle} h_{i,j}^{(\boxplus,f)}$ with the following properties:
	\begin{enumerate}
		\item $\HM = \bigoplus_{L>0} \HM(L) \oplus R'$.
		\item $R'\ge 0$.
		\item $\HM(L)$ has a unique ground state corresponding to a checkerboard tiling.
		Let $\HM(L)|_S$ be the restriction of the Hamiltonian to a single checkerboard square, then for all $L\ge 2$,
		\begin{align}\label{eq:mhambounds}
		-\frac{9/4}{4^{f(L)}} \le \lmin(\HM(L)|_S) \le -\frac{9/4-9/4^{f(L)}}{4^{f(L)}}.
		\end{align}
	\end{enumerate}
\end{theorem}
\begin{proof}
	We construct an augmented checkerboard tiling as in \cite[Sec.~6]{Bausch_Cubitt_Watson2019} which places the a special marker $\ket{\star}$ offset at $f(L)$, using a classical tiling within the checkerboard square; as $f(L)$ was computable within time and space $\le kL$ for some constant $k\in\field N$ the existence of such a tiling follows by \cite[Lem.~6.8]{Bausch_Cubitt_Watson2019}.
	
	Following \cite[Th.~7.2\&Th.~7.6]{Bausch_Cubitt_Watson2019} and the notation therein, the bounds for a Marker Hamiltonian of length $L$ to be found are denoted
	\[
	-\frac12-\mathrm{lwr}(w) \le \lmin(\Delta'_w) \le -\frac12-\mathrm{upr}(w)
	\]
	where $\Delta'_w$ denotes precisely one segment of the Marker Hamiltonian, encoding a computation of length $w$---which is $L$ here, but to follow the notation of the lemmas we will amend we stick to $w$: this computation runtime can then be augmented to $f(w)$.
	
	\paragraph{Lower Bound.} Note that in the proof of \cite[Lem.~E.1]{Bausch_Cubitt_Watson2019}, the lower bound was obtained by realising
	\[
	\frac{a-1}{a+1} \le 1\ \forall w
	\quad\Longleftarrow\quad \mathrm{lwr}(w) = \frac{3}{4^w}.
	\]
	Analysing \cite[Eq.~32]{Bausch_Cubitt_Watson2019} more carefully, we note the same bound also holds for $\mathrm{lwr}(w) = 9/4 \times 4^{-w}$.
	
	\paragraph{Upper Bound.} In \cite[Lem.~8]{Bausch_2020_Undecidability}, it is easy to check that the inequality
	\[
	\frac{a-1}{a+1} \le 4^{-w}
	\]
	also holds when starting with $p_w(-1/2-(9/4-\delta)\times 4^{-w})$, for any
	\[
	\delta \ge \frac{9(5\times 4^2-2)}{4(2+2^{1+4w}-4\times 4^w)} \ge \frac{9}{4^w}.
	\]
	Thus $\mathrm{upr}(w)=(9/4 + 9/4^w) \times 4^{-w}$ suffices.
	
	Finally, the $-1/2$ offset are removed as in \cite[Th.~10]{Bausch_2020_Undecidability}.
\end{proof}

\subsection{From Phase Comparison to Phase Transition}\label{sec:two-phases}
Following \cite[Sec.~8]{Bausch_Cubitt_Watson2019}, we will now combine the QTM Hamiltonian $\HTM$ with the 2D Marker Hamiltonian $\HM$, to translate the outcome of the comparison $\varphi \lessgtr \lmin(G_N)$, where $G_N$ is the Gottesman-Irani Hamiltonian simulated by $\HTM$, into the question of existence of a negative eigenstate within one square of the 2D Marker Hamiltonian.
We will assume $\HM$ is such that all checkerboard squares have square sizes $L\in4\field N$; this can always be achieved by adding a fixed-dimensional tiling, which we leave implicit in the following.

The aim is to produce an overall Hamiltonian which has \emph{negative energy density} when the encoded computation is accepting, but a \emph{positive energy density} when the encoded computation rejects.
The checkerboard structure allows us to create a repeated structure across the lattice; as each square will contribute a finite amount of either positive or negative energy, the density will follow suit.

\begin{lemma} \label{Lemma:Single_Square_Energy}
	\checked{James}
	Let $H \coloneqq \HTM\otimes\1 + \1\otimes\HM$ on a spin lattice. Then its ground state is a product state $\ket{\psi}\ox \ket{T}_c$, where $\ket T_c$ is the checkerboard tiling from $\HM$, and $\ket\psi$ the ground state of $\HTM$.
	Consider an $L\times L$ square denoted $S(L)$ within the tiling and let $H|_{S(L)}$ be the Hamiltonian restricted to such a square.
	Then, adopting the notation from \cref{eq:bds-22},
	\[
	\lmin(H|_{S(L)}) \begin{cases}
	< 0  &  \text{if\ } (L, m) = (L_N, m_N) \land \varphi \ge \lmin(G_N)-\BigO(N^{-6C}) \\
	\ge 0 & \text{if\ } (L, m) = (L_N, m_N) \land \varphi \le \lmin(G_N)- N^{-4C}+\BigO(N^{-6C}) \\
	\ge 0 & \text{if\ } (L, m) \neq (L_N, m_N).
	\end{cases}
	\]
	Furthermore, if $L=L_N$ and  $m = m_N$, then there is exactly one point in $\varphi$ where $\lmin(H|_{S(L)})$ changes from $<0$ to $=0$ which occurs in the interval
	\[
	\varphi\in \big[\lmin(G_N)- N^{-4C}+\BigO(N^{-6C})\,, \lmin(G_N)-\BigO(N^{-6C})\big].
	\]
\end{lemma}
\begin{proof}
	We choose the Marker falloff $f(L)$ such that
	\[
	\frac{9}{4}4^{-f(L)} = \frac{9}{16}\frac{1}{256L^b}
	\quad\implies\quad
	f(L) = 5 + \log_4(L^b).
	\]
	We now compare the energy (given by \cref{Theorem:Marker_Energy}) with the energy of the Hamiltonian encoding the QTM (given in \cref{Theorem:QTM_in_local_Hamiltonian}) and see the following bounds hold for sufficiently large $L$:  
	\[
	\frac{0.99}{256L^b}\left(1 - \frac{\pi^2}{24}\right) \geq \frac{9}{4}4^{-f(L)}\quad\text{and}\quad
	4^{-f(L)}\left(\frac{9}{4}- \frac{90}{4\times 2^L}\right)     \geq \frac{1.05}{256L^b}\frac{\pi^2}{24}
	\]
	where $b$ is the runtime exponent of $T=T(L)=L^{b/2}$, as given in \cref{Theorem:QTM_in_local_Hamiltonian}.
	Note $f(L)$ is trivially computable in time and space $kL$ for some constant $k$.\footnote{Indeed: define a tiling pattern that counts in base $4$, and does so $b$ times; then counts another 5 steps. Penalise tile configurations indicating that the base-$4$ expansion of $L$ is not of the form \texttt{100\ldots}, corresponding to a number $L=4^x$ for some integer $x$. This can all be done with $k=1$.}
	This yields a Marker Hamiltonian with a ground state energy as in \cref{Theorem:Marker_Energy}, such that the ground state energy is ``sandwiched'' with ample margins between the upper and lower bounds of the ground state energy of the Hamiltonian encoding the computation.
	
	As the spectrum is product by construction, the joint spectrum is then $\spec(H) = \spec(\HTM) + \spec(\HM)$, and the rest follows from \cite[Lem.~F.1]{Bausch_Cubitt_Watson2019}.
	
	Finally the fact there is exactly one point where $\lmin(H|_{S(L)})$ changes from $<0 $ to $ =0$ is due to the fact that (as per point \ref{Point:Monotonicity} of \cref{Theorem:QTM_in_local_Hamiltonian}) $\lmin(\HTM(L))$ is strictly decreasing for $\varphi\in [\lmin(G_N)- N^{-4C}+\BigO(N^{-6C})), \lmin(G_N)-\BigO(N^{-6C})]$. 
	As per the above analysis, the point at which $|\lmin(\HTM(L))| = |\lmin(\HM(L)|_S)|$ occurs for energy values corresponding to $\varphi$ in this interval, hence this point at which $|\lmin(\HTM(L))| < |\lmin(\HM(L)|_S)|$ changes to $|\lmin(\HTM(L))| = |\lmin(\HM(L)|_S)|$  can only happen at a single point.
\end{proof}

Since we want the trivial ground state in the gapped phase to have eigenvalue zero, we want to shift the Hamiltonian $H \coloneqq \HTM\otimes\1 + \1\otimes\HM$ constructed above by 1; this is a standard trick, summarised in the following lemma.
\begin{lemma}\label{Lemma:Finite_Size_Energy}
	\checked{James}
	There exists a Hamiltonian $H'$, with the same properties as $H' \coloneqq H+\sum_i P_i$, where $P_i$ is a projector, such that on  a lattice $\Lambda(L)$
	\[
	\lmin(H'(\varphi)) \begin{cases}
	= 1 + \big\lfloor \frac{L}{L_{N}}\big\rfloor^2\lmin(H(\varphi)|_{S(L_N)})  & \text{if\ } \varphi \ge \lmin(G_N) - \BigO(N^{-6C})  \\
	\geq 1  &\text{if\ } \varphi  \le \lmin(G_N) - N^{-4C} + \BigO(N^{-6C}),
	\end{cases}
	\]
	where $H \coloneqq \HTM\otimes\1 + \1\otimes\HM$.
\end{lemma}
\begin{proof}
	From \cref{Lemma:Single_Square_Energy} we know the energy of a single square $S(L_N)$.
	If the ground state $\lmin(H(\varphi)|_{S(L_N)}) <0$, then the overall ground state of the lattice becomes a checkerboard of these squares.
	It can be shown that incomplete squares contribute zero energy, giving a total energy of $\big\lfloor \frac{L}{L_{N}}\big\rfloor^2\lmin(H(\varphi)|_{S(L_N)})$ (see \cite[Corollary~F.3]{Bausch_Cubitt_Watson2019}).
	If $\lmin(H(\varphi)|_{S(L_N)}) \geq0$, then the lattice has $\geq 0$.
	
	Finally, by using the energy shift trick of \cite[Lem.~23]{Bausch2018c} (cf.~\cite[Lem.~F.5]{Bausch_Cubitt_Watson2019}), we can add on an energy shift of 1 to the Hamiltonian, giving the bounds stated in the lemma.
\end{proof}

The final step is then to combine $H'$ from \cref{Lemma:Finite_Size_Energy} with a trivial, a dense, and a guard Hamiltonian, to lift the ground state energy to a ground state energy density statement. 
This modifies the Hamiltonian so that phase transitions can occur between the ground state of the checkerboard Hamiltonian and the ground state of a trivial zero energy state.

\newcommand\Hdense{H_\mathrm{dense}}
\newcommand\Htrivial{H_\mathrm{trivial}}
\newcommand\Hguard{H_\mathrm{guard}}
\begin{theorem}[Existence of Two Phases] \label{Lemma:Two_Phases}
	\checked{James}
	Let $G_N$ be a Gottesman-Irani Hamiltonian with promise gap $\sim N^{-C}$ for some constant $C$.
	Then there exists a Hamiltonian $H^\Lambda(N,\varphi)=\sum_{\langle i,j \rangle}h^N_{i,j}(\varphi) + \sum_{i\in \Lambda}h^{N}_i$,
	and an order parameter $O_{A/B}$ acting on a subset $F\subset \Lambda$ of lattice sites, $|F|$ constant, 
	such that, as $\Lambda\rightarrow \infty$ the following holds.
	\begin{itemize}
		\item if $\varphi \le \lmin(G_N) - N^{-4C}+\BigO(N^{-6C})$, then
		\begin{enumerate}[i]
			\item $H^\Lambda$ is gapped with spectral gap $\geq 1/2$.
			\item product ground state.
			\item has order parameter expectation value $\langle O_{A/B}\rangle = 1$.
		\end{enumerate}
		\item if $\varphi \ge \lmin(G_N) -\BigO(N^{-6C})$, then
		\begin{enumerate}[i]
			\item $H^\Lambda$ is gapless.
			\item has a ground state with algebraically decaying correlations.
			\item has order parameter expectation value $\langle O_{A/B}\rangle = 0$.
		\end{enumerate}
	\end{itemize}
\end{theorem}
\begin{proof}
	Take $\Hdense$ to be a Hamiltonian that has an asymptotically dense spectrum in $[0,\infty)$ on $\mathcal H_2$.
	For convenience we choose $\Hdense$ to be the 1D critical XY-model \cite{Lieb_Schultz_Mattis_1961}.
	$\Htrivial$ to be diagonal in the computational basis, with a single product ground state $\ket0^{\otimes\Lambda}$, minimum eigenvalue $0$ and spectral gap $1$ acting on $\mathcal H_3$, and $\Hguard$ acting on $\mathcal H = \mathcal H_1 \otimes \mathcal H_2 \oplus \mathcal H_3$ via
	\[
	\Hguard \coloneqq \sum_{i\sim j}\left( \1_{1,2}^{(i)}\otimes \1_3^{(j)} + \1_3^{(i)} \otimes \1_{1,2}^{(j)} \right).
	\]
	Take $H'$ from \cref{Lemma:Finite_Size_Energy}, and set
	\[
	H^\Lambda(N,\varphi) \coloneqq H' \otimes \1_2 \oplus 0_3 + \1_1 \otimes \Hdense \oplus 0_3 + 0_{1,2} \oplus \Htrivial + \Hguard.
	\]
	Then
	\[
	\spec(H^\Lambda) = \{ 0 \} \cup \left( \spec(H') + \spec(\Hdense) \right) \cup G
	\]
	for some $G \subset [1, \infty)$, as in the proof of \cite[Th.~F.6]{Bausch_Cubitt_Watson2019}.
	From \cref{Lemma:Finite_Size_Energy}
	(i.e. using the energy shift trick of \cite[Lem.~23]{Bausch2018c} , cf.~\cite[Lem.~F.5]{Bausch_Cubitt_Watson2019}), we can assume that for lattice sizes going to infinity,
	\[
	\lmin(H') \begin{cases}
	\ge 1 & \varphi \le \lmin(G_N) - N^{-4C}+\BigO(N^{-6C}) \\
	\longrightarrow -\infty & \varphi \ge \lmin(G_N)-\BigO(N^{-6C}),
	\end{cases}
	\]
	as eventually there will exist a checkerboard square size such that $L=L_N$ and $t=m_N \ge |N|$ is satisfiable; the only differentiating condition left in \cref{Lemma:Single_Square_Energy} is then $\varphi \le \lmin(G_N) - N^{-4C}+\BigO(N^{-6C})$.
	Then if $\lmin(H) \ge 0$, we have that $\spec(H) + \spec(\Hdense) \subseteq [1,\infty)$. 
	The ground state of $H^\Lambda$ is the trivial ground state with spectral gap 1. 
	Otherwise, if $\lmin(H)\longrightarrow-\infty$, $H^\Lambda$ becomes asymptotically gapless and dense via $\Hdense$.
	
	The order parameter is then defined as
	\[
	O_{A/B}=\frac{1}{|F|}\sum_{i\in F} \left( 0_{1,2} \oplus \ketbra{0}_3 \right)^{(i)}
	\]
	which makes it clear that in case the ground state is determined by $\Htrivial$, the expectation value $\langle O_{A/B} \rangle = 1$; otherwise zero.
	
	The claim of the algebraically decaying correlation functions follows from the fact that the critical XY-model has algebraically decaying correlations functions \cite{Lieb_Schultz_Mattis_1961}.
\end{proof}

It is clear that for our construction we could choose $F$ to only contain a single spin, but we leave the statement in its generic form.

\subsection{Existence of Exactly One Critical Point}\label{sec:unique-critical-point}
The following lemma shows that there is exactly one critical point between the two phases of the Hamiltonian.
We will give the statement of the lemma here, but defer its proof to the two-parameter case (which is the more generic setting).
\begin{lemma}[Existence of Exactly One Critical Point] \label{Theorem:Crit_GS_Difference}
	\checked{James}
	Consider the Hamiltonian $H^\Lambda(N,\varphi)$ from \cref{Lemma:Two_Phases}. 
	This has exactly one critical point in the interval
	\[
	\varphi^*\in \big[\lmin(G_N)-N^{-4C}+\BigO(N^{-6C})\,, \lmin(G_N)-\BigO(N^{-6C})\big].
	\]
\end{lemma}
\begin{proof}
	As per \cref{Lemma:Single_Square_Energy} there is exactly one $\varphi$ within the given interval where $\lmin(H|_{S(L)})$ goes from $<0$ to $\geq 0$, which (as per the proof of \cref{Lemma:Two_Phases}) corresponds to the phase transition from gapped to gapless.
\end{proof}

\subsection{Reduction of Translationally Invariant Local Hamiltonian to 1-CRT-PRM}\label{sec:1-hardness-reduction}

\begin{figure}[!tb]
	\centering
	\includegraphics[height=6cm]{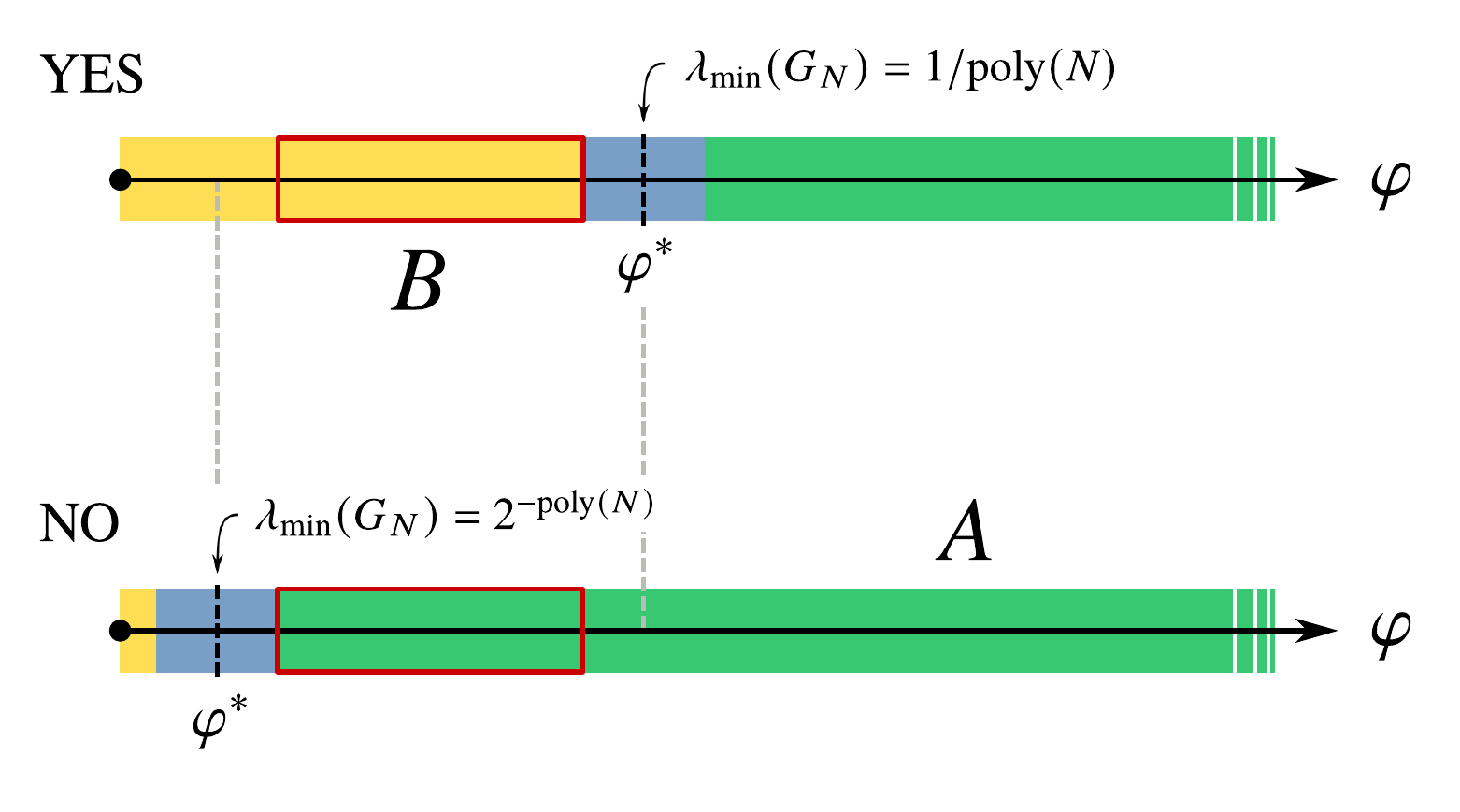}
	\caption{The two YES and NO cases (left resp.~right) of phase diagrams, for the Hamiltonian in \cref{Lemma:Two_Phases}.
		By \cref{cor:scaling-up-varphi}, the difference between the two dashed vertical lines can be scaled up to $\Omega(1)$.
		The light blue area indicates a $1/\poly(N)$-sized interval of uncertainty; yet in either case, by \cref{Theorem:Crit_GS_Difference}, there exists precisely \emph{one} critical point $\varphi^*$ therein.
		The red region indicates an interval (of up to size $\Omega(1)$, by \cref{cor:scaling-up-varphi}) which is either entirely in phase $A$ or $B$; determining which of the two cases holds is \QMAEXP-hard.}
	\label{Fig:1Param_Phase_Diagram}
\end{figure}

We first remark that the parameter range for $\varphi$ where the critical point can possibly be found is shrinking polynomially, due to the shrinking promise gap of the Gottesman-Irani Hamiltonian that we encode, and the associated comparison of its ground state energy $\lambda \lessgtr \varphi$.
There is now two approaches to scaling up the $\varphi$ parameter, so that we can get an $\BigO(1)$-area within the phase diagram where we cannot locate the critical point.
This is summarised in the following remark.
\begin{remark}\label{rem:scaling-up-varphi}
	Let $0 \le x < t$. There exists a modification to the phase comparator QTM in \cref{sec:phase-comparator-qtm} that allows one to perform the rescaled phase comparison $a \lessgtr 2^{-x} b$ for the two unitaries $U_a$ and $U_b$, where we assumed $1/10 \le b \le 1$, with an error (in the amplitudes of the resulting QPE output) upper-bounded by $2^{-2t}$, and with overhead $\poly(2^t)$.
\end{remark}
\begin{proof}
	By \cite{Sheridan2008}, we know that for an unknown ``black-box'' unitary $U$, we can implement any power $U^y$ for $y>0$ of it to precision $\epsilon$ (in trace norm) in time $\BigO(\lfloor y \rfloor + \log(1/\epsilon)/\epsilon)$ (i.e.\ with that many calls to $U$), as long as the phase $\varphi$ we want to estimate in $U$ is not too close to $0$, and no other phase lies in $(0, \varphi)$ (a gappedness constraint).
	
	For us, we want to implement $U_b^y$ for $y=2^{-x} \times 2^m$ for $m=0,\ldots,t$.
	By assumption, $b$ satisfies the gappedness condition (as $0$ and $\varphi$ are the only eigenphases of $U$, and $\epsilon \le \varphi$).
	Since $U_b$ acts on a single qubit only and $0\le x < t$, we know that a precision of $\epsilon = 2^{-4t}$  suffices to implement $U_b^y$ such that each controlled rotation gate is off (in operator norm) by at most $\epsilon \times 2^m \le 2^{-3t}\ \forall m$.
	As we have $t$ controlled rotation gates, the overall deviation is at most $t \times 2^{-3t} \le 2^{-2t}$.
	All amplitudes within the rescaled phase comparator thus at most deviate by $2^{-2t}$, and the overhead is $\poly 2^t$, as claimed.
\end{proof}

\begin{lemma}[Existence of Two Phases with $\BigO(1)$ YES/NO Threshold]\label{cor:scaling-up-varphi}
	\checked{James}
	Let $p,q$ be the polynomials defined in \cref{Theorem:Gottesman-Irani} such that $1/p(N)-1/q(N) = \Omega(N^{-C})$.
	There exists a variant of the Hamiltonian $H^\Lambda(N,\varphi)$ such that the two cases for $\varphi$ in \cref{Lemma:Two_Phases} read
	\begin{enumerate}
		\item if $\varphi^* \le A(N)=N^C(1/q(N)-\BigO(N^{-6C}))$, and
		\item if $\varphi^* \ge B(N) = N^C(1/p(N)-N^{4C} + \BigO(N^{-6C}) )$.
	\end{enumerate}
	The two bounds satisfy $B(N) - A(N) = \Omega(1)$.
\end{lemma}
\begin{proof}
	Follows immediately from \cref{rem:scaling-up-varphi}, by scaling $\varphi$ down to be within $\Theta(1)$ of $G_N$'s promise gap (which is a factor $1/\poly N \ge 2^{-t}$ for a polynomial we can compute efficiently, proportional to $p(N)$).
\end{proof}

\begin{corollary} \label{Lemma:1-CRT-PRM_QMA-hardness}
	\checked{James}
	1-CRT-PRM is \QMAEXP-hard for an $\Omega(1)$ gap. 
\end{corollary}
\begin{proof}
	Immediate from \cref{cor:scaling-up-varphi}.
\end{proof}

The fact that we can rescale the range of $\varphi$ to lie within an $\Omega(1)$ region is, in a sense, unsurprising: the same could be said to hold for the local Hamiltonian problem, where one can scale the overall Hamiltonian by a factor to have a $\Omega(1)$ promise gap as well.
Note, however, that this is a meaningless transformation: the \emph{precision} to which one wants to obtain the ground state energy is relative to the norm of the Hamiltonian (cf.~``relative promise gap'' or ``relative UNSAT penalty'', \cite{Bausch_Crosson2016}).
There are thus two scale choices for the local Hamiltonian problem: i.~the norm of the local terms, or ii.~the norm of the overall (finite-sized) Hamiltonian. Naturally, in the first case, one could obtain a stronger local interaction without increasing the individual coupling's norm by increasing the interaction degree (see e.g.~\cite{Cao2015}).
The safer definition is thus the second one---or by limiting the interaction degree of the Hamiltonian to some constant.

In our case, there is no natural ``finite size'' Hamiltonian relative to which one can define a meaningful precision; the arguably right scale with respect to which one thus has to define $\varphi$'s order of magnitude is either the local coupling strength (which is constant in our case), or relate it to the parameter $N$ itself.
In either case, and after the scaling has been applied, it makes sense to speak of $\varphi$ to be hard to approximate to $\Omega(1)$ precision, even if that means that $\varphi$ is now indeterminate in a range $[0, \poly N]$, as stated in \cref{th:main-1-intro}. It is also clear that if we know the polynomial $p(N)$ in \cref{cor:scaling-up-varphi}, and we know that it is tight for NO instances of the embedded Hamiltonian $G_N$, then it would suffice to scan $\varphi$ within a constant region.\footnote{We remark that this does not work for 2-CRT-PRM in \cref{sec:2-hard}, as there the ground state energy is not determined by a single \QMAEXP query.}

\subsection{Verifying the Local-Global Promise}\label{sec:verify-promise-1}

Finally, we need to check that the Hamiltonian used to prove hardness---defined above---satisfies the local-global promises as per \cref{Def:Local-Global_Phase} and \cref{Def:Local-Global_Gap}.

\begin{lemma} \label{Lemma:Promise_Satisfied_Phase}
	\checked{James}
	Consider an instance of the Hamiltonian $H^{\Lambda(L)}(N,\varphi)=\sum_{\langle i,j \rangle}h^N_{i,j}(\varphi) + \sum_{i\in \Lambda}h^{N}_i$, as defined in \cref{Lemma:Two_Phases}, with local terms describable in $|N|$ bits. 
	Then the Hamiltonian satisfies the global-local phase assumption \cref{Def:Local-Global_Phase} for the order parameter
	\[
	O_{A/B}=\frac{1}{|F|}\sum_{i\in F} \left( 0_{1,2} \oplus \ketbra{0}_3 \right)^{(i)}
	\]
	defined in \cref{Lemma:Two_Phases}, and for $L_0=N^{2+a+b}$.
	It also satisfies the global-local gap promise in \cref{Def:Local-Global_Gap} for the same $L_0$.
\end{lemma}
\begin{proof}
	Consider an $L>L_N=2+4t+t+N$.
	Then, by \cref{Lemma:Finite_Size_Energy}, we see that
	\begin{align}
	\lmin(H'(\varphi)) \begin{cases}
	= 1 + \lfloor \frac{L}{L_{N}}\rfloor^2\lmin(H(\varphi)|_{S(L_N)})  &\quad \varphi \le \lmin(G_N) - N^{-4C}+\BigO(N^{-6C})  \\
	\geq 1  &\quad \varphi \ge \lmin(G_N) - \BigO(N^{-6C}).
	\end{cases}
	\nonumber
	\end{align}
	Thus, when $\lmin(H(\varphi))\geq0$ the ground state of $H'(\varphi)$ that of $\Htrivial$, i.e.\ $\ket{0}^{\Lambda(L)}$, and $\langle O_{A/B}\rangle =1$.
	On the other hand, when $\lmin(H(\varphi)|_{S(L_N)})<0$, then eventually the ground state is a highly complex quantum plus classical state with $\langle O_{A/B}\rangle =0$.
	
	Thus, when $\lmin(H(\varphi)|_{S(L_N)})<0$, for the highly quantum ground state to appear the lattice size $L$ must meet the following condition:
	\begin{align}\label{Eq:Neg_Conditions}
	\bigg\lfloor \frac{L}{L_{N}}\bigg\rfloor^2\lmin(H(\varphi)|_{S(L_N)}) &< 1.
	\end{align}
	Otherwise the ground state is the zero energy state $\ket{0}^{\Lambda(L)}$.
	
	From \cref{Theorem:Marker_Energy}, when $\lmin(H(\varphi)|_{S(L_N)})<0$ and $\varphi \le \lmin(G_N) - N^{-4C}+\BigO(N^{-6C})$, then the ground state energy is
	\[
	\bigg\lfloor \frac{L}{L_{N}} \bigg\rfloor^2\left( \lmin(\HTM(L_N)) + \lmin(\HM(L_N)|_S \right) \leq \bigg\lfloor \frac{L}{L_{N}} \bigg\rfloor^2\left(  - \frac{-c_1}{L_N^{b}} \right),
	\]
	where we have used that $\lmin(\HTM(L_N)) + \lmin(\HM(L_N)|_S = -\Omega(T^{-2})=-c_1 L^{-b}$ for some constant $c_1$ (this can be seen by combining \cref{Theorem:Precise_Energies} and \cref{Lemma:Single_Square_Energy}).
	Thus by for $L>L_0$, where
	\[
	L_0\geq c_1^{1/2}L_N^{1+b/2},  
	\]
	\cref{Eq:Neg_Conditions} will be satisfied.
	Since $L_N=\BigO(N)$, we choose $L_0= \BigO(N^{2+b})$.
	
	For all $L\geq L_0$ the expectation value of $O_{A/B}$ is then constant, regardless of whether $\varphi \ge \lmin(G_N) - \BigO(N^{-6C})$ or $\varphi \le \lmin(G_N) - N^{-4C}+\BigO(N^{-6C})$.
	Thus the Hamiltonian satisfies the global-local phase promise in \cref{Def:Local-Global_Phase}.
	
	\paragraph{Global-Local Gap Promise:}
	The proof for the global-local gap promise is almost the same.
	For the $L_0$ above, we see that if $\varphi \ge \lmin(G_N)  - \BigO(N^{-6C})$, then the system has a very negative energy and a spectral gap $\Delta(L)=O(1/L^2)$.
	If $\varphi \le \lmin(G_N) - N^{-4C}+  \BigO(N^{-6C})$, then the ground state is $\ket{0}^{\Lambda(L)}$ with zero energy and has the same gap as $\Htrivial$: $\Delta \geq 1$.
\end{proof}

\section{P\textsuperscript{QMA\textsubscript{EXP}} Hardness of 2-CRT-PRM}\label{sec:2-hard}

In this section we modify the construction to prove hardness for a 2-parameter Hamiltonian.
To prove this we will make a reduction from $\forall$-TI-APX-SIM, as defined \cref{Def:forall_APX-SIM}, which was proved to be \PQMAEXP-complete by \cite{Watson_Bausch_Gharibian_2020}. 
This a variant of the APX-SIM problem, which itself was introduced by Ambainis and shown to be $\PQMALOG$-complete for for $\log(n)$-local Hamiltonians \cite{Ambainis2013}, a result that was then extended to $\BigO(1)$-locality by \cite{Gharibian_Piddock_Yirka2019}.

\paragraph{Proof Outline.}
The proof method here will be similar to the 1-parameter case, but instead of a reduction to the Local Hamiltonian problem, we perform a reduction to $\forall$-TI-APX-SIM, which is the question of approximating the expectation value of all low-energy states of a Hamiltonian with respect to a local observable.
This means we construct---just as described in \cref{sec:1-hard}---a Hamiltonian $H_N(\varphi)$ which, in its ground state, encodes the following computation.
\begin{enumerate}
    \item Perform QPE to extract $N$ from local terms.
    \item Perform a phase comparison QPE on the unitary encoding $\varphi$ and $\exp(\ii t K_N)$, where $K_N$ is a translationally-invariant local spin Hamiltonian with a \PQMAEXP-complete $\forall$-TI-APX-SIM problem (on a spin chain of length $N$).
    The joint witness stems from an unconstrained input state.
    If this input state was an eigenstate of $K_N$ with eigenvalue $\lambda$, the phase comparator QPE extracts the difference $\lambda-\varphi$ to bit precision $\sim|N|$.
    \item If $\varphi<\lambda$, an output flag is set to $\ket0$; otherwise it is set to $\ket1$.
    \item Another flag qubit captures the output bit of the \PQMAEXP computation performed within the history state of $H_N(\varphi)$.
\end{enumerate}
\noindent
An energy penalty is then given to the joint energy eigenvalue comparison \emph{and} output bit of the \PQMAEXP computation, in the sense that
\begin{enumerate}
    \item All eigenstates of $K_N$ that are not considered ``low energy'' are penalised.
    \item Those eigenstates of $K_N$ that fall below the ``low energy'' cutoff are not inflicted with a penalty; but they are subject to a penalty from the \emph{observable operator} $(B-\theta\1)$ for some scalar offset $\theta>0$.
    Here $B$ is the operator for which determining the expectation of on low energy states of $K_N$ is \PQMAEXP-complete.
\end{enumerate}
The result is that the low-energy eigenspace of $H_N(\varphi)$ plus penalties is greater or smaller than some polynomial falloff we can calculate to high precision, and which will depend on the scalar expectation value offset $\theta$ that serves as the second parameter.

As in the one-parameter case, we can combine this energy penalty with a bonus of a Marker Hamiltonian to obtain a joint 1D spin Hamiltonian with the property that it has a
\emph{negative} ground state energy if we have a \YES-instance---i.e.\ expectation values of low-energy states lie below some threshold---and the scalar offset is below a cutoff;
and a \emph{positive} ground state energy for a \NO instance, \emph{or} for a too-small scalar offset $\theta$.
With standard techniques this dichotomy is then amplified to a gapless resp.~gapped phase in the thermodynamic limit.

In order to understand why this two-parameter family of Hamiltonians has a \PQMAEXP-hard-to-compute phase diagram, note that to figure out the relevant $\varphi$ region within which a phase transition can occur takes multiple queries to a \QMA oracle, as we need to identify $\lmin(K_N)$ to sufficient precision; and then we don't yet know whether the output is a \YES or \NO case, so there is \emph{two} possible $\theta$-regions around which to explore the phase diagram.

\subsection{Additional Preliminaries} \label{Sec:Prelims_2-CRT-PRM}
In \cite{Watson_Bausch_Gharibian_2020}, a Hamiltonian $K_N\in \mathcal{B}(\C^d)^{\ox N}$ was used to prove hardness of $\forall$-TI-APX-SIM.
Importantly, this Hamiltonian has the following properties:
\begin{lemma}[From \cite{Watson_Bausch_Gharibian_2020}]\label{Lemma:TI-APX-SIM_Completeness}
There exists a fixed one-local observable $A$ and interaction terms $k_{i,i+1}\in  \mathcal{B}(\C^d\otimes \C^d)$ acting between pairs of nearest neighbour qudits, which define a Hamiltonian on a 1D chain of length $N$, $K_N=\sum_{i=1}^{N-1}k_{i,i+1}$ such that for all states $\ket\psi$ that satisfy $\bra{\psi}K_N\ket{\psi}\leq \lambda_0(K_N) + \delta$ for $\delta=\Omega(1/\poly(N))$ either of the following holds:
\begin{itemize}
    \item[\YES:] $1-1/\poly(N)\leq \bra{\psi}A\ket{\psi}\leq 1$, or
    \item[\NO:] $0\leq \bra{\psi}A\ket{\psi}\leq  1/\poly(N)$.
\end{itemize}
Determining which case is true is $\PQMAEXP$-complete.
\end{lemma}
We note that \cref{Lemma:TI-APX-SIM_Completeness} is not quite what was proven in \cite{Watson_Bausch_Gharibian_2020}, where e.g.\ the overlap with the observable in the first case was $1/T_{K_N}-\BigO(2^{-\poly(N)})\leq \bra{\psi}A\ket{\psi}\leq 1/T_{K_N}$, for $T_{K_N}$ the length of the encoded computation.
However, we can adjust the construction using an idling technique from \cite{Caha_Landau_Nagaj_2018} which increases the weight on the output bit of the computation (such that the encoded computation has its runtime increased to $P_1(N)T_{K_N}$, for some polynomial $P_1(N)$ we are free to choose).
Furthermore, we ask that there be some marker flag, placed next at or next to the output qubit, which indicates when the first part of the computation---before the idling---has finished (this allows us to keep $A$ as a 1-local operator).
These techniques are by now standard, and we will not go into details.

We denote the set of eigenstates of $K_N$ for which the energy expectation value is below the cutoff as
\begin{equation}\label{eq:S_delta}
    S_\delta \coloneqq \big\{ \psi: H\ket{\psi} = \lambda\ket\psi \ \text{where\ } \lambda \leq \lmin(K_N) + \delta \big\},
\end{equation}
which means $\bra{\psi}H\ket{\psi}\leq \lmin(K_N) + \delta$ for all $\ket{\psi} \in \Span(S_\delta)$.
We also define the shifted observable
\begin{align}\label{Eq:B_Definition}
    B\coloneqq  A+  \1, 
\end{align}
which if $A$ is a projector has eigenvalues in the set $\{1,2\}$, which we label as $\lambda_0(B)=1$ and $\lambda_1(B)=2$.
This offset merely simplifies some of the maths in due course.
Together with \cref{Lemma:APX-SIM_Complexity}, this choice of $B$ in \cref{Eq:B_Definition} immediately yields the following corollary.
\begin{corollary}\label{cor:all-states-expvalue}
We use the notation of \cref{Lemma:APX-SIM_Complexity}, for an observable $A$ that is a one-local projector, and $B$ as defined in \cref{Eq:B_Definition}.
Any state $\ket\psi \in S_\delta$ for the Hamiltonian $K_N$ then has expectation value either $\le \lambda_0(B) + \BigO(1/P_1(N))$, or $\ge \lambda_1(B) - \BigO(1/P_1(N))$ for a polynomial $P_1(N)$ we are free to choose.
\end{corollary}

\subsection{A Modified Phase Comparator QTM}\label{sec:qtm-2}
The following lemma follows the same setup as \cref{Lemma:QTM_Output}.
\begin{lemma}[Multi-QPE QTM] \label{Lemma:QTM_Output_2-CRT-PRM}
\checked{James}
Let $K_z$ be the translationally invariant Hamiltonian on chain of length $z$ described in \cref{Lemma:TI-APX-SIM_Completeness}.
Take the same setup as in \cref{Lemma:QTM_Output}, but where the Hamiltonian $G_z$ is replaced by $K_z$.
The output of this QTM $\mathcal{M}(N,\varphi, t,\ket\nu)$ will then be
{\small
\begin{align}
    \ket\chi =
     &\ \sum_{ z\in V_t} \sum_{x\le 0}\sum_{g} \alpha_{x}(\zp,g) \gamma_z \kappa_{g}(\zp) \ket{11}_f\ket{z}\ket{x}\sum_j \sigma_j(g,\zp)\ket j \ket{g_{\zp,j}}\ket{\xi_{\zp,g}} +  \label{eq:chi-with-flag-2}\\
    &\ \sum_{z\in V_t} \sum_{x > 0}\sum_{g} \alpha_{x}(\zp,g) \gamma_z \kappa_{g}(\zp) \ket{10}_f\ket{z}\ket{x}\sum_j \sigma_j(g,\zp)\ket j \ket{g_{\zp,j}}\ket{\xi_{\zp,g}} + \nonumber  \\
    &\ \sum_{ z\not\in V_t} \sum_{x \le 0} \sum_{g}  \alpha_{x}(\zp,g)  \gamma_z \kappa_{g}(\zp) \ket{01}_f\ket{z}\ket{x}\sum_j \sigma_j(g,\zp)\ket j \ket{g_{\zp,j}}\ket{\xi_{\zp,g}} + \nonumber \\
    &\ \sum_{ z\not\in V_t} \sum_{x > 0} \sum_{g}  \alpha_{x}(\zp,g)  \gamma_z \kappa_{g}(\zp) \ket{00}_f\ket{z}\ket{x}\sum_j \sigma_j(g,\zp)\ket j \ket{g_{\zp,j}}\ket{\xi_{\zp,g}}, \nonumber
\end{align}}
\noindent
where we have expanded $\ket{g_\zp} = \sum_j \sigma_j(g,\zp)\ket j \ket{g_{\zp,j}}$ such that the $\ket j$ denote the eigenvectors of $B$ defined in \cref{Eq:B_Definition}.
\end{lemma}

\begin{proof}
Follows from the output state given in \cref{Lemma:QTM_Output} and the form of the unconstrained state taken as ``input''.
\end{proof}

We now need an equivalent expression to \cref{eq:eta} which captures the expected output penalty that we wish to inflict later on.
Here, we will modify the flag projector slightly; instead of using $\ketbra{11}_f$ that just singles out those eigenstates of $K_z$ that have low energy, we also add in the observable $B$ as defined in \cref{Eq:B_Definition}, which acts on the \emph{output} of the computation.\footnote{This penalty can be made 1-local using standard methods. We omit this here.}
As $K_z$ encodes a $\PQMAEXP$-hard computation, this output bit is a single qubit; and can be assumed to satisfy the bounds given in \cref{cor:all-states-expvalue}.

Taking the output state $\ket\chi$  of $\mathcal{M}(N,\varphi, t,\ket\nu)$ from \cref{Lemma:QTM_Output_2-CRT-PRM}, and letting $B$ be the local observable from \cref{Eq:B_Definition} \& \cref{cor:all-states-expvalue}, we set
{\small
\begin{align}
    \eta(N,\varphi, \theta, t,\ket{\nu}) \coloneqq & \Tr\left( \left[ \ketbra{11}_f \otimes \left( B - \theta\1 \right) \otimes \1 \right] \ketbra\chi \right) \nonumber \\
    =& \sum_{z\in V_t}|\gamma_z|^2\sum_{x\leq 0}\sum_{g} |\alpha_x(\zp,g)|^2|\kappa_{g}(\zp)|^2\left( \sum_j\lambda_j(B)|\sigma_{j}(\zp,g)|^2 - \theta   \right)  . \label{Eq:eta_Multiple_Variables}
\end{align}}
Here, as before, $\gamma_z$ represents the amplitudes of QPE on $U_N$, while $\alpha_x(\zp,g)$ represent the amplitudes of QPE over $\varphi$ and $K_z$ on eigenstate $\ket{g}$, and $\sigma_j(\zp,g)$ are the coefficients of the eigenstates of $B$.
$\kappa_g(\zp)$ are coefficients of basis expansions of $\ket\nu$ in the energy eigenbasis of $K_z$.

We further define
\begin{equation}\label{eq:eta_max_multi}
\eta_{\max}(N,\varphi, \theta, t) \coloneqq \max_{\ket{\nu}}\eta(N,\varphi, \theta, t,\ket{\nu}).
\end{equation}
This is the maximum acceptance probability that the computation can output for a given $N,\varphi, \theta,t$, maximised over all states $\ket\nu$.
\begin{remark}\label{rem:etamax-eigenstate}
\checked{James}
    For $t \ge |N|$,
    $\eta(N,\varphi, \theta, t,\ket\nu)$ assumes its maximum for an eigenstate $\ket\nu=\ket g$ of $K_N$.
\end{remark}
\begin{proof}
    For $t \ge |N|$, $\gamma_N=1$ and all other $\gamma_z=0$ for $z\neq N$.
    Hence
    \begin{align*}
        &\argmax_{\ket\nu}\eta(N,\varphi,\theta,t,\ket\nu) \\
        = &\argmax_{\ket\nu} \sum_g |\kappa_{g}(N)|^2 \left( \sum_{x\le 0}|\alpha_x(N,g)|^2\left( \sum_j\lambda_j(B)|\sigma_{j}(N,g)|^2 - \theta   \right)\right) \\
        = &\argmax_{\ket\nu} \sum_g |\kappa_{g}(N)|^2 \left( \sum_{x\le 0}|\alpha_x(N,g)|^2\right) \Gamma(g, \theta).
    \end{align*}
    where $\ket\nu = \sum_g \kappa_g(N) \ket g$ and $\Gamma(g, \theta) \coloneqq \left( \sum_j\lambda_j(B)|\sigma_{j}(N,g)|^2 - \theta   \right)$.
    Thus $\eta(N,\varphi,\theta,t,\ket\nu)$ is a convex combination of the $\left( \sum_{x\le 0}|\alpha_x(N,g)|^2\right)\Gamma(g, \theta)$, and its maximum is assumed at an extremal point. 
    The claim follows.
\end{proof}

As a first step, we prove the following lemma.
\begin{lemma}[Technical Lemma]\label{Lemma:eta-tech}
\checked{James}
Let $f(t)$ be any function such that $f(t)/2^t =\BigO(2^{-ct})$ for some constant $c>0$.
Further let  $z\in [N]$.
Let the phase estimation be done for the operator $U_{K_z}U_{\varphi}^\dagger$ on an eigenstate $\ket{g} = \ket{g_z}$ with eigenvalue $\lambda$.
Using the same notation as in \cref{Eq:eta_Multiple_Variables}, if $\varphi\leq \lambda - 1/f(t)$,
\begin{align*}
    \sum_{x \le 0}|\alpha_{x}(z,g)|^2 &= \BigO\left(\frac{f(t)}{2^t}\right). \\
\intertext{In contrast, if $\varphi\geq \lambda + 1/f(t)$}
    \sum_{x \le 0}|\alpha_{x}(z,g)|^2&=1-\BigO\left(\frac{f(t)}{2^t}\right).
\end{align*}
\end{lemma}
\begin{proof}
We closely follow the notation in \cite[Sec.~5.2]{Nielsen_and_Chuang}.
Let $X=\lambda- \varphi$ and $b_X$ be the best $t$ bit estimate of $X$ such that $0.b_X$ is smaller than $X$; this in turn means that the difference $\delta_X  \coloneqq  X - b_X/2^t$ satisfies $\delta_X \in [0, 2^{-t})$.
We further denote with $\alphal$ the amplitude of $\ket x = \ket{(b_X+\ell)\bmod 2^t}$.

\paragraph{Case \paramath{\varphi \le \lambda - 1/f(t)}.}
In this case we have
\[
    2^t\varphi \le 2^t \lambda - \frac{2^t}{f(t)}
    \quad\Longrightarrow\quad
     \frac{2^t}{f(t)} - 1 \le b_X 
\]
and thus
\begin{align*}
x\le0&&\Longleftrightarrow&& b_X + \ell &\le 0 &&\bmod 2^t \\
&&\Longrightarrow&& \ell &\le - b_x \leq 1 - \frac{2^t}{f(t)} && \bmod 2^t.
\end{align*}
Then 
\begin{align}
    \sum_{x<0} |\alpha_{x}(z,g)|^2 &= \sum_{\ell <- 2^t/f(t)+1 } |\alpha_{\ell}|^2 \nonumber \\
    &\leq  \sum_{\ell=-2^{t-1}}^{-2^{t-1}/f(t)} \frac{1}{(2^t\delta_X -\ell )^2} && \hspace{-12mm}\text{by \cite[eq.~5.29]{Nielsen_and_Chuang}} \nonumber\\
    &\leq \sum_{k=2^{t-1}/f(t)}^{2^{t-1}} \frac{1}{k^2} \label{Eq:l^-2_sum} \\
    &= \psi^{(1)}\left(2^{t-1}/f(t)\right) - \psi^{(1)}\left(1+2^{t-1}\right) \nonumber \\
    &= \BigO\left(\frac{f(t)}{2^t}\right), \nonumber
\end{align}
where $\psi^{(n)}(z)$ denotes the $n$\textsuperscript{th} derivative of the digamma function $\psi(z)  \coloneqq  \Gamma'(z)/\Gamma(z)$.

\paragraph{Case \paramath{\varphi\ge\lambda+1/f(t)}.} 
In this case we see that
\[
    2^t\varphi \ge 2^t \lambda + \frac{2^t}{f(t)}
    \quad\Longrightarrow\quad
     1 - \frac{2^t}{f(t)} \ge b_X .
\]
As the above inequality goes the other way to before, we will find it useful to consider
\begin{align*}
x>0&&\Longleftrightarrow&& b_X + \ell &> 0 &&\bmod 2^t \\
&&\Longrightarrow&& \ell &> - b_x \ge   \frac{2^t}{f(t)}-1 && \bmod 2^t.
\end{align*}
Using
\begin{align}
    \sum_{x<0} |\alpha_{x}(z,g)|^2 &= \sum_{\ell \le - b_X } |\alpha_{\ell}|^2 \nonumber \\
    &= 1 - \sum_{\ell > - b_X } |\alpha_{\ell}|^2 \nonumber \\
    &= 1 - \sum_{\ell > -1+2^t/f(t) } |\alpha_{\ell}|^2  \nonumber\\ 
    &\geq 1 - \sum_{\ell=2^{t-1}/f(t)}^{2^{t-1}} \frac{1}{(2^t\delta_X -\ell )^2} && \hspace{-12mm}\text{by \cite[eq.~5.29]{Nielsen_and_Chuang}} \nonumber\\
    &\geq 1- \sum_{\ell =2^{t-1}/f(t)}^{2^{t-1}} \frac{1}{\ell^2} \label{Eq:alpha_1} \\
    &= 1- \BigO\left(\frac{f(t)}{2^t}\right). \label{Eq:alpha_2}
\end{align}
\Cref{Eq:alpha_2} follows from \cref{Eq:alpha_1} using the same analysis as for \cref{Eq:l^-2_sum}.
The claim follows.
\end{proof}

As per the 1-CRT-PRM case, $\eta$ will play a fundamental role in our construction; and as such we will derive bounds on its magnitude and derivative in the following.

\begin{lemma}[Bounds on $\eta_{\max}$]\label{Lemma:Rejection_Probabilities_2}
\checked{James}
Let $K_N$, as defined in \cref{Lemma:TI-APX-SIM_Completeness}, and choose $D>0$ such that the cutoff $\delta$ from \cref{Lemma:TI-APX-SIM_Completeness} satisfies $\delta = \Theta( N^{-D})$, and as in \cref{eq:S_delta} denote with $S_\delta$ the set of eigenstates of $K_N$ with eigenvalue $\le \lmin(K_N) + \delta$.
Let $\varphi \in [\lmin(K_N) + \delta/3,\lmin(K_N) + 2\delta/3]$
Then there exists an integer $D'\ge 1$ and a $\lambda_{j*}(B)\in\{\lambda_0(B), \lambda_1(B)\}$ such that the following bounds hold.
\begin{itemize}
   \item  If $|\bra{\psi}B\ket{\psi}-\lambda_{j*}(B)\,| \le 1/P_1(N)$ for all $\ket{\psi}\in S_\delta$ and $t\geq D'|N|$, then
\[
    \eta_{\max}(N,\varphi, \theta, t) = \lambda_{j^*}(B) - \theta + \frac{1}{P_2(N)}.
\]
for a polynomial $P_2(N)$ such that $\BigO( 1/P_1(N)) + \BigO(f(t)/2^t) \leq  1/P_2(N)$.
\item If $t<|N|$, then
\begin{align}\label{Eq:Short_Tape}
     \eta_{\max}(N,\varphi, \theta, t)&=\BigO\left(\frac{1}{2^{t/2}}\right).
\end{align}
\end{itemize}
\begin{proof}
We consider each case individually.

\paragraph{For $t<|N|$.}
In this case and irrespective of $\varphi$ or $\theta$ we have that $\sum_{z\in V_t}|\gamma_z|=\BigO(2^{-t/2})$ for any state $\ket{\psi}$ as input, thus using \cref{Eq:eta_Multiple_Variables} we get \cref{Eq:Short_Tape}.

\paragraph{For $t\geq D'|N|$.}
As we have $D'\ge 1$ to be determined, $N$ is expanded in full; hence $\gamma_{\enc(N)}=1$, and all other $\gamma_z=0$.
By \cref{rem:etamax-eigenstate}, as $t\ge |N|$, the maximum of $\eta(N,\varphi,\theta,t,\ket\nu)$ is assumed for an eigenstate $\ket{g}$ of $K_N$, thus we can write:
\begin{align} 
    \eta(N,\varphi, \theta, t,\ket{g})&= \sum_{x\leq 0} |\alpha_x(N,g)|^2\left( \sum_j\lambda_j(B)|\sigma_{j}(N,g)|^2 - \theta   \right)  \nonumber\\
    &\eqqcolon  \left(\sum_{x\leq 0}|\alpha_x(N,g)|^2\right)\Gamma(g,\theta), \label{Eq:Gamma_Definition}
\end{align}
and note that $| \Gamma(g,\theta) | \le 1$.\footnote{We reuse $\Gamma$ here; it is not the same $\Gamma$ as in \cref{rem:etamax-eigenstate}.}

Now assuming $\varphi \in [ \lmin(K_N) + \delta/3, \lmin(K_N) + 2\delta/3]$, we have that any eigenstate of $K_N$ above the energy cutoff $\delta$---i.e.\ any $\ket{g} \not \in S_\delta$---has eigenvalue $\lambda$ at least $\delta/3$ above $\varphi$.
Therefore, by \cref{Lemma:eta-tech}, we can find a function $f(t)$ such that
\[
    \varphi \le \lambda - \frac{\delta}{3} \eqqcolon \lambda - \frac{1}{f(t)}
    \quad\Longleftarrow\quad
    f(t) = \frac{3}{\delta} = \Theta( N^D).
\]

To employ \cref{Lemma:eta-tech} productively, we also need that $\frac{f(t)}{2^t} = 2^{-ct}$ for some constant $0<c$.
These constraints determine the constant $D'$: as $t\ge D'|N|$, we have $N  \le 2^{t/D'}$, and thus
\[
    2^{(1-c)t} = f(t) = N^D \le 2^{t D/D'}
\]
and thus $D/D' \ge 1-c$ or $D' \le D/(1-c)$.
This means we need to choose $0<c<1$, in which case one can choose $1 \le D' \le D/(1-c)$.
For this choice of $D'$, by \cref{Lemma:eta-tech}, we have
\[
    \sum_{x\leq 0}|\alpha_x(N,g)|^2 = \BigO(2^{-ct})
\]
and hence with $|\Gamma(g,\theta)|\le 1$,
\begin{align}
    \sum_{\ket g\neq S_\delta} \left(\sum_{x\leq 0}|\alpha_x(N,g)|^2\right)|\kappa_{g}(N)|^2 \Gamma(g,\theta) &= \BigO(2^{-ct}) \sum_g |\kappa_g(N)|^2  \nonumber\\
    &= \BigO(2^{-ct}) = 1/\poly N.  \label{eq:high-states-small-contrib}
\end{align}

We have just proven that none of the high-energy eigenstates $\ket g \neq S_\delta$ make any significant contribution to $\eta_{\max}$; as such it suffices to limit our further analysis to the case where $\ket g \in S_\delta$.

We now make use of the assumption that $| \bra\psi B \ket\psi - \lambda_{j^*}(B) | \le 1/P_1(N)$ for all $\ket\psi \in S_\delta$.
As the $\sigma_j$ in \cref{eq:chi-with-flag-2} label the coefficients of the basis expansion with respect to the eigenbasis of $B$, and $B$ only has two eigenvalues, this statement means that
\begin{align}
    |\sigma_{j^*}(N,g)|^2 =1-1/P_1(N)
    \quad\text{and}\quad
    |\sigma_{1-j^*}(N,g)|^2 =1/P_1(N). \label{Eq:Sum_kappa_1}
\end{align}

Employing \cref{Lemma:eta-tech} again to lower-bound the amplitudes on the low-energy coefficients $\alpha_x$ for $x\le 0$, 
\begin{align*}
    \eta(N,\varphi, \theta, t,\ket{g})&= \left(1 - \BigO(2^{-ct})\right)\left( \lambda_0(B)|\sigma_0|^2 + \lambda_1(B)|\sigma_1|^2 -\theta   \right) \\
    &= \left(1 - \BigO(2^{-ct})\right)\left( \lambda_{j^*}(B)
    + \BigO(1/P_1(N)) -\theta   \right).
\end{align*}
As $\theta$ is assumed not to scale with $N$, we can write
\[
    \eta(N, \varphi, \theta, t, \ket g) + \theta = \lambda_{j^*}(B) + \BigO(1/P_1(N)) + \BigO(2^{-ct}).
\]
We now combine (keeping the hidden constants)
\[
    P_2(N) \coloneqq \BigO(1/P_1(N)) + \BigO(2^{-ct}).
\]
The claim follows.
\end{proof}
\end{lemma}

\begin{corollary}\label{Cor:Rejection_Probabilities_2-2}
\checked{James}
Let the notation be as in \cref{Lemma:Rejection_Probabilities_2}. 
Then, for $t\ge D'|N|$ and $|\bra{\psi}B\ket{\psi}-\lambda_{j*}(B)\,| \le 1/P_1(N)$ for all $\ket{\psi}\in S_\delta$ and $\varphi \in [\lmin(K_N) + \delta/3,\lmin(K_N) + 2\delta/3]$, we have that
\[
\eta_{\max}(N,\varphi, \theta, t)
\begin{cases}
    \geq \frac{1}{2} & \theta \leq \lambda_{j^*}(B) - \frac{1}{2} - 1/P_2(N). \\
     \leq  \frac{2}{5} & \theta \geq \lambda_{j^*}(B) - \frac{2}{5}  + 1/P_2(N) 
    \end{cases}
\]
\end{corollary}
We chose the two special points $2/5$ and $1/2$ as they bracket $2/5 < 7/16 < 1/2$, which will be important later on to offset the bonus from a marker Hamiltonian (cf.~\cref{Lemma:Single_Square_Energy_2}).

We note that \cref{Lemma:Rejection_Probabilities_2} leaves a ``gap'' between the two conditions on $t$, i.e.\ $t\ge D'|N|$ on the one hand, and $t<|N|$ on the other.
To make sense of this, we note that as in \cref{Theorem:QTM_in_local_Hamiltonian}, we will be able to set $t$ independently by an extra marker $\ket\blacksquare$, whose location will be treated as input for $t$.
This means that if we are in the branch $t\ge|N|$, and since $D'$ is a constant, we will be able to compare---with a Turing machine injected in between steps (2.) and (3.) in \cref{Lemma:QTM_Output_2-CRT-PRM}---$t$ to $D'|N|$, and set another failure flag.
Leaving the details to the reader, we summarise this argument in the following corollary.
\begin{corollary}\label{cor:Rejection_Probabilities_2}
\checked{James}
    A modified Multi-QPE QTM as in \cref{Lemma:QTM_Output_2-CRT-PRM} satisfies the bounds in \cref{Lemma:Rejection_Probabilities_2,Cor:Rejection_Probabilities_2-2}, but such that if $t<D'|N|$, $\eta_{\max}(N,\varphi,\theta,t) = \BigO(2^{-t/2})$.
    The overhead in runtime is at most $\poly(N)$.
\end{corollary}

\subsection{An Approximate Phase Comparator QTM, Season 2}
As in the 1-CRT-PRM case, we will need to remove the QTM's direct access to $U_{K_z}$ and instead have the QTM approximate the unitary using Hamiltonian simulation.
In the following, we will prove the analogy of \cref{Lemma:hamSim_Error} for the two-parameter setting.

\begin{lemma}[Hamiltonian Simulation Error]
\label{Lemma:2-PRM-Ham_Sim_Error}
\checked{James}
Let $\mathcal{M}(N,\varphi, \theta, t, \ket{\nu})$ be the QTM described in \cref{Lemma:QTM_Output_2-CRT-PRM} with all gates implemented without error.
Then there exists a QTM $\mathcal{M}'(N,\varphi, \theta, t, \ket{\nu})$ performing the same algorithm, except where the gate $U_{K_z}$ is instead performed by a Hamiltonian simulation algorithm in \cref{Lemma:HamSim-QTM} and such that $\mathcal{M}'(N,\varphi, \theta, t, \ket{\nu})$ satisfies the following:
\begin{enumerate}
    \item Let $\eta'(N,\varphi, \theta, t, \ket{\nu})$ be defined in the same way as $\eta(N,\varphi, \theta, t, \ket{\nu})$ from \cref{Eq:eta_Multiple_Variables}, but corresponding to the output of $\mathcal{M}'$.
    Then $\max_{\ket{\nu}}\eta'(N,\varphi, \theta, t, \ket{\nu})$ satisfies the same bounds as $\max_{\ket{\nu}}\eta(N,\varphi, \theta, t, \ket{\nu})$ from \cref{Cor:Rejection_Probabilities_2-2} for $\varphi\in [\lmin(K_N)+\delta/3,\lmin(K_N)+2\delta/3]$.
    \item The runtime overhead is at most $\poly(N, 2^t)$.
\end{enumerate}

\end{lemma}
\begin{proof}
The proof is almost identical to \cref{Lemma:hamSim_Error}: we use the same Trotterised simulation as \cref{Lemma:hamSim_Error}; the result is a unitary $\tilde{{U}}=\ee^{\ii\pi H''}$ such that the spectrum of the generating Hamiltonian $|\lmin(H'')-\lmin(K_N)|_\infty\leq \frac{\kappa \epsilon}{2^t}$.
By choosing $\epsilon=2^{-2t}$, there is an error of $\BigO(2^{-2t})$.

Now consider any eigenstate of $H'$, denoted $\ket{\nu}$, such that $\bra{\nu}H'\ket{\nu}\leq \lambda_0(H)+2\delta/3$.
Then by \cref{Lemma:Close_Evolution_Hamiltonian} $\bra{\nu}K_N\ket{\nu}\leq \lmin(K_N)+2\delta/3+2^{-2t}$.

Since we are promised either all states $\bra{\nu}H\ket{\nu}\leq \lmin(K_N)+\delta$ have $|\bra{\nu}B\ket{\nu}-\lambda_0(B)|=  1/P_1(N)$ or $|\bra{\nu}B\ket{\nu}-\lambda_1(B)|=   1/P_1(N)$, and $\delta = N^{-D}$, we can in addition ensure that simulation error satisfies $\epsilon = \BigO(\delta^2)=\BigO(N^{-2D})=O(2^{-2t})$; we then see all states satisfying $\bra{\nu}H'\ket{\nu}\leq \lmin(H')+2\delta/3$ also satisfy these same bounds on $\bra{\nu}B\ket{\nu}$.
Thus $\max_{\ket{\nu}}\eta'(N,\varphi, \theta, t, \ket{\nu})$ satisfies the same bounds as $\max_{\ket{\nu}}\eta(N,\varphi, \theta, t, \ket{\nu})$ in \cref{Cor:Rejection_Probabilities_2-2}, and the $\poly(N,2^t)$ runtime overhead follows by construction.
\end{proof}

\begin{lemma}[Gate Approximation Error]\label{Lemma:SK_HamSim_Approximation_2}
\checked{James}
Let $\mathcal{M}'(N,\varphi, \theta, t, \ket{\nu})$ be the QTM described in \cref{Lemma:2-PRM-Ham_Sim_Error} with the Hamiltonian simulation subroutine performed, but all other gates still done exactly.
Then there exists a QTM $\mathcal{M}''(N,\varphi, \theta t,\epsilon, \ket{\nu})$ that satisfies the following:
\begin{enumerate}
    \item $\mathcal{M}''$ only has access universal gate set and the gates $U_N$ and $U_\varphi$.
    \item Let $\eta''(N,\varphi, \theta,t,\epsilon,\ket{\nu})$ be define in the same way as $\eta(N,\varphi, \theta,t, \ket{\nu})$ from \cref{Eq:eta_Multiple_Variables}, but corresponding to the output of $\mathcal{M}''$.
    Then $\max_{\ket{\nu}}\eta''(N,\varphi,\theta ,t,\epsilon,\ket{\nu})$ satisfies the same bounds as $\max_{\ket{\nu}}\eta(N,\varphi,\theta,t, \ket{\nu})$ in \cref{Cor:Rejection_Probabilities_2-2}.
    \item
    The runtime overhead due to the Hamiltonian simulation is at most a factor of $\poly\log(N,2^t)$.
\end{enumerate}
\end{lemma}
\begin{proof}
As in \cref{Lemma:SK_HamSim_Approximation} we simply choose the Solovay-Kitaev algorithm to approximate the necessary gates to high enough precision. The details are identical.
\end{proof}

The analysis of the rest of the Hamiltonian is exactly the same as the 1-CRT-PRM case, where we encode the action of $\mathcal{M}''$ in the ground state of a Hamiltonian.

\subsection{Existence of a Unique Critical Line}\label{sec:unique-critical-point-2}
For any fixed $\theta$, it is clear that the derivative bounds on $\partial \eta/\partial \varphi$ follow immediately from \cref{lem:eta-monotonous,lem:eta-solovay}.
In the one-parameter case, this monotonicity of $\eta$ with respect to $\varphi$ ensured that there can only be a single phase transition point $\varphi^*$.
Since in the two-parameter case there is \emph{two} parameters to tune---$\varphi$ and $\theta$---critical points occur as a family $\theta^*(\varphi)$.
What we want to show here is that for each $\varphi$, there exists only a single such critical point $\theta^*(\varphi)$.

\begin{lemma}[$\eta$ Monotonous]\label{Lemma:eta_2_Decreasing}
\checked{James}
    Let $t\ge |N|$.
    $\max_{\ket{\nu}}\eta(N,\varphi, \theta, t, \ket{\nu})$ is strictly monotonically decreasing with respect to $\theta$. More precisely, we have that
\[
    \frac{\partial}{\partial\theta}\eta_{\max}(N,\varphi, \theta, t) \le -1/2.
\]
\end{lemma}
\begin{proof}
    By \cref{rem:etamax-eigenstate}, we can assume that the maximum of $\eta$ is assumed for an eigenstate $\ket g$ of $K_N$.
    This means $\eta_{\max}(N,\varphi,\theta,t) = \eta(N,\varphi, \theta, t, \ket g)$ and thus
\begin{align*}
    \frac{\partial\eta_{\max}}{\partial\theta} = - \sum_{x \le 0}\sum_{g'} |\alpha_x(N,g')|^2 \overbrace{|\kappa_{g'}(N)|^2}^{=\delta_{gg'}} = -\sum_{x \le 0} |\alpha_x(N,g)|^2 \le -1 + \frac{1}{2^{t/2}}.
\end{align*}
In the last step we have applied \cref{Lemma:String_Encoding}.
\end{proof}

The above lemma gives a bound on the derivative of $\eta$ with respect to $\theta$.
However, we are actually interested in the case where the QTM is limited to a universal gate set and $U_\varphi$ and $U_N$.
Thus we need to prove the value of $\eta$ is still strictly decreasing even when the Solovay-Kitaev algorithm is used to approximate the relevant gates.
\begin{corollary}[$\eta''$ Monotonous]\label{Corollary:eta_2_monotonicity}
\checked{James}
For $t\ge|N|$, the quantity $\max_{\ket{\nu}}\eta''(N,\varphi, \theta, t, \ket{\nu})$ is strictly monotonically decreasing.
\end{corollary}
\begin{proof}
The proof is identical to \cref{lem:eta-solovay}, except with a bound on the gradient satisfying $\partial\eta/\partial\theta \le -1/2$ (i.e.\ a constant) from \cref{Lemma:eta_2_Decreasing}.
\end{proof}

We now have bounds on the value of $\eta''$ and its gradient with respect to $\theta$, as per \cref{Lemma:SK_HamSim_Approximation_2} and \cref{Corollary:eta_2_monotonicity}.
We can now encode the QTM and its associated computation in a standard-form Hamiltonian in the same way as the 1-CRT-PRM case.
In particular, we notice that all the lemmas/theorems after \cref{Theorem:QTM_in_local_Hamiltonian} in the 1-CRT-PRM proof depend only on the value of $\eta''$ and bounds on its derivatives.

\begin{theorem}[Phase Comparator Hamiltonian]\label{Theorem:QTM_in_local_Hamiltonian_2_PRM}
\checked{James}

    Let $N\in\field N$, and $\varphi\in[0,1]$.
    For any Hamiltonian $K_N$
    there exists a constant $d>0$, and Hermitian operators $h^{(1)}\in\mathcal B(\field C^d)$, $h^{(2)}\in\mathcal B(\field C^d\times\field C^d)$, such that
    \begin{enumerate}
        \item $h^{(1)},h^{(2)}\ge 0$, with matrix entries in $\field{Z}$.
        \item $h^{(2)}=A+\ee^{\ii\pi\varphi} B + \ee^{-\ii\pi\varphi} B^\dagger +  \ee^{\ii\pi 0.\enc(N) }C + \ee^{-\ii\pi 0.\enc(N)}C^\dagger + \theta(D+D^\dagger)$, where
        	\begin{itemize}
				\item $B,C,D \in \mathcal{B}(\C^d)$ with coefficients in $\field{Z}$, and
				\item $ A\in \mathcal{B}(\C^d)$ is Hermitian and with coefficients in $\field{Z}+\field{Z}/\sqrt 2+\ee^{\ii\pi/4}\field Z$.
			\end{itemize}
    \end{enumerate}
   Define a translationally-invariant nearest-neighbour Hamiltonian on a spin chain of length $L$ via
    \[
        \HTM(L)  \coloneqq  \sum_{i=1}^L h^{(1)}_i + \sum_{i=1}^{L-1} h^{(2)}_{i,i+1}.
    \]
    Denote with $\ket*{\blacksquare}$ and $\ket*{\midend}$ two special basis states of $\field C^d$, and for $m\in\field N$, denote the \emph{bracketed} subspace
    \[
    \Sbr(m)  \coloneqq  \ket*{\midend} \ox (\field C^d)^{\ox m} \ox \ket*{\blacksquare} \ox (\field C^d)^{\ox (L-m)} \ket*{\midend}.
    \]
    Then $\HTM(L)$ has the following properties.
    \begin{enumerate}\addtocounter{enumi}{2}
        \item $\HTM(L) = \bigoplus_{m=1}^{L-1} H(L,m) \oplus R$, where $\HTM(L,m)  \coloneqq  \HTM(L)|_{\Sbr(m)}$; i.e.\ $H(L)$ is block-diagonal with respect to the subspaces spanned by $\Sbr(m)$, and $R$ captures the remaining block.
        \item $R \ge 1$.
        \item $\lmin(H(L,m)) \ge 1$ if $m=0,1$.
        \item
        There exist $L_N = \poly N$ and $m_N = \poly\log_2 N$ and an integer constant $b$, such if $\varphi\in [\lmin(K_N)+\delta/3,\lmin(K_N)+2\delta/3]$, the ground state energy of the other blocks satisfies
        {\small
        \[
            \lmin(H(L,m))  \begin{cases}
                \leq \frac{1.05}{256L^b}\frac{1}{2} & (m,L)=(m_N,L_N) \land \theta \le \lambda_{j*}(B)-\frac{1}{2}-1/P_2(N) \\
                \geq \frac{0.99}{256L^b}\frac{3}{5} & (m,L)=(m_N,L_N) \land \theta \ge \lambda_{j*}(B)-\frac{2}{5}+1/P_2(N) \\
                \geq \frac{0.99}{256L^b}\frac{3}{5} & m < |N| \lor (m,L) \neq (m_N, L_N),
            \end{cases}
        \]}
        where $T(L)=L^{b/2}$ is the runtime of the encoded computation. \label{Point:Energy_Bounds_2}
        \item  $\lmin(H(L,m))$ is strictly monotonically increasing with $\theta$ for \label{Point:Monotonicity_2} 
        \[
        \theta\in \big[\lambda_{j^*}(B) - 1/2+ 1/P_2(N)\,, \lambda_{j^*}(B)-2/5-1/P_2(N)\big].
        \]
        
    \end{enumerate}
\end{theorem}
\begin{proof}
The proof is almost identical to \cref{Theorem:QTM_in_local_Hamiltonian}, except for points \ref{Point:Energy_Bounds_2} and \ref{Point:Monotonicity_2}.
Point \ref{Point:Energy_Bounds_2} follows by noting that the new penalty is 
\begin{align*}
    \Pi &= \frac{1}{256L^b}(\1 -\ketbra{11}_f\ox (B - \theta \1)\ox \1), \\
    \implies Tr(\ketbra{\chi}\Pi) &= \frac{1}{256L^b}(1 - \eta(N,\varphi, \theta, t,\ket{\nu})) \\
    \implies \min_{\ket{\nu}}Tr(\ketbra{\chi}\Pi) &= \frac{1}{256L^b}(1 - \max_{\ket{\nu}}\eta(N,\varphi, \theta, t,\ket{\nu})).
\end{align*}
Combining \cref{cor:Rejection_Probabilities_2} with the bounds in \cref{Theorem:Precise_Energies} gives the energy bounds in the theorem statement.

Point \ref{Point:Monotonicity_2} follows directly from \cref{Corollary:eta_2_monotonicity}, or from the fact that $\Pi = \frac{1}{256L^b}(\1 -\ketbra{11}_f\ox (B - \theta \1)\ox \1)$, and hence the eigenvalues of the projector are non decreasing as $\theta$ increases.
Since for $\theta$ in the range given, the projector is always positive semi-definite, the fact that the ground state energy is strictly monotonically increasing follows (cf.~\cref{Corollary:eta_2_monotonicity}).

As in the proof of \cref{Theorem:QTM_in_local_Hamiltonian}, we can add another ancilla to remove the explicit $L$-dependence of $\Pi$ (as per \cref{Eq:Qubit_Rotation}, and the method of equalising the bounds on the second and third line for $\lmin(H(L,m))$ is identical as well.
\end{proof}

We now prove a lemma analogous to \cref{Lemma:Single_Square_Energy} showing that we can combine the QTM encoding the computation with a negative energy ``marker Hamiltonian'' such that the total energy is positive if $\theta$ is sufficiently larger than $\lmin_{j^*}(B)$, and negative if it is sufficiently smaller.
\begin{lemma} \label{Lemma:Single_Square_Energy_2}
\checked{James}
Let $H \coloneqq \HTM\otimes\1 + \1\otimes\HM$ on a spin lattice. Then its ground state is a product state $\ket{\psi}\ox \ket{T}_c$, where $\ket T_c$ is the checkerboard tiling from $\HM$, and $\ket\psi$ the ground state of $\HTM$.
Consider an $L\times L$ square denoted $S(L)$ within the tiling and let $H|_{S(L)}$ be the Hamiltonian restricted to such a square.
Then, adopting the notation from \cref{eq:bds-22}, and assuming $|\bra{\psi}B\ket{\psi} - \lambda_{j^*}(B)|=\BigO(1/P_1(N))$ for all states $\ket{\psi}\in S_\delta$, and letting $\varphi\in [\lmin(K_N) + \delta/3,\lmin(K_N) + 2\delta/3]$, we have that
\[
    \lmin(H|_{S(L)}) \begin{cases}
        < 0  &  \text{if\ } (L,m) = (L_N, m_N) \land \theta \le \lambda_{j*}(B)-\frac{1}{2}-1/P_2(N)  \\
        \ge 0 & \text{if\ }(L,m) = (L_N, m_N) \land \theta \geq \lambda_{j*}(B)-\frac{2}{5}+1/P_2(N) \\
        \ge 0 & \text{if\ } m < |N| \lor (L,m) \neq (L_N, m_N).
    \end{cases}
\]
Furthermore, if $L=L_N$ and  $m = m_N$, then there is exactly one point $\theta^*$ where $\lmin(H|_{S(L)})$ changes from $<0$ to $>0$ which occurs in the interval
\[
    \theta^*\in \big[\lambda_{j*}(B)-\frac{1}{2}-1/P_2(N)\,, \lambda_{j*}(B)-\frac{2}{5}+1/P_2(N)\big].
\]
\end{lemma}
\begin{proof}
We choose the Marker falloff $f(L)$ such that
\[
    \frac{9}{4}4^{-f(L)} = \frac{9}{16}\frac{1}{256L^b} 
    \quad\Longrightarrow\quad
    f(L) = 5 + \log_4(L^b).
\]
Note $f(L)$ is trivially computable in time and space $kL$ for some constant $k$.

We now compare the energy of the marker Hamiltonian (given by \cref{Theorem:Marker_Energy}) with this $f(L)$  with the energy of the Hamiltonian encoding the QTM (given in \cref{Theorem:QTM_in_local_Hamiltonian_2_PRM}).
For all $L\ge 7$, one can show by explicit calculation that
\[
  \frac{1.05}{256L^b}\frac{1}{2} \leq 4^{-f(L)}\left(\frac{9}{4}- \frac{10}{2^L}\right)  \quad \text{and} \quad  \frac{9}{4}4^{-f(L)}   \leq \frac{0.99}{256L^b}\frac{3}{5}
\]
where $b$ is the runtime exponent of $T=T(L)=L^{b/2}$, as given in \cref{Theorem:QTM_in_local_Hamiltonian_2_PRM}.

As the spectrum is product by construction, the joint spectrum is then $\spec(H) = \spec(\HTM) + \spec(\HM)$, and the rest follows from \cite[Lem.~F.1]{Bausch_Cubitt_Watson2019}.

As per point \ref{Point:Monotonicity_2} of \cref{Theorem:QTM_in_local_Hamiltonian_2_PRM}), the ground state energy $\lmin(\HTM(L))$ is strictly monotonically decreasing for $\theta\in \big[\lambda_{j^*}(B) - 1/2+ 1/P_2(N)\,, \lambda_{j^*}(B)-2/5-1/P_2(N)\big]$. 
There is thus a unique point at which $|\lmin(\HTM(L))| = |\lmin(\HM(L)|_S)|$, and it occurs for energy values corresponding to $\theta$ in the given interval.
\end{proof}

As in the one-parameter case, the final step is then to apply \cref{Lemma:Finite_Size_Energy} to the Hamiltonian $H$ from \cref{Lemma:Single_Square_Energy_2}; the resulting Hamiltonian $H'$ is then combined with a trivial, a dense, and a guard Hamiltonian, to lift the ground state energy to a ground state energy density statement. 
As aforementioned, this modifies the Hamiltonian so that phase transitions can occur between the ground state of the checkerboard Hamiltonian and the ground state of a trivial zero energy state.
Analogous to \cref{Lemma:Two_Phases}, we obtain the following central result.
\begin{theorem} \label{Theorem:2-PRM-Crit_GS_Difference}
\checked{James}
Let $K_N\in (\C^d)^{\ox N}$ be a the Hamiltonian from \cite{Watson_Bausch_Gharibian_2020} described in \cref{Lemma:TI-APX-SIM_Completeness} such that $\delta = \Omega(N^{-D})$ for some constant $D$.
Define the order parameter $O_{A/B}$ acting on a const-sized subset of the lattice as in \cref{Lemma:Two_Phases}.
We can explicitly construct a Hamiltonian $H^\Lambda(N,\varphi, \theta)=\sum_{\langle i,j \rangle}h^N_{i,j}(\varphi, \theta) + \sum_{i\in \Lambda}h^{N}_i$ such that, in the infinite lattice size limit the following conditions hold.
For any $\varphi\in [\lmin(K_N)+\delta/3,\lmin(K_N)+2\delta/3]$ and supposing for the $\lambda_{j^*}(B)\in \{1,2\}$ which satisfies $|\bra{\psi_0}B\ket{\psi_0} - \lambda_{j^*}(B)|=\BigO(1/P_1(N))$, where $\ket{\psi_0}$ is the ground state of $K_N$, then:
\begin{itemize}
    \item if $\theta \ge \lambda_{j*}(B)-\frac{2}{5}+1/P_2(N)$:
    \begin{enumerate}[i]
        \item $H^\Lambda$ is gapped with spectral gap 1.
        \item product ground state.
        \item has order parameter expectation value $\langle O_{A/B}\rangle = 1$.
    \end{enumerate}
    \item if $\theta \le \lambda_{j*}(B)-\frac{1}{2}-1/P_2(N)$:
    \begin{enumerate}[i]
        \item $H^\Lambda$ is gapless.
        \item has a ground state with algebraically decaying correlations.
        \item has order parameter expectation value $\langle O_{A/B}\rangle = 0$.
    \end{enumerate}
\end{itemize}

Furthermore, for any $\varphi$ in the given interval, this Hamiltonian has exactly one critical point in terms of $\theta$, which we denote $\theta^*$. 
This occurs in the interval
\[
    \theta^* \in \big[\lambda_{j*}(B)-\frac{1}{2}-1/P_2(N)\,, \lambda_{j*}(B)-\frac{2}{5}+1/P_2(N)\big].
\]
\end{theorem}
\begin{proof}
The Hamiltonian which is constructed is done in the same way as \cref{sec:phase-comparator-ham}---we create a Hamiltonian with an energy shifted up by 1, as per \cref{Lemma:Finite_Size_Energy} and construct a new Hamiltonian using the Hamiltonians $H_\mathrm{dense}$, $H_\mathrm{guard}$, and $H_\mathrm{trivial}$---except now we use the $\theta$ bounds from \cref{Lemma:Single_Square_Energy_2} (an equivalent to \cref{Lemma:Two_Phases} can trivially be proven using the new bounds).
\end{proof}

\subsection{Reduction of $\forall$-TI-APX-SIM to 2-CRT-PRM} \label{Sec:2-hardness-reduction}

\begin{figure}[!tb]
	\centering
	\hspace{-1cm}\includegraphics[height=9cm]{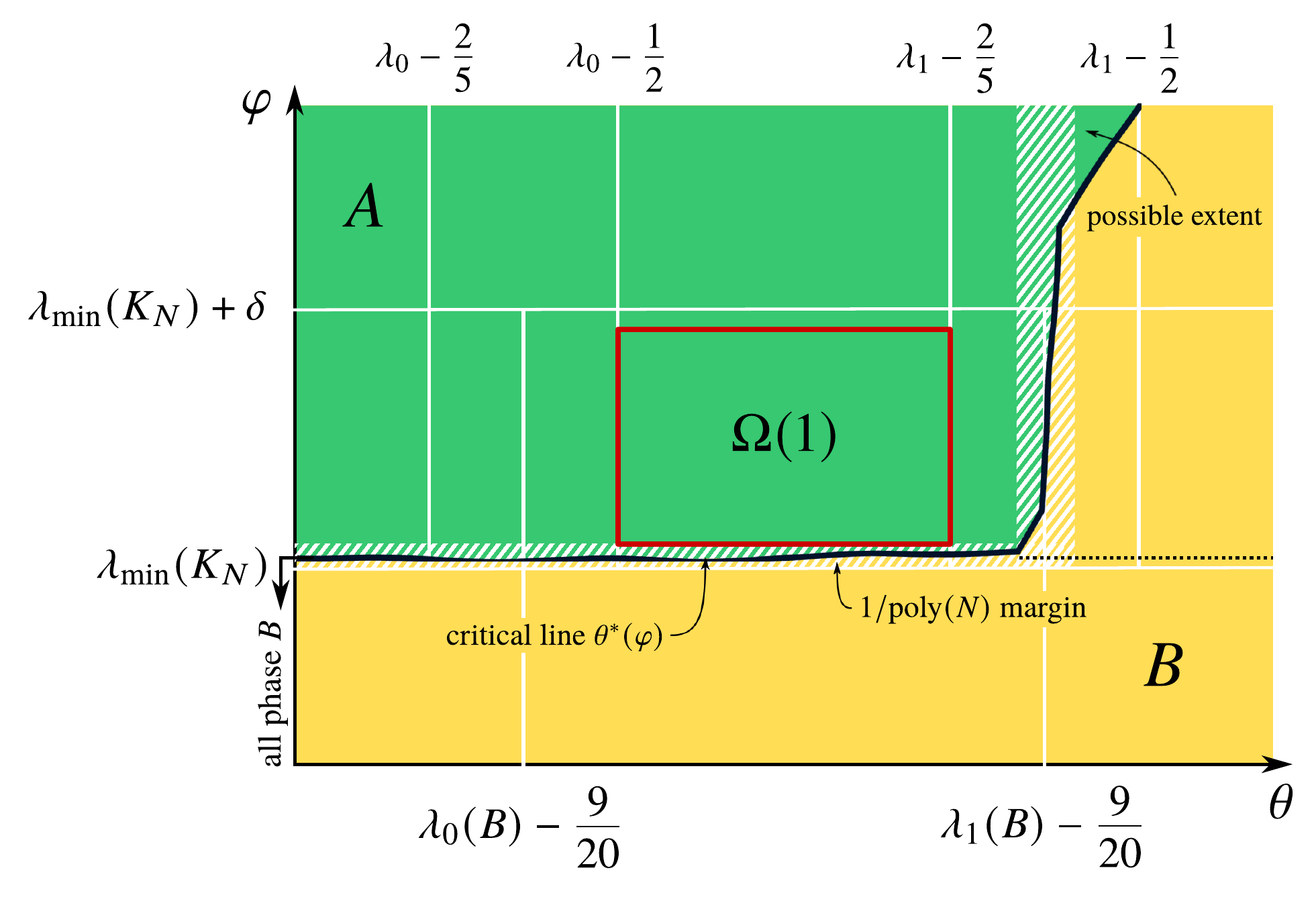}
	\caption{The YES case two-parameter phase diagram.
	The shaded white area shows an uncertainty region (of size $1/\poly N$); its \emph{inner} extent indicates the minimal area circumscribed by the critical line $\theta^*(\varphi)$; as shown in \cref{sec:unique-critical-point-2} we know that the true critical line has precisely one critical point whenever $\varphi$ is fixed and $\theta$ is varied, as well as vice versa; i.e., the critical line $\theta^*(\varphi)$ is a function, and monotonous, within an $\Omega(1)$ area of the phase space.
	It encompasses an $\Omega(1)$ area (given $\varphi$ is scaled such that effecively $\delta=\Omega(1)$, as explained in \cref{Corollary:2-PRM_Rescale}) of the phase space for which the system is guaranteed to be completely in phase $A$ in this case. 
	The location of the rectangle is efficiently computable relative to the point along the $\varphi$ axis below which the system is completely in phase $B$, irrespective of $\theta$.
	The NO case phase diagram is shown in \cref{Fig:2Param_Phase_Diagram-NO}.}
	\label{Fig:2Param_Phase_Diagram-YES}
\end{figure}

\begin{figure}[!tb]
	\centering
	\hspace{-1cm}\includegraphics[height=9cm]{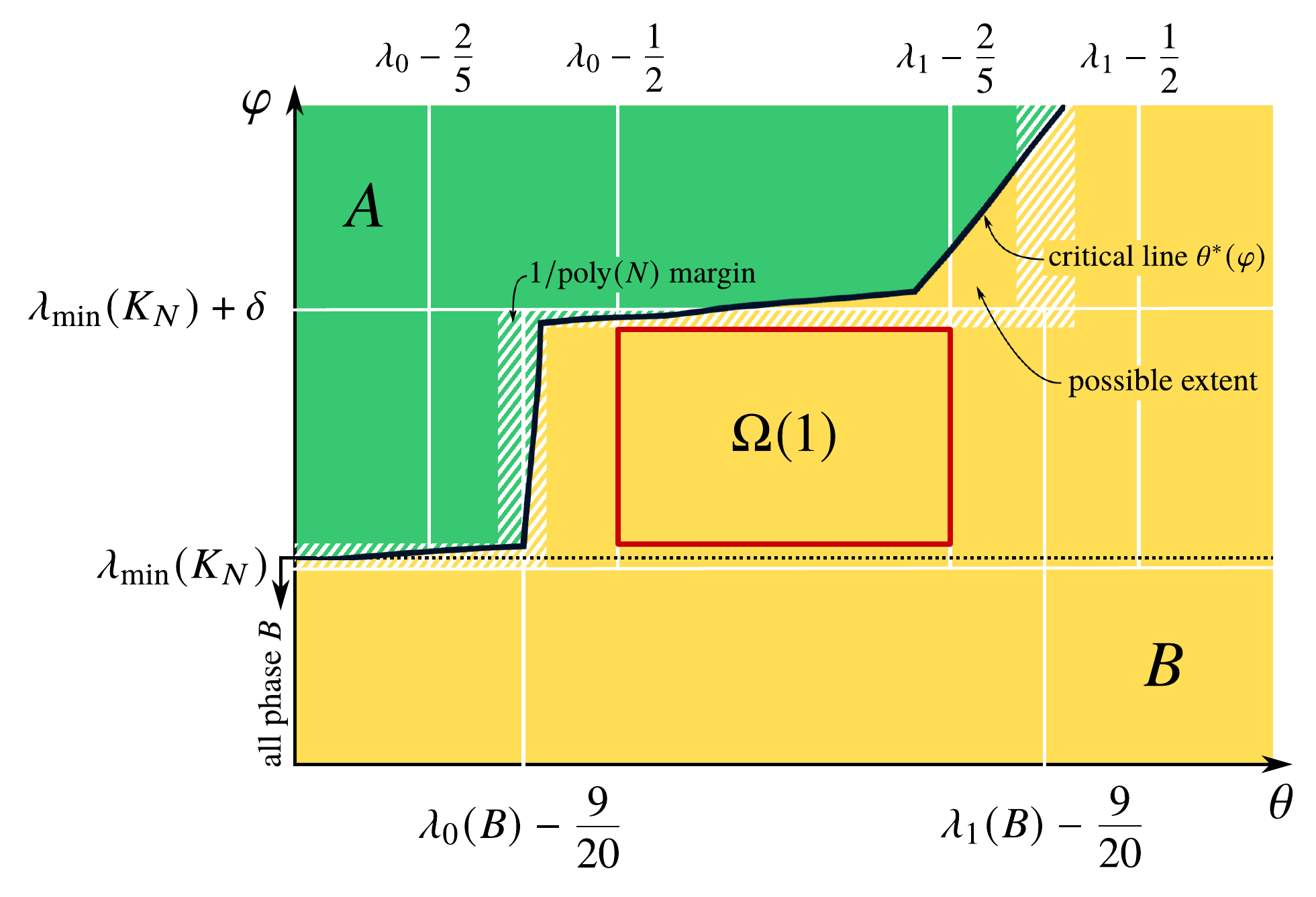}
	\caption{The NO case two-parameter phase diagram.
	The shaded white area shows an uncertainty region (of size $1/\poly N$); its \emph{outer} extent indicates the maximal area circumscribed by the critical line $\theta^*(\varphi)$
	It encompasses an $\Omega(1)$ area of the phase space for which the system is guaranteed to be completely in phase $B$ in this case, and its location is efficiently computable relative to the point along the $\varphi$ axis below which the system is completely in phase $B$, irrespective of $\theta$.
	The YES case phase diagram is shown in \cref{Fig:2Param_Phase_Diagram-YES}.}
	\label{Fig:2Param_Phase_Diagram-NO}
\end{figure}

As per the 1-CRT-PRM case, we will find it useful to rescale the QPE process implemented $\varphi$ as per \cref{rem:scaling-up-varphi}. 
This has the effect of mapping $\varphi \rightarrow N^{-D}\varphi$.
This allow us to write the following corollary:
\begin{corollary}[$\varphi$-Rescaled Hamiltonian]\label{Corollary:2-PRM_Rescale}
\checked{James}
\Cref{Theorem:2-PRM-Crit_GS_Difference} holds for a modified Hamiltonian if $\varphi\in [N^D(\lmin(K_N)+\delta/3),N^D(\lmin(K_N)+2\delta/3)]$ such that $N^D\delta = \Omega(1)$.
\end{corollary}

Having introduced this rescaling of $\varphi$, we now show that determining the phase transition point $\theta^*(\varphi)$ is $\PQMAEXP$-hard by showing that determining $\theta^*(\varphi)$ for a specific $\BigO(1)$ interval of $\varphi$ is gives the answer to a $\forall$-TI-APX-SIM instance.
\begin{theorem}[\PQMAEXP-hardness] \label{Lemma:2-CRT-PRM_QMA-hardness}
\checked{James}
There is a polynomial time Turing reduction from $\forall$-TI-APX-SIM to 2-CRT-PRM, and hence 2-CRT-PRM is \PQMAEXP-hard.
\end{theorem}
\begin{proof}
We refer the reader to \cref{Fig:2Param_Phase_Diagram-YES} and \cref{Fig:2Param_Phase_Diagram-NO} to aid this proof.

From \cite{Watson_Bausch_Gharibian_2020} it is known that determining whether $\bra{\psi}B\ket{\psi}>\beta$ or $\bra{\psi}B\ket{\psi}<\alpha$ for states $\ket{\psi}\in S_\delta$ is \PQMAEXP-hard.
From \cref{Theorem:2-PRM-Crit_GS_Difference} we know that for all $\varphi\in [N^D(\lmin(K_N)+\delta/3),N^D(\lmin(K_N)+2\delta/3)]$ it holds that the critical point $\theta^*(\varphi)$ is determined by whether $\bra{\psi}B\ket{\psi}>\beta$ or $\bra{\psi}B\ket{\psi}<\alpha$ for states $\ket{\psi}\in S_\delta$ and $\beta-\alpha = \Omega(1)$.
\end{proof}

\begin{corollary}
2-CRT-PRM is \PQMAEXP-complete.
\end{corollary}
\begin{proof}
Hardness and containment follow from \cref{Lemma:2-CRT-PRM_QMA-hardness,lem:tech-containment-3}, respectively.
\end{proof}

The two phase diagrams in the YES and NO cases are shown in \cref{Fig:2Param_Phase_Diagram-YES,Fig:2Param_Phase_Diagram-NO}.

\subsection{Verifying the Local-Global Promise}

\begin{lemma} \label{Lemma:Promise_Satisfied_Phase_2}
\checked{James}
Consider an instance of the Hamiltonian
\[
H^{\Lambda(L)}(N,\theta, \varphi)=\sum_{\langle i,j \rangle}h^N_{i,j}(\theta, \varphi) + \sum_{i\in \Lambda}h^{N}_i
\]
as defined in \cref{Corollary:2-PRM_Rescale} and \cref{Theorem:2-PRM-Crit_GS_Difference}, with local terms describable in $|N|$ bits.
Then the Hamiltonian satisfies the global-local phase assumption (\cref{Def:Local-Global_Phase}) for the order parameter $O_{A/B}$ given in \cref{Theorem:2-PRM-Crit_GS_Difference}, and for $L_0=N^{2+a+b}$.
It also satisfies the global-local gap promise in \cref{Def:Local-Global_Gap} for the same $L_0$.
\end{lemma}
\begin{proof}
Same as \cref{Lemma:Promise_Satisfied_Phase}.
\end{proof}

We have already shown that $\theta^*(\varphi)$ is unique, see \cref{Theorem:2-PRM-Crit_GS_Difference}; this shows that for a fixed $\varphi$, $H(\theta, \varphi)$ is has a unique critical point.
Finally, we can show that if we restrict ourselves to the set of parameters $(\theta=\mathrm{const}, \varphi)$ (i.e.\ where $\theta$ is held constant, but $\varphi$ is varied), then there is a unique phase transition point $\varphi^*=\varphi^*(\theta) \in [0, \poly N]$. 

\begin{lemma}\label{lem:2-crt-satisfy-1-crt-props}
$H(\theta, \varphi)|_{\theta=0}$ has precisely one critical point $\varphi^*$ delineating the two phases $A$ and $B$.
\end{lemma}
\begin{proof}
It is straightforward to check that the resulting Hamiltonian satisfies the preconditions of \cref{Theorem:Crit_GS_Difference}, but where now
\[
    \varphi^* \in \big[ \lmin(K_N) - N^{-4C} + \BigO(N^{-6C}, \lmin(K_N) - \BigO(N^{-6C} \big].
\]
The second claim then follows by monotonicity of $\eta$ in $\varphi$, \cref{lem:eta-monotonous}.
\end{proof}

By this result (which is slightly stronger than the necessary promise of \cref{def:2-crt-prm}, which only demands \cref{lem:2-crt-satisfy-1-crt-props} to hold for $\theta=0$) we can now use containment of 1-CRT-PRM in \PQMAEXP to answer the location of $y=\varphi^*$, i.e.\ the point in the 2D phase diagram below which the system is completely in phase $B$, as explained in \cref{sec:containment}.

\section{Discussion and Conclusion}\label{sec:discussion}

\paragraph{Completeness for Other Families of Hamiltonians.}
In the hardness results above we proved hardness for a family of Hamiltonians with a phase transition between a ``highly classical'' phase with product state eigenstates, zero correlations and with $\langle O_{A/B}\rangle=1$, and a separate ``highly quantum'' phase with algebraically decaying correlations and $\langle O_{A/B}\rangle=0$.
However, we can generalise our construction to other phases fairly trivially in a way that we outline here
\begin{remark}
Let $h_X,h_{\neg X}\in \mathcal{B}(\C^d\ox\C^d)$ be an interaction terms between neighbouring qudits such that for all lattices sizes $L\geq L_0$ the ground state Hamiltonians
\begin{align*}
    H_X^{\Lambda(L)} = \sum_{i\in \Lambda(L)}h_{X(i, i+1)} 
    \quad\text{and}\quad H_{\neg X}^{\Lambda(L)} = \sum_{i\in \Lambda(L)}h_{\neg X(i, i+1)}
\end{align*}
have and do not have some property $X$, respectively, and both have zero energy ground state $\lmin(H_X^{\Lambda(L)})=\lmin(H_{\neg X}^{\Lambda(L)})=0$ .
Further assume that the property $X$ can be efficiently distinguished (e.g.\ by a local order parameter) and that the Hamiltonian constructed in \cref{Lemma:Single_Square_Energy_2} does not have property $X$.
Then we can explicitly construct a Hamiltonian with a phase transition between two phases, one with property X and another without property X, such that 2-CRT-PRM is \PQMAEXP-complete.
\end{remark}
\begin{proof}
The proof is a simple modification of Hamiltonian used to prove \cref{Theorem:Crit_GS_Difference} and \cref{Theorem:2-PRM-Crit_GS_Difference}.
However, instead of using $H_{dense}$ and $H_{trivial}$, we replace them with $H_X^{\Lambda(L)}$ and $H_{\neg X}^{\Lambda(L)}$.
As before, define
\[
    \Hguard \coloneqq \sum_{i\sim j}\left( \1_{1,2}^{(i)}\otimes \1_3^{(j)} + \1_3^{(i)} \otimes \1_{1,2}^{(j)} \right).
\]
Let
\[
    H^\Lambda(N,\varphi, \theta) \coloneqq H'(\varphi, \theta) \otimes \1_2 \oplus 0_3 + \1_1 \otimes H_{\neg X} \oplus 0_3 + 0_{1,2} \oplus H_X  + \Hguard.
\]

Since $\lmin(H_X^{\Lambda(L)})=\lmin(H_{\neg X}^{\Lambda(L)})=0$, the phase transitions still occur in the same place as \cref{Theorem:2-PRM-Crit_GS_Difference} (i.e. at the point where $\lmin(H(\varphi, \theta))=0)$.
Furthermore, due to the properties of $H_X^{\Lambda(L)}$ and $H_{\neg X}^{\Lambda(L)}$, these phase can be distinguished by the property $X$.
Thus distinguishing where the phase changes for $\varphi\in [N^D(\lmin(K_N)+\delta/3),N^D(\lmin(K_N)+2\delta/3)]$ must be \PQMAEXP-hard by the same proof as \cref{Lemma:2-CRT-PRM_QMA-hardness}.
\end{proof}

A major constraint on this theorem is that $H(\varphi, \theta)$ must not have property $X$.
In particular, this means our construction means we cannot create a phase transition between two gapped phases, two frustration free phases, etc.
Cases such as these must be considered separately. 



\paragraph{Comparison to Undecidability and Uncomputability Results.}
We also take care to distinguish our results from the size driven quantum phase transitions \cite{Bausch_19_Size_Driven}.
Here we are promised that in the thermodynamic we are always in a particular phase, but that the transition takes place at some uncomputably large lattice size. 
Our result differs significantly in that the spectral gap and phase are explicitly computable for some finite size lattice, and the system can be gapped or gapless in the thermodynamic limit.

We also emphasise the differences to the previous undecidability results \cite{Cubitt_Perez-Garcia_Wolf2015, Bausch_2020_Undecidability, Bausch_Cubitt_Watson2019}.
There are two key differences here: the promise of the global-local gap/phase means that the gap/phase are computable. 
Furthermore, for undecidability or uncomputability to occur, it is often remarked that there must be an ``infinity'' to encode the different computational outputs.
The result is that in the previous works there must be an infinite number of phase transitions\footnote{For various technical reasons it is not possible to define phases for the systems studied in \cite{Cubitt_Perez-Garcia_Wolf2015, Bausch_2020_Undecidability}. Instead the authors remark that there are an uncountably infinite number of points where the system changes from being gapped to gapless.}.
The systems in our work contain either only a single phase transition, or a finite number, which arguably better reflects the systems we see in nature.

\paragraph{Containment of the 1-CRT-PRM Case.}
In the 1-CRT-PRM result, we prove that determining the gap/phase 1-CRT-PRM for an $\Omega(1)$ region of the $\varphi\in [0,1]$ parameter space is \QMAEXP-hard.
It is natural to ask whether we can prove containment in \QMAEXP or at least $\mathrm{P}^{\QMAEXP[\mathrm{const}]}$?
In particular, the fact that for the Hamiltonian constructed, the spectral gap is actually either $\geq 1/2$ or $\leq 1/\poly N$ suggests we might be able to distinguish the two cases by estimating the spectral gap to only constant precision in \cref{sec:containment} rather than the $1/\poly N$ precision we currently perform the algorithm to.
However, both of the local-global algorithms (used for determining the spectral gap or the order parameter at a given point, respectively) require knowledge about the ground state to the relevant precision.
In case we promise that the Hamiltonian's ground state energy can be resolved within a constant number of bits, containment in aforementioned stricter classes follows.
Naturally, this leaves open the question whether an algorithm exists that can answer the spectral gap or order parameter problems to constant precision without knowing the ground state energy to the same precision.

\paragraph{Precise Variant.}
As mentioned below the definition of 1-CRT-PRM, \cref{def:1-crt-prm}, there is a natural ``precise'' variant, Precise-1-CRT-PRM, where we want to approximate the critical point to exponential precision.
As explained in \cite[Th.~4.1]{Kohler2020}, for exponential precision one can allow the embedded computation to run for time $\exp \poly(L)$ in the size of the spin chain segment $L$; as such, one can extract exponentially many bits of the parameter $N$, and the distinction between translationally invariant and non-translationally-invariant models vanish.
In this case, as detailed in \cite[Cor.~29\&31]{Watson_Bausch_Gharibian_2020}, the APX-SIM variant is simply PSPACE-complete, by simulating a PSPACE computation within the history state.
As such, it follows that the Precise-1-CRT-PRM problem---and by a similar argument also Precise-2-CRT-PRM---are PSPACE-hard.
Containment in PSPACE, for a suitable definition of a local-global gap (that is now allowed to shrink exponentially in the system size),
follows from a precise variant of \citeauthor{Ambainis2013}'s algorithm to determine spectral gap; and because $\mathrm{P}^\mathrm{PreciseQMA} = \mathrm{P}^\mathrm{PSPACE} = \mathrm{PSPACE}$.
\cite{Deshpande_Gorshkov_Fefferman_2020}.

\paragraph{Open Questions.}
The following points are natural continuations of this line of work.
\begin{enumerate}
\item One major open question is to pin down the exact hardness of the 1-CRT-PRM problem, as we could show \QMAEXP hardness, but only an upper bound of \PQMAEXP, mostly because it is not obvious a priori how to reduce the algorithm that decides the Local-Global promise to constant precision with only constantly-many queries without negating them (as one would, otherwise, require co-\QMAEXP queries). In all likelihood, a problem variant as formulated in 1-CRT-PRM with the promise parameters $\alpha$ and $\beta$ provided as input is indeed contained in \QMAEXP, but that a more physically-motivated variant that just asks about the approximation to some precision (i.e., the non-decision variant; or a variant that just asks for the last bit of a poly-precision approximation to be 0 or 1) to be \PQMAEXP-hard.
\item The Knabe and Martingale methods for determining the spectral gap apply to frustration free Hamiltonians, but the Hamiltonians used to prove our results here are not.
Can we prove a similar result to ours for frustration free Hamiltonians, for the classes QMA$_1$EXP or similar?   
The class QMA$_1$ naturally characterises the Local Hamiltonian problem in where YES case correspond to frustration free Hamiltonians \cite{Bravyi_2006, Bravyi_Terhal_2010}.
\item In addition to proving a variant of this result for frustration free problems, it would be further interesting to see if there are frustration free constructions for the undecidability/uncomputability problems, or even hardness of the local Hamiltonian problem.
By previous results about the stability of spectral gaps \cite{Michalakis_Zwolak_2013}, these constructions may be stable to perturbations to the matrix elements.
\end{enumerate}
Last but not least, it is a natural question to ask whether the Hamiltonian constructed can be made more physically realistic.
While our couplings are translationally-invariant, the local spin dimension is not a variable we attempted to keep low; and even if it is finite, techniques such as those from \cite{Bausch2016,Bausch2017} or a reduction via perturbation-gadget-based universality results \cite{Piddock2020,Kohler2020} might be a viable way to bring the local dimension down considerably.

\section*{Acknowledgements}
We are grateful for early discussions with Ashley Montanaro, which prompted us to work on this idea, Benjamin Beri for a helpful conversation about finite-size criteria, and we gratefully acknowledge feedback on the final manuscript by Angelo Lucia, Emilio Onorati, and Sevag Gharibian.
J.\,D.\ W. is supported by
the EPSRC Centre for Doctoral Training in Delivering Quantum Technologies
(grant EP/L015242/1).
J.\,B.\ acknowledges support from the Draper's Research Fellowship at Pembroke College.

\printbibliography

\appendix

\end{document}